\documentclass[b5paper,11pt]{book}

\pdfoutput=1

\usepackage{amsmath}
\usepackage{amssymb}
\usepackage{amsfonts}
\usepackage{amsthm}
\usepackage{placeins}
\usepackage{times,epsfig}
\usepackage{amscd}
\usepackage{mathtools}
\usepackage{color,graphicx}
\usepackage[spanish,english]{babel}
\usepackage{rotating}

\DeclareMathAlphabet{\Ma}{U}{msa}{m}{n}
\DeclareMathAlphabet{\Mb}{U}{msb}{m}{n}

\DeclareMathAlphabet{\Meuf}{U}{euf}{m}{n}
\DeclareSymbolFont{ASMa}{U}{msa}{m}{n}
\DeclareSymbolFont{ASMb}{U}{msb}{m}{n}

\DeclareMathOperator{\ran}{ran}

\newcommand{\scalar}[2]{\langle#1\,,#2\rangle}
\newcommand{\scalarb}[2]{\langle#1\,,#2\rangle_{\partial \Omega}}
\newcommand{\pair}[2]{(#1\,,#2)}
\newcommand{\pairb}[2]{(#1\,,#2)_{\partial \Omega}}
\newcommand{\norm}[1]{\|#1\|}
\newcommand{\normm}[1]{|\negthinspace\|#1\|\negthinspace|}
\newcommand{\vol}{\mathrm{vol}}
\renewcommand{\H}{\mathcal{H}}

\newcommand{\D}{\mathcal{D}}
\renewcommand{\d}{\mathrm{d}}
\newcommand{\pO}{{\partial \Omega}}

\newcommand{\C}{\mathcal{C}}
\newcommand{\diff}{\mathrm{d}}

\newcommand{\1}{\mathbb{I}}
\newcommand{\bx}{\mathbf{x}}
\newcommand{\bt}{\boldsymbol\theta}

\newcommand{\comm}[2]{\bigl[#1,#2\bigr]}

\newtheorem{theorem}{Theorem}

\newtheorem{corollary}[theorem]{Corollary}
\newtheorem{proposition}[theorem]{Proposition}
\newtheorem{definition}[theorem]{Definition}
\newtheorem{lemma}[theorem]{Lemma}
\newtheorem{example}[theorem]{Example}
\newtheorem{remark}[theorem]{Remark}

\numberwithin{equation}{section}
\numberwithin{theorem}{section}

\usepackage[indentafter,nobottomtitles]{titlesec}
\usepackage{titletoc}
\usepackage{fancyhdr}
\newcommand{\clearemptydoublepage}{\newpage{\pagestyle{empty}\cleardoublepage}}
\newcommand{\numerogordopart}[1]{\definecolor{gris}{gray}{0.75}\fontfamily{phv}\selectfont{\fontsize{30mm}{30mm}\selectfont\color{gris}{#1}}}

\titleformat{\chapter}
	[display]
	{\rmfamily\bfseries\scshape\Large}
	{\filleft \numerogordopart{\thechapter}}
	{2ex}
	{\titlerule \vspace{2ex}\filright}
	[\vspace{2ex}\titlerule]

\titleformat{\section}{\vspace{.8ex} \normalfont\bfseries\upshape} {\bf\thesection.}{.5em}{}[\titlerule]
\titleformat{\subsection}{\vspace{.8ex} \normalfont\bfseries\upshape} {\bf\thesubsection.}{.5em}{}

\titlecontents{chapter}
	[0pt]
	{\addvspace{0.5pc}\scshape}
	{\contentsmargin{0pt}%
		\bfseries
		\makebox[0pt][r]{
		\thecontentslabel
		\enspace
		}
	}
	{\contentsmargin{0pt}%
		\bfseries
		\makebox[0pt][r]{
		\thecontentslabel
		\enspace
		}
	}
	{	\bf
		\hspace
		\fill
		\thecontentspage
		\hspace*{-1.4pc}
	}
	[\addvspace{0.5pc}]

\def\listofsymbols{
\begin{tabbing}
$A_U$~~~~~~~~~~~\=\parbox{4.3in}{Partial Cayley transform of the unitary operator $U$\dotfill \pageref{symb:AU}}\\
\addsymbol \C^\infty (\Omega): {Space of smooth functions defined over the manifold $\Omega$}{symb:smooth}
\addsymbol \C_c^\infty (\Omega): {Space of smooth functions with compact support in the interior of $\Omega$}{symb:smooth compact}
\addsymbol \mathcal{C}: {Cayley transformation}{symb:Cayley}
\addsymbol \D(T): {Dense domain of the operator $T$}{symb:domainT}
\addsymbol \D_U: {Domain of the quadratic form associated to the Laplace-Beltrami operator described by the unitary at the boundary $U$}{symb:DU}
\addsymbol \d: {Exterior differential}{symb:d}
\addsymbol \d\mu_\eta: {Riemannian volume form}{symb:volume}
\addsymbol \Delta_\eta: {Laplace-Beltrami operator of the Riemannian manifold $(\Omega,\eta)$}{symb:laplacian}
\addsymbol \Delta_0: {Laplace-Beltrami operator defined over the compactly supported smooth functions}{symb:Delta0}
\addsymbol \Delta_{\mathrm{min}}: {Minimal extension of $\Delta_0$}{symb:Deltamin}
\addsymbol \Delta_{\mathrm{max}}: {Maximal extension of $\Delta_0$\,. It coincides with the adjoint operator of $\Delta_0$\quad}{symb:Deltamax}
\addsymbol \Delta_U: {Self-adjoint extension of $\Delta_{\mathrm{min}}$ associated to the unitary operator $U$}{symb:DeltaU}
\addsymbol E_\lambda: {Spectral resolution of the identity}{symb:spectralres}
\addsymbol \gamma(\cdot): {Trace map}{symb:trace}
\addsymbol \H^k(\Omega): {Sobolev space of class $k$ over the manifold $\Omega$\,, $k\in\mathbb{R}$}{symb:sobolevomega}
\addsymbol \H^k_0(\Omega): {Completion of the smooth functions with compact support in the Sobolev norm of order $k$ over the manifold \smash{$\Omega$}}{symb:sobolevomega0}
\addsymbol \H^k(\pO): {Sobolev space of class $k$ over the boundary manifold of $\Omega$}{symb:sobolevpo}
\addsymbol I: {Canonical isometric bijection of the scale of Hilbert spaces $\H_+\subset\H\subset\H_-$}{symb:I}
\addsymbol \Lambda^1(\Omega): {Space of smooth one-forms of the manifold $\Omega$}{symb:oneforms}
\addsymbol \eta: {Riemannian metric}{symb:metric}
\addsymbol \partial\eta: {Riemannian metric induced at the boundary}{symb:partialeta}
\addsymbol \norm{\cdot}: {Norm of the Hilbert space $\H$}{symb:norm}
\addsymbol \norm{\cdot}_k: {Sobolev norm of order $k$}{symb:normk}
\addsymbol \norm{\cdot}_{\H^k(\Omega)}: {Sobolev norm of order $k$}{symb:normkomega}
\addsymbol \normm{\cdot}_T: {Graph norm of the operator $T$}{symb:graphnormT}
\addsymbol \normm{\cdot}_Q: {Graph norm of the lower semi-bounded quadratic form $Q$}{symb:graphnormQ}
\addsymbol \mathcal{N}_{\pm}: {Deficiency spaces}{symb:deficiencyspace}
\addsymbol n_{\pm}: {Deficiency indices}{symb:deficiencyindex}
\addsymbol \Omega: {Smooth, orientable, compact Riemannian manifold}{symb:omega}
\addsymbol \partial \Omega: {Boundary manifold of the manifold $\Omega$}{symb:po}
\addsymbol \tilde{\Omega}: {Smooth Riemannian manifold without boundary}{symb:omeganotboundary}
\addsymbol \smash{\overset{\scriptscriptstyle\circ}{\Omega}}: {Interior of the manifold $\Omega$}{symb:intOmega}
\addsymbol \varphi,\,\psi: {Restriction to the boundary of the functions denoted by the corresponding capital greek letters $\Phi,\,\Psi$}{symb:varphi}
\addsymbol \dot{\varphi},\,\dot{\psi}: {Restriction to the boundary of the normal derivative of the functions denoted by the corresponding capital greek letters  $\Phi,\,\Psi$}{symb:dotvarphi}
\addsymbol P: {Orthogonal projection onto $W$}{symb:P}
\addsymbol Q_U: {Quadratic form associated to the Laplace-Beltrami operator described by the unitary operator $U$}{symb:QU}
\addsymbol \scalar{\cdot}{\cdot}: {Scalar product of the Hilbert space $\H$}{symb:scalar}
\addsymbol \pair{\cdot}{\cdot}: {Pairing associated to the scale of Hilbert spaces $\H_+\subset\H\subset\H_-$}{symb:pair}
\addsymbol \scalarb{\cdot}{\cdot}: {Induced scalar product at the boundary}{symb:scalarb}
\addsymbol \sigma(T): {Spectrum of the operator $T$}{symb:spectrum}
\addsymbol V(g): {Unitary representation of the element $g\in G$ of the group $G$}{symb:Vg}
\addsymbol \mathrm{v}(g): {Induced unitary representation at the boundary of the element $g\in G$ of the group G}{symb:v(g)}
\addsymbol W: {Invertibility boundary space}{symb:W}
\addsymbol \mathfrak{X}(\Omega): {Space of smooth vector fields of the manifold $\Omega$}{symb:vectorfield}
\addsymbol \Xi: {Collar neighbourhood of the boundary manifold $\pO$}{symb:collar}
\end{tabbing}

 \clearemptydoublepage}
\def\addsymbol #1: #2#3{$#1$ \> \parbox{4.3in}{#2$\mathllap{\phantom{\Bigr]}}$ \dotfill \pageref{#3}}\\} 
\def\newnot#1{\label{#1}}

\begin{document}

%
%
%
%
%
%
%
%
%
%
%
%
%
%

\thispagestyle{empty}
\begin{center}

~
\mbox{\parbox{3cm} {\centering\scriptsize\mbox{\includegraphics[width=3.2cm] {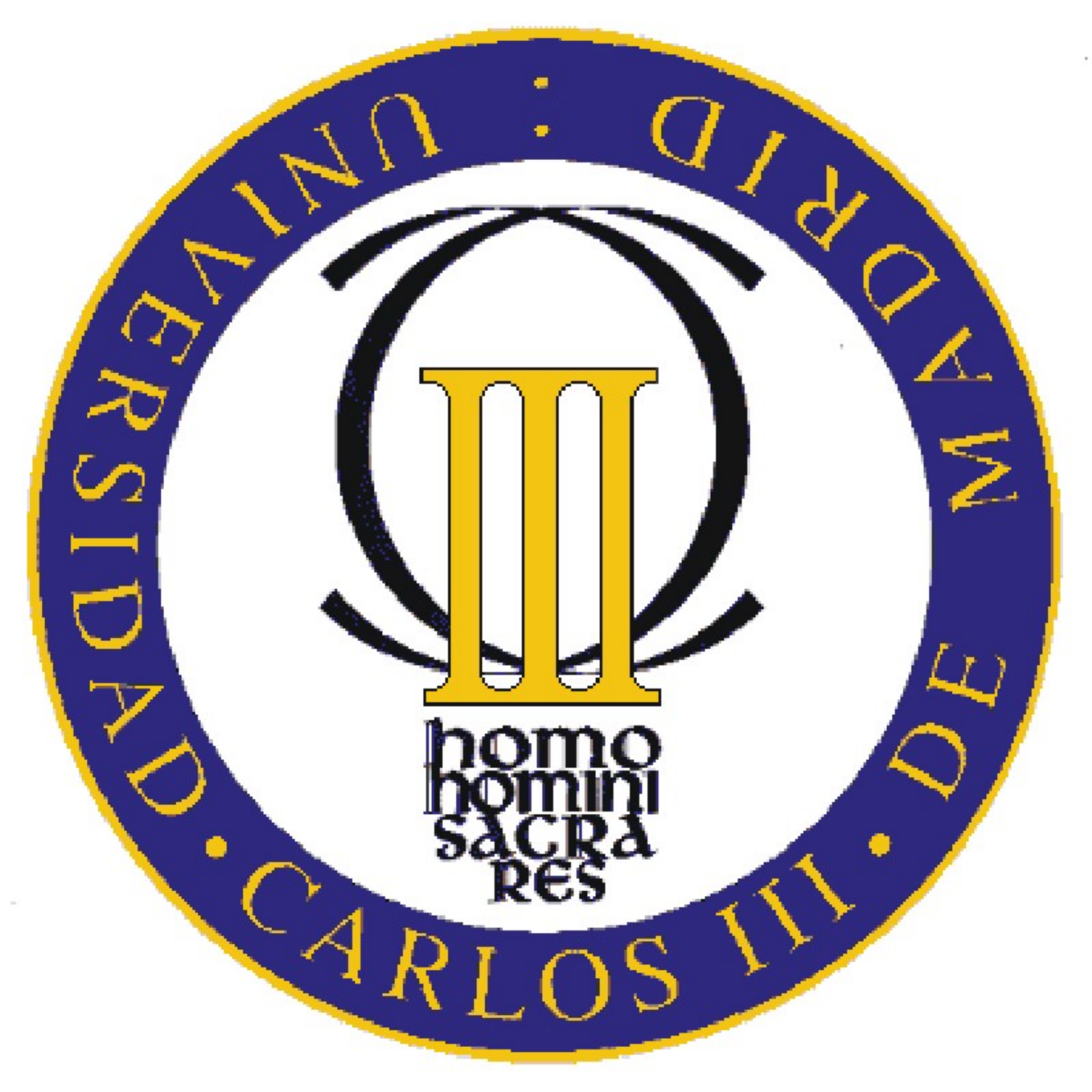}}}}

\vspace{0.5cm}

{\large UNIVERSIDAD CARLOS III DE MADRID}

\vspace{0.75cm}

{\Large TESIS DOCTORAL} \vspace{0.75cm}

\rule{\linewidth}{0.2mm}

\vspace{2mm}


{\Large\sc\textbf{ On the Theory of Self-Adjoint Extensions of }}
\vspace{1mm}

{\Large\sc\textbf{the Laplace-Beltrami Operator,}}
\vspace{1mm}

{\Large\sc\textbf{Quadratic Forms and Symmetry}}
\rule{\linewidth}{0.2mm}


\vspace{1cm}

{\textbf{Autor: }}

\vspace{1mm}

{\textbf{Juan Manuel P\'erez Pardo}}

\vspace{3mm}

{\textbf{Directores:}}

\vspace{1mm}

{\textbf{Prof.~Alberto Ibort Latre}}

\vspace{1mm}

{\textbf{Prof.~Fernando Lled\'o Macau}}

\vspace{1.5cm}

{\large  DEPARTAMENTO DE MATEM\'ATICAS} \vspace{5mm}

{\large \textrm{Legan\'es, 2013}}
\end{center}



\newpage\thispagestyle{empty}

{\center
{\large Firma del Tribunal Calificador} \vspace{1cm}

\begin{tabular}[h]{l | c | r}
	  &\phantom{aaaaaaaaaaaaaaaaaaaaaaaaaaaaaaaaaaaaaaaaaaa} &{\large Firma}\phantom{aaaaaaaa}\\ \hline
	& &\\
	{\large Presidente:}& & \\
	 & &\\ \hline
	  & &\\
	{\large Vocal:} & &\\
	 & & \\ \hline
	 & &\\
	{\large Secretario: } & &\\
	& &\\ \hline
\end{tabular}

\vspace{2cm}

{\large Calificaci\'{o}n\;:}
}
\vspace{2cm}

\begin{center}
{Legan\'{e}s, \phantom{aaaaaaa}de \phantom{aaaaaaaaaaaaaaaaaaaaaaaaaaaaaaaaa} de}
\end{center}

\newpage\thispagestyle{empty}

\begin{center}

\rule{\linewidth}{0.2mm}

\vspace{2mm}


{\Large\sc\textbf{ On the Theory of Self-Adjoint Extensions of }}
\vspace{1mm}

{\Large\sc\textbf{the Laplace-Beltrami Operator,}}
\vspace{1mm}

{\Large\sc\textbf{Quadratic Forms and Symmetry}}
\rule{\linewidth}{0.2mm}


\vspace{1.5cm}

{\Large PhD Thesis}

\vspace{3.5cm}

{\textbf{Author: }}

\vspace{1mm}

{\textbf{Juan Manuel P\'erez Pardo}}

\vspace{3mm}

{\textbf{Advisors:}}

\vspace{1mm}

{\textbf{Prof.~Alberto Ibort Latre}}

\vspace{1mm}

{\textbf{Prof.~Fernando Lled\'o Macau}}

\vspace{1.5cm}

\end{center}

\clearemptydoublepage

\pagenumbering{roman}

\thispagestyle{empty}
\vspace*{50mm}
\begin{flushright}
{\large \textit{Por un mundo feliz}}
\end{flushright}
\clearemptydoublepage

\addcontentsline{toc}{chapter}{Contents}
\tableofcontents
\markboth{\sc Contents}{\sc Contents}
\clearemptydoublepage



\thispagestyle{empty}
\chapter*{Agradecimientos}
\addcontentsline{toc}{chapter}{Agradecimientos}
\markboth{Agradecimientos}{}

\selectlanguage{spanish}

Despu\'{e}s de varios a\~{n}os embarcado en esta aventura ha llegado el momento de escribir los agradecimientos. Aunque para algunas personas escribir esta parte de la tesis  pueda ser poco m\'{a}s que un tr\'{a}mite o una obligaci\'{o}n, yo me veo en la necesidad de agradecer verdaderamente todo el apoyo que he recibido y el esfuerzo que mucha gente me ha dedicado durante todo este tiempo.\\

En primer lugar quiero citar a mis flamantes directores de tesis Alberto Ibort y Fernando Lled\'{o}. No tengo palabras para describir el gran respeto que me merecen como personas y como cient\'{i}ficos. La cualidad que m\'{a}s aprecio en Alberto (como cient\'{i}fico, como persona no sabr\'{i}a por d\'{o}nde empezar) es su capacidad para vislumbrar la soluci\'{o}n de un problema con la m\'{i}nima informaci\'{o}n, por muy complejo que \'{e}ste sea , ``\emph{!`Que s\'{i} que sale hombre, que s\'{i}!}'' Por poner una analog\'{i}a de c\'{o}mo es Alberto a nivel cient\'{i}fico aqu\'{i} ten\'{e}is el siguiente acertijo:

\emph{Se abre el tel\'{o}n y aparecen un tornillo, una goma el\'{a}stica y un l\'{a}piz. Se cierra el tel\'{o}n. Se abre el tel\'{o}n y aparece un f\'{o}rmula 1. ?`Qui\'{e}n ha pasado por all\'{i}?\footnote{\begin{turn}{180}MacGyver.\end{turn}}}

De Fernando me gustar\'{i}a destacar su atenci\'{o}n al detalle, ``\emph{Es cuesti\'{o}n de estilo}'', que m\'{a}s de una vez nos ha costado una discusi\'{o}n pero que sin lugar a dudas ha mejorado notablemente los contenidos que se encuentran en esta memoria. Much\'{i}simas gracias a los dos y a vuestras maravillosas familias.\\

Tambi\'{e}n quer\'{i}a agradecer a los profesores Beppe Marmo, \foreignlanguage{english}{``\emph{It is better not to increase Entropy}''}; A.P.~
Balachandran, \foreignlanguage{english}{``\emph{Why is that so?}''}; Manolo Asorey, ``\emph{?`Yo he dicho eso?}'' y Franco Ventriglia, \foreignlanguage{italian}{``\emph{Scorzette d'arancia ricoperte di cioccolato}''} todos sus comentarios, ense\~{n}anzas y hospitalidad. Las discusiones con ellos han dado lugar a muchos de los resultados que aparecen en esta memoria.\\

Otros profesores y miembros del departamento de matem\'{a}ticas de la UCIIIM que han contribuido en mayor o menor medida a mi formaci\'{o}n predoctoral y a la realizaci\'{o}n de esta tesis son Olaf Post, Froil\'{a}n Mart\'{i}nez, Daniel Peralta, Julio Moro, Bernardo de la Calle y Alberto Calvo.\\

A lo largo de estos a\~{n}os he conocido muchos compa\~{n}eros y amigos y, aunque va a ser dif\'{i}cil citarlos a todos, voy a hacer un intento. Todos ellos han contribuido a crear un ambiente agradable donde poder trabajar. Gio, por todos los ratos del caf\'{e}. Sveto,

\emph{- Hoy das t\'{u} el seminario intergrupos, ?`no?
\newline\indent
- ?`De verdad? Ay, Ay, Ay.}

\noindent Sergio, que se puede hablar con \'{e}l a trav\'{e}s de las paredes. Edo y Daniel, los compa\~{n}eros silenciosos. Espero que piensen lo mismo de mi. Jele, con su peculiar visi\'{o}n del mundo, ``\emph{!`En Londres se come bien! Hay ensaladas, comida china, tailandesa, filipina,\dots}'' Alberto L\'{o}pez, ``\emph{Esto es un calculillo}''. Leonardo, ``\emph{No puedo dormir en esa cama}''. Esta \'{u}ltima frase explica algunos acontecimientos recientes. Y, en definitiva, a todos con los que he tenido el placer de trabajar estos a\~{n}os. Esta lista es larga y si me olvido a alguien me lo va a recordar toda la vida. Os pido por favor a los que le\'{a}is esta lista que no le deis mucha publicidad, pod\'{e}is herir la sensibilidad de alguien: Javi, Elyzabeth, Pedro, Julio, Kenier, Alfredo, Mar\'{i}a, Ulises, Hector, Mari Francis, Natalia, Walter y Yadira.\\

No me olvido del grupo de los seminarios informales Alejandro, Borja, Carlos, Edu, Emilio y Marco. Por cierto, Borja, nos debes uno.\\

A los amigotes de Madrid y Valencia, esa fuerza poderosa que ha empujado en contra de la finalizaci\'{o}n de esta memoria. 

``\emph{S\'{o}lo una ca\~{n}a y a casa.}''

``\emph{?`Ya te vas? !`Si acabo de llegar!}''

``\emph{?`Bo que pacha?}''

``\emph{Ataco Ucrania desde Ural.}''

``\emph{?`Cu\'{a}ndo vienes?}''

``\emph{?`Ma\~{n}ana piscina?}''

``\emph{?`Una partidita?}''

``\emph{Tienes que comprarte la Play.}'' 

\noindent Ya sab\'{e}is quienes sois, no empiezo a citaros que corro el riesgo de acabar como Almodovar.\\

Tambi\'{e}n le quiero agradecer a toda mi familia, incluyendo la pol\'{i}tica, todo su apoyo. En especial a mi madre y a mi hermano.

``\emph{Hijo, ?`en qu\'{e} piensas?}''

``\emph{Hijo, ?`en qu\'{e} piensas?}''

``\emph{Hijo, ?`en qu\'{e} est\'{a}s pensando?}''

``\emph{Hijo, me preocupas.}''

\noindent Los que compartan profesi\'{o}n conmigo entender\'{a}n lo que quiero expresar. ?`Qu\'{e} se puede decir de una madre que te lo ha dado todo? Pues eso. A mi hermano, ``\emph{Los hipsters son como nuestros modernetes pero con un toque lumpen muy calculado}'', le deseo lo mejor en estos momentos de cambio y le tengo que agradecer haberme ense\~{n}ado que siempre hay otra manera de hacer las cosas.\\

Y por supuesto a Lara, ``\emph{Juanmi\dots}'', ``\emph{!`Ay Gal\'{a}n!}'', ``\emph{?`Y mi sorpresa?}'' Si hay una persona que sabe que este trabajo no acaba cuando sales del despacho, es ella. Por eso, a parte de corresponder su amor, cari\~{n}o y comprensi\'{o}n tambi\'{e}n quiero disculparme por todas las veces que no puedo dedicarle el tiempo que se merece.\\

En resumen, a todos, mil gracias.

\vspace*{5mm}
\begin{flushright}
Juan Manuel P\'{e}rez Pardo
\\
Legan\'{e}s, 2013
\end{flushright}


\selectlanguage{english}

\footnotesize{This work was partly supported by the project MTM2010-21186-C02-02 of the spanish {\em Ministerio de Ciencia e Innovaci\'on}, QUITEMAD programme P2009 ESP-1594 and the 2011 and 2012 mobility grants of ``Universidad Carlos III de Madrid''.} \normalsize

\clearemptydoublepage

\chapter*{Resumen}
\addcontentsline{toc}{chapter}{Resumen}
\markboth{Resumen}{}

\hyphenation{cons-truc-ci 
			o-pe-ra-dor 
			auto-ad-jun-tas 
			de-sa-rro-llan-do 
			o-pe-ra-do-res 
			ma-ne-ra 
			par-ti-cu-lar 
			na-tu-ral 
			co-rres-pon-dien-te 
			ca-rac-te-ri-za-ci
			de-di-ca
			ma-yo-res
			}

\selectlanguage{spanish}

El objetivo principal de esta memoria es analizar en detalle tanto la construcci\'{o}n de extensiones autoadjuntas del operador de Laplace-Beltrami definido sobre una variedad Riemanniana compacta con frontera, como el papel que juegan las formas cuadr\'{a}ticas a la hora de describirlas. M\'{a}s a\'{u}n, queremos enfatizar el papel que juegan las formas cuadr\'{a}ticas a la hora de describir sistemas cu\'{a}nticos.

Es bien conocido que $-\Delta_{\mathrm{min}}$, la extensi\'{o}n minimal del operador de Laplace-Beltrami, es autoadjunta cuando la variedad Riemanniana no tiene frontera. Sin embargo, cuando la variedad Riemanniana tiene frontera este operador es sim\'{e}trico pero no autoadjunto. Esta situaci\'{o}n es com\'{u}n en el el an\'{a}lisis de sistemas cu\'{a}nticos con frontera. Por ejemplo, el operador anterior describe la din\'{a}mica de una part\'{i}cula libre confinada a desplazarse por la variedad. El an\'{a}lisis de sistemas cu\'{a}nticos con frontera est\'{a} recibiendo atenci\'{o}n creciente por parte de la comunidad ya que hay un gran n\'{u}mero de situaciones f\'{i}sicas en las que la frontera juega un papel prominente. Por citar algunos ejemplos: el efecto Hall cu\'{a}ntico, los aislantes topol\'{o}gicos, el efecto Casimir, los grafos cu\'{a}nticos, etc. Otras situaciones f\'{i}sicas en las que la frontera juega un papel importante es en el dise\~{n}o de modelos efectivos que describen las interacciones con impurezas o con interfases entre materiales, v\'{e}ase \cite{albeverio05} y las referencias all\'{i} citadas. Definir un operador autoadjunto en un sistema con frontera requiere que ciertas condiciones de contorno sean especificadas. Resulta por lo tanto de gran importancia fijar las condiciones de contorno apropiadas, ya que la din\'{a}mica de un sistema cu\'{a}ntico no queda bien determinada hasta que se haya elegido un operador autoadjunto que la describa. 
La importancia de los operadores autoadjuntos en Mec\'{a}nica Cu\'{a}ntica radica en que, de acuerdo con los postulados de la Mec\'{a}nica Cu\'{a}ntica, son interpretados como las magnitudes observables de los sistemas que describen. Su espectro se interpreta como los posibles resultados de un proceso de medida. M\'{a}s a\'{u}n, el teorema de Stone \cite{stone32} establece una correspondencia un\'{i}voca entre grupos uniparam\'{e}tricos fuertemente continuos y operadores autoadjuntos. Por ello, los operadores autoadjuntos son los principales objetos que caracterizan la evoluci\'{o}n unitaria de un sistema cu\'{a}ntico. Vale la pena mencionar en este punto que aunque el espectro de los operadores autoadjuntos constituya el conjunto de posibles resultados para un proceso de medida, es la forma cuadr\'{a}tica asociada a ellos, $\scalar{\Phi}{T\Phi}$, siendo $T$ el operador autoadjunto, la que describe el valor esperado resultado de una medida cuando el estado del sistema queda descrito por el vector de estado $\Phi$\,. Las formas cuadr\'{a}ticas pueden, por lo tanto, desempe\~{n}ar el mismo papel que los operadores autoadjuntos en la descripci\'{o}n de la Mec\'{a}nica Cu\'{a}ntica.

La primera persona en darse cuenta de la diferencia entre operadores sim\'{e}tricos y autoadjuntos fue J.~von~Neumann en la d\'{e}cada de 1920, quien resolvi\'{o} completamente la cuesti\'{o}n en el marco m\'{a}s abstracto posible. Desde entonces ha habido varias revisiones, en distintos contextos, a la teor\'{i}a de extensiones autoadjuntas. Por ejemplo, G.~Grubb \cite{grubb68} en las d\'{e}cadas de 1960 y 1970 obtuvo caracterizaciones en el contexto de las EDP's. Por otro lado M.G.~Krein \cite{krein47a,krein47b} y J.M.~Berezanskii \cite{berezanskii68}, entre otros, obtuvieron caracterizaciones utilizando la teor\'{i}a de escalas de espacios de Hilbert. M\'{a}s recientemente se ha venido desarrollando la teor\'{i}a de ``boundary triples''. Una revisi\'{o}n detallada de esta caracterizaci\'{o}n es debida a J.~Br\"{u}ning, V.~Geyler y K.~Pankrashkin \cite{bruning08}, aunque varios de los resultados son conocidos desde tiempo atr\'{a}s \cite{visik52,kochubei75,derkach91,gorbachuk91}.

La estrategia que vamos a seguir en esta memoria ocupa, en cierto sentido, el espacio existente entre la caracterizaci\'{o}n propuesta por G.~Grubb y la caracterizaci\'{o}n en t\'{e}rminos de ``boundary triples''. En la primera se utiliza la estructura en la frontera para caracterizar las extensiones autoadjuntas en t\'{e}rminos de operadores pseudo-diferenciales que act\'{u}an sobre los espacios de Sobolev de la frontera. Por el contrario, la teor\'{i}a de ``boundary triples'' substituye la estructura de la frontera por otra en t\'{e}rminos de un espacio abstracto. Se demuestra entonces que las extensiones autoadjuntas est\'{a}n en correspondencia un\'{i}voca con el conjunto de operadores unitarios que act\'{u}an sobre este espacio abstracto. 

La estrategia presentada en esta memoria preserva la estructura de la frontera al mismo tiempo que parametriza el conjunto de extensiones autoadjuntas en t\'{e}rminos de operadores unitarios. Es por lo tanto m\'{a}s cercana a la caracterizaci\'{o}n conocida como ``quasi-boundary triples'' \cite{behrndt07,behrndt12}. Sin embargo, la estrategia aqu\'{i} presentada, sigue un camino totalmente distinto a las caracterizaciones citadas anteriormente. En contraste con ellas, incluida la de von Neumann, donde el operador adjunto se restringe a un conjunto en el que se demuestra con posterioridad que es autoadjunto, nosotros vamos a utilizar las ideas y resultados introducidos por K.~Friedrichs \cite{friedrichs34} y T.~Kato \cite{kato95}. De acuerdo con ellas la extensi\'{o}n autoadjunta es obtenida gracias a la caracterizaci\'{o}n de una forma cuadr\'{a}tica semiacotada que est\'{a} asociada al operador sim\'{e}trico. M\'{a}s concretamente, las formas cuadr\'{a}ticas semiacotadas est\'{a}n asociadas de manera un\'{i}voca a un operador autoadjunto siempre que sean cerrables. La mayor dificultad radica pues en demostrar que son cerrables. Al contrario de lo que ocurre con los operadores sim\'{e}tricos, que siempre son cerrables. Sin embargo, los operadores sim\'{e}tricos tienen en general infinitas extensiones autoadjuntas posibles. Desafortunadamente no hay ninguna generalizaci\'{o}n conocida al caso de formas cuadr\'{a}ticas que no sean semiacotadas. El ejemplo can\'{o}nico de esta estrategia es la forma cuadr\'{a}tica $Q(\Phi)=\norm{\d\Phi}^2$\,. Cuando \'{e}sta se define sobre $\H^1_0(\Omega)$\,, el operador autoadjunto asociado resulta ser la extensi\'{o}n de Dirichlet del operador $-\Delta_{\mathrm{min}}$. Si se define sobre $\H^1(\Omega)$\,, el operador autoadjunto asociado es la extensi\'{o}n de Neumann. Para obtener caracterizaciones de las extensiones autoadjuntas del operador de Laplace-Beltrami vamos a analizar la forma cuadr\'{a}tica 
$$\norm{\d\Phi}^2-\scalarb{\varphi}{\dot{\varphi}}\;.$$
Esta forma cuadr\'{a}tica puede ser entendida como una perturbaci\'{o}n singular de la forma cuadr\'{a}tica de Neumann citada anteriormente. Sin embargo, los resultados que aparecen en \cite{koshmanenko99} sobre perturbaciones singulares cerrables de formas cuadr\'{a}ticas no se pueden aplicar directamente a nuestro caso y otro enfoque es necesario. Desgraciadamente, las extensiones que vamos a poder caracterizar no incluyen todas las posibles, pero s\'{i} una vasta familia.

La caracterizaci\'{o}n que aqu\'{i} se ha propuesto utiliza como caracterizaci\'{o}n principal la siguiente ecuaci\'{o}n de condiciones de contorno en la frontera
\begin{equation*}\label{Resasoreycha6}
	\varphi-\mathrm{i}\dot{\varphi}=U(\varphi+\mathrm{i}\dot{\varphi})\;,
\end{equation*}
donde $U\in\mathcal{U}(\mathcal{L}^2(\Omega))$\,, y siendo $\Omega$\, la variedad Riemanniana. V\'{e}ase la Proposici\'{o}n \ref{prop: asorey}. Las condiciones de contorno descritas por esta ecuaci\'{o}n est\'{a}n en correspondencia un\'{i}voca con las extensiones autoadjuntas del operador de Laplace-Beltrami s\'{o}lo si la variedad $\Omega$ es de dimensi\'{o}n uno. Sin embargo, el hecho de que est\'{e} \'{i}ntimamente relacionada con la estructura de la frontera la convierte en un objeto muy adecuado para estudiar las extensiones autoadjuntas del operador $-\Delta_{\mathrm{min}}$\,. Los resultados expuestos en el Cap\'{i}tulo \ref{cha:QF}, el Cap\'{i}tulo \ref{cha:FEM} y el Cap\'{i}tulo \ref{cha:Symmetry} dan buena cuenta de ello. 

En el Cap\'{i}tulo \ref{cha:QF} hemos demostrado cu\'{a}les son las suposiciones necesarias para que la condici\'{o}n de contorno anterior determine una extensi\'{o}n autoadjunta y semiacotada del operador de Laplace-Beltrami. M\'{a}s a\'{u}n, hemos a\~{n}adido una serie de ejemplos relevantes que verifican las anteriores condiciones. En particular, el conjunto de condiciones de contorno de tipo peri\'{o}dico descrito en el Ejemplo \ref{periodic} es apropiado para el an\'{a}lisis de cambios de topolog\'{i}a en la teor\'{i}a cu\'{a}ntica, v\'{e}anse \cite{asorey05,balachandran95,shapere12}. La caracterizaci\'{o}n de operadores unitarios con ``gap en -1'', v\'{e}ase la Definici\'{o}n \ref{DefGap}, es equivalente a que el operador autoadjunto correspondiente sea semiacotado inferiormente. Como se demuestra en el Ejemplo \ref{generalized Robin} y en el Ejemplo \ref{generalized Robin2} estas extensiones autoadjuntas incluyen las condiciones de contorno de tipo Robin. As\'{i} que, como caso particular, hemos demostrado que las condiciones de contorno de la forma 
$$\dot{\varphi}=g\cdot\varphi\quad g\in\C^0(\pO)\;,$$
dan lugar a extensiones semiacotadas del operador de Laplace-Beltrami.

Este tipo de condiciones de contorno aparece en el estudio de sistemas cu\'{a}n\-ticos con frontera como los relacionados con los aislantes topol\'{o}gicos y con el efecto Hall cu\'{a}ntico. De hecho, estas condiciones de contorno aparecen de manera natural en la interfase entre dos materiales  cuando uno de ellos est\'{a} en el r\'{e}gimen superconductor. El hecho de que los operadores que describen la din\'{a}mica correspondiente sean semiacotados inferiormente es de gran importancia para la consistencia de la teor\'{i}a f\'{i}sica que los describe. M\'{a}s a\'{u}n, es conocido que los estados con las energ\'{i}as m\'{a}s bajas se hayan fuertemente localizados en la frontera. El estado fundamental que se muestra en la parte derecha de la Figura \ref{laplaciano4} es un ejemplo de esta localizaci\'{o}n. En investigaciones recientes llevadas a cabo junto con M.~Asorey y A.P.~Balachandran, \cite{asorey13}, se demuestra que la existencia de estos estados de frontera es un fen\'{o}meno com\'{u}n que ocurre no s\'{o}lo para el Laplaciano sino tambi\'{e}n para el operador de Dirac o para la generalizaci\'{o}n del operador de Laplace a tensores de mayor orden, el operador de Laplace-de Rham, lo que incluye el campo electromagn\'{e}tico. De hecho, los resultados del Cap\'{i}tulo \ref{cha:QF} se pueden extender a la situaci\'{o}n general de fibrados Herm\'{i}ticos, donde el operador de Laplace-Beltrami es generalizado al operador conocido como Laplaciano de Bochner. La obstrucci\'{o}n m\'{a}s grande para obtener esta generalizaci\'{o}n es encontrar un substituto adecuado para el operador radial $R$, v\'{e}ase la Definici\'{o}n \ref{Defunidimensional} y la Proposici\'{o}n \ref{prop: intervalrobin}. Afortunadamente esta obstrucci\'{o}n puede ser solventada, pero dejamos esta generalizaci\'{o}n para futuros trabajos. La estructura particular de las condiciones de contorno antes mencionadas se usa en los siguientes cap\'{i}tulos para demostrar una serie de resultados adicionales.

En el Cap\'{i}tulo \ref{cha:FEM} se propone un esquema num\'{e}rico para aproximar el problema espectral. En particular, se demuestra que la convergencia del esquema propuesto est\'{a} asegurada siempre y cuando se d\'{e} la condici\'{o}n de gap. Como muestra de la factibilidad de dicho esquema num\'{e}rico se propone una realizaci\'{o}n concreta del mismo para el caso de dimensi\'{o}n uno. Los experimentos num\'{e}ricos realizados con ella muestran, entre otras cosas, que las tasas de convergencia se satisfacen. Vale la pena mencionar que la convergencia del esquema propuesto se ha demostrado con independencia de la dimensi\'{o}n. Por lo tanto, esta familia de algoritmos num\'{e}ricos se puede usar para aproximar las soluciones de los problemas en dimensi\'{o}n dos y superiores. En particular, se puede aplicar al estudio de los cambios de topolog\'{i}a y de los estados de frontera antes mencionados. La tarea de programar una versi\'{o}n en dimensi\'{o}n dos del algoritmo propuesto en esta memoria ya est\'{a} siendo llevada a cabo junto a A.~L\'{o}pez Yela.

El Cap\'{i}tulo \ref{cha:Symmetry} est\'{a} dedicado al an\'{a}lisis que juegan las invarianzas por simetr\'{i}a en la construcci\'{o}n de las diferentes extensiones autoadjuntas. Usando el marco m\'{a}s abstracto posible, es decir, utilizando la caracterizaci\'{o}n de von Neumann de extensiones autoadjuntas, hemos demostrado que el conjunto de extensiones autoadjuntas invariantes de un operador sim\'{e}trico se haya en correspondencia un\'{i}voca con aquellas isometr\'{i}as entre los espacios de deficiencia que se encuentran en el conmutante de la representaci\'{o}n unitaria del grupo de simetr\'{i}a, v\'{e}anse el Teorema \ref{thmvonNeumann} y el Teorema \ref{Ginvariantoperator}. Esto muestra que, aunque un operador sim\'{e}trico sea invariante bajo un grupo de simetr\'{i}a, puede ocurrir que sus extensiones autoadjuntas no lo sean. Consideremos un operador que se construye como el producto tensorial de dos operadores sim\'{e}tricos. Es un error com\'{u}n asumir que todas las extensiones autoadjuntas del operador producto se pueden obtener a trav\'{e}s del producto tensorial de las isometr\'{i}as que definen las extensiones autoadjuntas de los factores. El motivo de esto se entiende m\'{a}s f\'{a}cilmente utilizando la caracterizaci\'{o}n en t\'{e}rminos de condiciones de contorno presentada en esta memoria. Consideremos un operador sim\'{e}trico definido en una variedad 
$$\Omega=\Omega_1\times\Omega_2\;.$$ 
Entonces, las extensiones autoadjuntas vendr\'{a}n parametrizadas utilizando el espacio de Hilbert de funciones de cuadrado integrable definido sobre la frontera de la variedad
$$\pO=\pO_1\times\Omega_2 \sqcup \Omega_1\times\pO_2\;.$$
El error com\'{u}n consiste en considerar \'{u}nicamente aquellas condiciones de contorno que se pueden definir sobre $$\pO_1\times\pO_2\;.$$
Los resultados obtenidos en el Cap\'{i}tulo \ref{cha:Symmetry}, principalmente
el Teorema \ref{Ginvariantoperator}, el Teorema \ref{QAinvariant} y el Teorema \ref{InvariantRepresentation}, proporcionan herramientas adecuadas para lidiar con estas consideraciones. En particular, se trata con detalle la familia de formas cuadr\'{a}ticas asociadas al operador de Laplace-Beltrami introducida en el Cap\'{i}tulo \ref{cha:QF} y demostramos que el conjunto de extensiones autoadjuntas compatibles con las transformaciones por simetr\'{i}a de la variedad Riemanniana tambi\'{e}n est\'{a} relacionado con el conmutante de la representaci\'{o}n unitaria del grupo.
La \'{u}ltima secci\'{o}n del Cap\'{i}tulo \ref{cha:Symmetry} se dedica a la generalizaci\'{o}n del teorema de representaci\'{o}n de Kato, Teorema \ref{fundteo}, al caso en el que las formas cuadr\'{a}ticas no son semiacotadas. Para ello introducimos la noci\'{o}n de sector de la forma cuadr\'{a}tica, v\'{e}ase la Definici\'{o}n \ref{def: poa}. Esta noci\'{o}n juega un papel an\'{a}logo al concepto de subespacios invariantes que aparece en el an\'{a}lisis de operadores autoadjuntos. Hemos demostrado en el Teorema \ref{representationpoa} que las formas cuadr\'{a}ticas cuyos sectores sean simult\'{a}neamente cerrables y semiacotados son representables en t\'{e}rminos de operadores autoadjuntos. Por lo tanto, hemos cimentado el camino hacia una generalizaci\'{o}n del Teorema de representaci\'{o}n de Kato para formas cuadr\'{a}ticas genuinamente no acotadas, es decir, no semiacotadas. Esta generalizaci\'{o}n contin\'{u}a siendo uno de los mayores problemas abiertos en el campo de formas cuadr\'{a}ticas. El \'{u}ltimo paso para obtener dicha generalizaci\'{o}n ser\'{i}a identificar bajo qu\'{e} circunstancias, las formas cuadr\'{a}ticas Herm\'{i}ticas, poseen sectores semiacotados y cerrables. En particular hemos demostrado a trav\'{e}s del Ejemplo \ref{ex:nogap} que la condici\'{o}n de gap, v\'{e}ase la Definici\'{o}n \ref{DefGap}, no es condici\'{o}n necesaria para que la clase de formas cuadr\'{a}ticas del Cap\'{i}tulo \ref{cha:QF} conduzca a extensiones autoadjuntas del operador de Laplace-Beltrami.

\selectlanguage{english}

\clearemptydoublepage
\pagenumbering{arabic}
\chapter{Introduction}
\label{cha:introduction}

The main objective of this dissertation is to analyse thoroughly the construction of self-adjoint extensions of the Laplace-Beltrami operator defined on a compact Riemannian manifold with boundary and the role that quadratic forms play to describe them. Moreover, we want to emphasise the role that quadratic forms can play in the description of quantum systems. 

It is well known that $-\Delta_{\eta}$, the Laplace-Beltrami operator associated to the Riemannian metric $\eta$, is essentially self-adjoint, i.e., it possesses a unique self-adjoint extension, when the Riemannian manifold has void boundary \cite{lawson89}. However, when the Riemannian manifold has boundary this operator is symmetric but not self-adjoint. Such situation is common in the analysis of quantum systems with boundary. For instance, the aforementioned operator describes the dynamics of a free particle that is constrained to the given manifold. The analysis of quantum systems with boundary is receiving increasing attention from the community since there is a number of physical situations where the boundary plays a prominent role, i.e., quantum Hall effect \cite{morandi88}, topological insulators \cite{hasan10}, the Casimir effect \cite{plunien86}, quantum graphs \cite{Post12a}, \dots \\ Other situations in physics where the determination of self-adjoint extensions of operators is of great importance is in the definition of effective models that describe interactions with impurities or with interphases of materials, see \cite{albeverio05} and references therein. Defining a self-adjoint operator on a system with boundary requests certain boundary conditions to be specified in order to extend the corresponding symmetric operator to a self-adjoint one. Selecting the appropriated boundary conditions is thus of great importance since the dynamics of a quantum system are not well determined until a self-adjoint operator is selected. 

The importance of self-adjoint operators in Quantum Mechanics relies in that they are interpreted, according to the postulates of Quantum Mechanics, as the observables of the corresponding quantum systems. Their spectrum is interpreted as the possible outcomes of a given measurement. Moreover, Stone's Theorem \cite{stone32} establishes a one-to-one correspondence between strongly continuous one-parameter unitary groups and self-adjoint operators. Hence, self-adjoint operators are the main objects characterising the unitary evolution of a quantum system. It is worth to mention at this point that, even though the spectrum of self-adjoint operators is considered to be the outcome of possible measurements, the quadratic forms associated to them, i.e., $\scalar{\Phi}{T\Phi}$, describe the mean values expected to be measured when the state of the system is described by the state vector $\Phi$. Thus, quadratic forms may very well play the same role that self-adjoint operators do in the description of Quantum Mechanics. 

It was J. von Neumann who realised the fundamental difference between symmetric and self-adjoint operators and then solved the problem of determining all self-adjoint (and symmetric) extensions of a densely defined symmetric operator (if any) in the late 1920's \cite{neumann30}. Since then there have been several approaches in different contexts to treat the theory of self-adjoint extensions. For instance, G.~Grubb \cite{grubb68} in the 1960's and 1970's provided characterisations in the context of PDE's; M.G.~Krein \cite{krein47a,krein47b} and J.M.~Berezanskii \cite{berezanskii68}, among others, provided characterisations using the theory of scales of Hilbert spaces; and more recently the theory of boundary triples has been developed. J.~Br\"{u}ning, V.~Geyler and K.~Pankrashkin do a systematic review on this approach in \cite{bruning08}, although some of the results were known earlier, for instance \cite{derkach91,gorbachuk91,kochubei75,visik52}.

The approach that we will follow in this dissertation fills, in certain sense, the gap between the characterisation provided by G.~Grubb and the characterisation in terms of boundary triples. In the former the structure of the boundary is used to characterise the self-adjoint extensions in terms of pseudo-differential operators acting on the Sobolev spaces of the boundary. On the other hand the theory of boundary triples substitutes the structure of the boundary by other in an abstract space. The self-adjoint extensions are then proved to be in one-to-one correspondence with the set of unitary operators acting on this abstract space. The approach presented in this dissertation preserves the structure of the boundary but at the same time parameterises the set of self-adjoint extensions in terms of unitary operators. It is therefore closer to the approach called quasi boundary triples \cite{behrndt07,behrndt12}.

However we follow a slightly different path than in the previous approaches. In contrast to the them, including von Neumann's, where the adjoint operator is restricted to a set where it is utterly proved to be self-adjoint, we are going to follow the ideas and results introduced by K.~Friedrichs \cite{friedrichs34} and T.~Kato \cite{kato95}. According to them, the self-adjoint operator is obtained by the characterisation of an associated semi-bounded quadratic form. More concretely, semi-bounded quadratic forms are associated to a unique self-adjoint operator if they are closable. The main difficulty in this approach is to prove that the corresponding quadratic forms are closable. In contrast, symmetric operators are always closable but the closure need not be a self-adjoint operator and in general there are infinite possible self-adjoint extensions. The following table summarises the relationship between closable forms and operators.\\

Symmetric Operators
\begin{itemize}
\item Always closable.
\item The minimal extension does not need to be a self-adjojnt operator.
\item It is possible that none of the extensions is self-adjoint\\
\end{itemize}

Hermitean Quadratic Forms
\begin{itemize}
\item Not always closable.
\item The minimal extensions is always associated to a s.a. operator.
\end{itemize}

This approach has been applied successfully in \cite{albeverio79,albeverio97,albeverio91} to address the problem of perturbations of self-adjoint operators. The intimate relation between elliptic differential operators and quadratic forms has long been known, see, for instance, \cite{grubb73}. The  canonical example of this situation is the quadratic form $Q(\Phi)=\norm{\d\Phi}^2$. If this quadratic form is defined over $\H^1_0(\Omega)$\,, the associated operator is the Dirichlet extension of $-\Delta_{\mathrm{min}}$; and when defined on $\H^1(\Omega)$\,, the associated operator is the Neumann extension. Also equivariant and Robin-type Laplacians can be naturally described in terms of closed and semi-bounded quadratic forms (see, e.g., \cite{grubb11,kovarik12,lledo07,post12b}). In this context the subtle relation between quadratic forms and representing operators manifests through the fact that the form domain $\mathcal{D}(Q)$ always contains the operator domain $\mathcal{D}(T)$ of the representing operator. Therefore it is often possible to compare different form domains while the domains of the representing operators remain unrelated. This fact allows, e.g., to develop spectral bracketing techniques in very different mathematical and physical situations using the language of quadratic forms \cite{lledo08b,lledo08a}.

To obtain the characterisation of the self-adjoint extensions of the Laplace-Beltrami operator we will analyse the quadratic form 
$$\norm{\d\Phi}^2-\scalarb{\varphi}{\dot{\varphi}}\;.$$
This quadratic form can be understood as a singular perturbation of the Neumann quadratic form above. In this context, the theory of scales of Hilbert spaces \cite[Chapter I]{berezanskii68} provides a natural framework to characterise the corresponding extensions. However, the results appearing in \cite{koshmanenko99} on closable singular perturbations of quadratic forms can not be applied directly in our case and a different approach is needed. 

It was suggested by M.~Asorey, A.~Ibort and G.~Marmo \cite{asorey05} that the particular structure of the boundary conditions
\begin{equation}\label{intro: asorey}
	\varphi-\mathrm{i}\dot{\varphi}=U(\varphi+\mathrm{i}\dot{\varphi})\;,
\end{equation}
can be used to analyse the self-adjoint extensions of the Laplace-Beltrami operator. There it was pointed out that the identification of the set of possible self-adjoint extensions with the corresponding unitary group might be used to analyse the non-trivial structure of the space of self-adjoint extensions. For instance, this identification provides the set of self-adjoint extensions with the structure of a group. They identify that the particular subset of boundary conditions that plays the relevant role in topology change phenomena, and that they call the Cayley submanifold, can be characterised in terms of the spectrum of the aforementioned operator $U$. This approach serves as inspiration to construct quadratic forms on domains similar to those provided by \eqref{intro: asorey}. Now the problem is to proof their closability and, consequently, their semi-boundedness. For this matter, some restrictions have to be introduced in the unitary operators. After a careful analysis of the domains of the self-adjoint operators using the Lions trace theorem \cite[Theorem 8.3]{lions72} we arrive to the spectral gap condition. This is a condition on the spectrum of the unitary operator $U$ that describes the boundary condition. We prove that this condition is sufficient to ensure the semi-boundedness of the quadratic form and is one of the main contributions of this dissertation.

To prove this result a careful analysis of the structure of the underlying Riemannian manifold is needed. The manifold is split into a collar neighbourhood of the boundary and its complement, that shall be called the bulk. The contribution in the bulk is always positive whereas the contribution of the collar neighbourhood is always semi-bounded below provided that the gap condition holds. In order to ensure that the splitting procedure does not interfere with the domains of the different quadratic forms that are involved, we use the remarkable fact that, at the level of quadratic forms, imposing Neumann boundary conditions amounts to consider no conditions at all, cf., \cite[Theorem 7.2.1]{davies95}.

The previous analysis is useful not just for theoretical purposes but also in applications. The evaluation of the spectrum of self-adjoint operators is of main importance since, as stated previously, it consists of the quantities that will be measured in the laboratories. Thus, to possess tools, either analytical or numerical, to solve the corresponding eigenvalue problems is of main importance. Unfortunately, even in dimension one, to implement boundary conditions that are not standard can be very hard. This difficulty relies in part in the difficulties associated to the description of the corresponding domains. The particularly simple form of the boundary condition \eqref{intro: asorey} will help us with these purposes. Another important contribution of this dissertation is to provide a numerical scheme, based on the finite element method, that allows for the numerical computation of the eigenvalue problems of the quadratic forms introduced so far. Again, the gap condition remains to be crucial. In order for such algorithms to work one needs to introduce a special family of finite elements, that are delocalised along the boundary, and that allow for the implementation of the general boundary conditions of the form \eqref{intro: asorey}. We provide here a particular instance for the case of dimension one that is capable of computing the spectrum of any self-adjoint extension of the Laplace-Beltrami operator. This is so because in dimension one all the unitary operators describing boundary conditions verify the gap condition.

The convergence of this family of numerical schemes is proved independently of the dimension. Hence, they can be used to compute the spectrum of any self-adjoint extension of the Laplace-Beltrami operator on any Riemannian manifold provided that the gap condition holds.

Symmetries of quantum and classical systems play an important role in studying their properties. Symmetries are usually associated with ``symmetry groups'', in the case of continuous variables with the Lie groups, and their study has been instrumental in the development of Quantum Mechanics, c.f. \cite{cohen77,weyl50} or the celebrated paper by E.~Wigner on the representations of the Poincar\`{e} group \cite{wigner37}.
It was already noticed by Wigner that transformations of quantum systems preserving the transition probabilities must be implemented by (anti-)unitary operators \cite[Chapter 7]{galindo90}. 
Thus, groups of symmetries of quantum systems must be realised as (anti-)unitary representations.   Finally if $H$ is a self-adjoint operator describing either an observable of the quantum system or a one-parameter group
of unitary operators, it will be said to be $G$-invariant, or that the group $G$ is a symmetry group for $H$, if there is a unitary representation $V\colon G \to \mathcal{U}(\mathcal{H})$\,, $g \mapsto V(g)$\,, of the group $G$, cf., \cite[Chapter 2, Chapter 3]{thaller92}, such that: 
$$V(g)  e^{itH} V(g)^\dagger =  e^{itH} \;.$$

Consider a symmetric operator that is $G$-invariant. It is a common error to assume that its self-adjoint extensions will be $G$-invariant too. In fact, there are plenty of them that will not be invariant. The preceding considerations makes very relevant the characterisation of those extensions that will be $G$-invariant. In the present dissertation we provide an easy characterisation of those self-adjoint extensions that will be compatible with the given symmetries, namely those whose unitary operator $U$ commutes with the unitary representation of the group $G$. We will obtain such characterisation in several ways. First, in the most general context provided by von Neumann's theorem. Then, using Kato's representation theorem. And finally, using the class of examples of the family of quadratic forms associated to the Laplace-Beltrami operator. More concretely, we will study the case when the manifold itself is invariant under the action of the group. Two meaningful examples of this latter case will be discussed.

To conclude the dissertation we will try to step forward into one of the most important open problems in the field of quadratic forms. As we have been discussing so far, Kato's representation theorem is a powerful tool that helps in the characterisation of self-adjoint extensions of symmetric operators. It is also one of the main tools, if not the main, in the analysis of perturbations of self-adjoint operators, cf., \cite{albeverio05,koshmanenko99,reed78}. However, it does not hold if the corresponding quadratic forms are not semi-bounded. Of course, we restrict ourselves to the case of Hermitean quadratic forms. It exists a generalisation of the aforementioned representation theorem for quadratic forms that are not Hermitean, also proved by T.~Kato \cite[Theorem 2.1]{kato95}. In this latter case the assumption of semi-boundedness is substituted by the assumption of sectoriality. This latter condition requires the range of the quadratic form to be contained in a subset of the complex plain with the shape of a wedge of angle smaller than $\pi$. The spectrum of the corresponding operator is then proved to be contained in the same wedge.

The assumption of semi-boundedness is used in the proof of Kato's representation theorem as follows. One uses it to define a proper scalar product on the domain of the quadratic form that promotes it to a Hilbert space (pre-Hilbert space in the case that the quadratic form is closable but not closed). Subsequent use of Riesz's representation theorem provides the desired operator. The latter is ulteriorly proved to be self-adjoint.

Inspired by the theory of self-adjoint extensions of operators with symmetry and their corresponding invariant quadratic forms, we will introduce a family of not necessarily semi-bounded quadratic forms possessing a representation in terms of a self-adjoint operator. They will be called partially orthogonally additive quadratic forms and they constitute a generalisation of the so called orthogonally additive polynomials in Banach lattices.   It is well-known that (continuous) orthogonally additive polynomials on Banach lattices can be represented in a precise form by means of multiplication by a measurable function \cite{benyamini06}. In P.~Linares' thesis there were explored some extensions of these ideas to the particular instance of quadratic polynomials in Hilbert spaces \cite[Section 2.4, Section 2.5]{linares09}. We will extend here such approach dramatically by showing that an appropriate generalisation of the notion of orthogonal additivity leads in a natural way to a representation theorem of  (not necessarily semi-bounded) quadratic forms in terms of self-adjoint operators.   Particular instances of such theorem are provided by quadratic forms invariant under the action of compact Lie groups.\\

The dissertation is organised as follows. Chapter \ref{cha:Preliminaries} is devoted to establish the basic definitions and results on the theory of extensions of symmetric operators and closable quadratic forms. The basic definitions of the theory of Sobolev spaces, the Laplace-Beltrami operator and the theory of scales of Hilbert spaces can be found in this chapter too. In Chapter \ref{cha:QF} we introduce the class of quadratic forms associated to the Laplace-Beltrami operator that will be the main object of our study. We proof there the main results concerning closability and semi-boundedness, as well as we introduce the notion of unitary operators with gap and admissible unitary operators. In Chapter \ref{cha:FEM} we introduce a numerical scheme that can be used in any dimension to approximate the spectral problem of the quadratic forms introduced in Chapter \ref{cha:QF}. We proof the convergence of this scheme as well as we introduce an explicit one-dimensional version of it. Chapter \ref{cha:Symmetry} is devoted to the analysis of self-adjoint extensions when there is a group of symmetries present. The main results of Chapter \ref{cha:Preliminaries} and Chapter \ref{cha:QF} are reviewed in this context and the corresponding characterisation of the self-adjoint extensions compatible with the symmetry group, which will be called $G$-invariant, is introduced. In this chapter we propose a way to generalise Kato's representation theorem to Hermitean quadratic forms that are not semi-bounded. This remains to be one of the main open problems in the field.

\clearemptydoublepage
\chapter{Preliminaries}
\label{cha:Preliminaries}

This chapter is devoted to the introduction of well established results that we shall need during the rest of this dissertation as well as to fix our notation. Appropriated references are cited at the beginning of the corresponding sections.

\section{General theory of symmetric operators in Hilbert space}
\label{sec:vonNeumann}

In this section we review standard material concerning the theory of linear unbounded operators in Hilbert space. The proofs of the main results introduced in this section can be found at \cite[Chapter~VII, Chapter~VIII]{reed72} and \cite[Chapter~X]{reed75}.

In what follows we consider that $\H$ is a complex Hilbert space with inner product and norm denoted respectively by $\scalar{\cdot}{\cdot}$\newnot{symb:scalar}\,, $\norm{\cdot}$\newnot{symb:norm}\,. Our main object of interest is the study of self-adjoint extensions of the Laplace-Beltrami operator. It is worth to point out that the concept of self-adjoint extension of a symmetric operator arises when considering unbounded operators acting on the Hilbert space. Let us introduce briefly the notion of bounded and unbounded operators. As long as it is not stated otherwise we shall consider that we have a linear operator $T$ with dense domain $\D(T)$\newnot{symb:domainT} acting on $\H$.

\begin{definition}
	An operator $T$ is said to be \textbf{bounded} if it exists a positive constant $M$ such that $$\norm{T\Phi}\leq M\norm{\Phi}\quad\Phi\in\D(T)\;.$$
\end{definition}

Bounded operators are continuous on $\H$ and since $\D(T)$ is assume to be dense, they can be extended continuously to the whole Hilbert space. Hence, for bounded operators it is not necessary to select a domain for the operator. On the contrary, unbounded operators are not defined on the whole Hilbert space but only in a dense domain. In fact, since they are not bounded, there exist sequences $\{\Phi_n\}_n\in\D(T)$ that are convergent in $\H$ and such that $$\lim_{n\to\infty}\norm{T\Phi_n}\to\infty\;.$$ Hence, the operator $T$ is not continuous and can not be defined for the limit of such a sequence. It is possible to define a special class of unbounded operators that at least are continuous on their domain. These are the closed operators.

\begin{definition}
	Let $T$ be an operator with dense domain $\D(T)$ and let $\{\Phi_n\}_n$ be a Cauchy sequence in $\H$ such that $\{T\Phi_n\}_n$ is also a Cauchy sequence in $\H$. The operator $T$ is said to be \textbf{closed} if $$\lim_{n\to\infty}\Phi_n\in\D(T)\;.$$ An operator is said to be \textbf{closable} if it has closed extensions, i.e., if it exists an operator $\overline{T}$ such that $\D(\overline{T})\supset\D(T)$ and $$\overline{T}\Bigr|_{\D(T)}=T\;.$$ 
\end{definition}

This definition is equivalent to claim that $\D(T)$ is a complete Hilbert space with respect to the norm $$\normm{\cdot}_T=\sqrt{\norm{\cdot}^2+\norm{T\cdot}^2}\;.\newnot{symb:graphnormT}$$ In general it is always possible to complete $\D(T)$
with respect to the norm $\normm{\cdot}_T$. Unfortunately the extended object need not be a well defined operator and therefore not all the unbounded operators are closable. However there is a special class of operators, the symmetric operators, for which closed extensions always exist. In fact, given a symmetric operator $T$  one can always consider its extension to $\overline{\D(T)}^{\normm{\cdot}_T}$. This extension is closed and symmetric.  Before we define symmetric operators let us recall first the definition of the adjoint operator.

\begin{definition}\label{adjointoperator}
	Let $T$ be a densely defined operator. Let $\Psi\in\H$. Then $\Psi$ is in the domain of the \textbf{adjoint operator}, $\D(T^\dagger)$\,, if it exists $\chi\in\H$ such that
	$$\scalar{\Psi}{T\Phi}=\scalar{\chi}{\Phi}\quad\forall\Phi\in\D(T)\;.$$ In such a case $$T^\dagger\Psi:=\chi\;.$$
\end{definition}
Notice that $\chi$ is uniquely determined because $\D(T)$ is a dense subspace of $\H$\,. It is easy to show that the adjoint operator is always a closed operator.

\begin{definition}\label{def:symmetric}
	An operator $T$ is said to be \textbf{symmetric} if $$\scalar{\Psi}{T\Phi}=\scalar{T\Psi}{\Phi}\quad\forall\Psi,\Phi\in\D(T)\;.$$
\end{definition}

From the last two definitions it is easy to conclude that if $T$ is a symmetric operator then $T^\dagger$ is a closed extension of it, $$T^\dagger\Bigr|_{\D(T)}=T\;.$$ It is a general fact that the domain of the adjoint operator $\D(T^\dagger)$ contains the the domain $\D(T)$.

\begin{definition}
	An operator $T$ is said to be \textbf{self-adjoint} if it is symmetric and $\D(T)=\D(T^\dagger)$.
\end{definition}

The difference between symmetric operators and self-adjoint operators is subtle. For instance, if $T$ is a closed, bounded, symmetric operator then it is automatically self-adjoint. In general it is much easier to define symmetric operators than self-adjoint operators. This is due to the fact that the adjoint operator is a derivative object that depends on the definition of the symmetric operator. Unfortunately, the good properties that Hermitean operators do have in finite dimension, for instance, real spectrum or the spectral resolution of the identity, can be generalised only for self-adjoint operators. Since we will make use of it later let us formulate here, without proof, the spectral theorem, cf., \cite[Section VIII.3]{reed72}, \cite[Section VI.66]{akhiezer61b}. We need first the following definition.

\begin{definition}
	A spectral resolution of the identity is a one-parameter family of orthogonal projection operators $E_t$\newnot{symb:spectralres}, where $t$ runs through a finite or infinite interval $[\alpha, \beta]\subset\mathbb{R}$\,, which satisfy the following conditions:
	\begin{enumerate}
		\item $E_\alpha=0$\,, $E_\beta=\mathbb{I}_\H$\,.
		\item $\lim_{t\to t_0^{-}}E_t=E_{t_0}$\,,\quad $(\alpha<t_0<\beta)$\,.
		\item $E_uE_v=E_s$\,,\quad $(s=\min\{u,v\})$\,.
	\end{enumerate}
	In the case that $[\alpha, \beta]$ is infinite we define $E_{-\infty}:=\mathrm{s}-\lim_{t\to-\infty}E_t$ and $E_{\infty}:=\mathrm{s}-\lim_{t\to\infty}E_t$\,.
\end{definition}

\begin{theorem}[Spectral theorem]\label{spectraltheorem} Let $T$ be a self-adjoint operator. Then it exists a spectral resolution of the identity $E_\lambda$ such that $$\D(T)=\{\Phi\in\H\bigr|\int_{\mathbb{R}}\lambda^2\d(\scalar{\Phi}{E_\lambda\Phi})<\infty\}$$
and
$$T=\int_{\mathbb{R}}\lambda \d E_\lambda\;.$$
\end{theorem}

The last integral has to be understood in the sense of Bochner, cf., \cite{bochner33}. A similar theorem can be formulated for normal operators and in particular for unitary operators. The main difference between this cases is the support of the spectral resolution of the identity. A characterisation of the differences between the spectra of closed symmetric operators and self-adjoint operators is the following \cite[Theorem X.1]{reed75}.

\begin{theorem}\label{thesymmcases}
	Let $T$ be a closed symmetric operator. Then the spectrum of $T$ is either:
	\begin{enumerate}
		\item The full complex plain.
		\item The upper complex half-plain.
		\item The lower complex half-plain.
		\item A subset of the real axis.
	\end{enumerate}
	Case iv) happens iff $T$ is a self-adjoint operator.
\end{theorem}

A natural question that arises is if given a symmetric operator one is able to find a closed extension of it that is self-adjoint and wether or not it is unique. John von Neumann addressed this issue in the late 1920's and answered the question in the most abstract setting, cf., \cite{neumann30}. Let us state the result that will be of most interest for us \cite[Theorem X.2]{reed75}.

\begin{definition}\label{def:deficiencyspaces}
	Let $T$ be a closed, symmetric operator. The \textbf{deficiency spaces} $\mathcal{N}_{\pm}$\newnot{symb:deficiencyspace} are defined to be $$\mathcal{N}_{\pm}=\{\Phi\in\H\bigr|(T^\dagger\mp\mathbf{i})\Phi=0\}\;.$$ The \textbf{deficiency indices} are $$n_{\pm}=\operatorname{dim}\mathcal{N}_{\pm}\;.\newnot{symb:deficiencyindex}$$
\end{definition}

\begin{theorem}[von Neumann]\label{thmvonNeumann}
	Let $T$ be a closed symmetric operator. The self-adjoint extensions of $T$ are in one-to-one correspondence with the set of unitaries (in the usual inner product) of $\mathcal{N}_+$ onto $\mathcal{N}_-$. If $K$ is such a unitary then the corresponding self-adjoint operator $T_K$ has domain 
	$$\D(T_K)=\{\Phi+(\mathbb{I}+K)\xi\bigr|\Phi\in\D(T),\,\xi\in\mathcal{N}_+\}\;,$$ and
	$$T_K\bigl(\Phi+(\mathbb{I}+K)\xi\bigr)=T^\dagger\bigl(\Phi+(\mathbb{I}+K)\xi\bigr)= T\Phi+\mathbf{i}(\mathbb{I}+K)\xi\;.$$
\end{theorem}

\begin{remark}\label{rem: sa cases}
From all the possibilities we must point out the following ones
	\begin{itemize}
		\item If $n_+=n_-=0$ then the operator is self-adjoint.
		\item If $n_+=n_-\neq 0$ then the operator has infinite self-adjoint extensions.
		\item If $n_+\neq n_-$ then the operator has no self-adjoint extensions.
	\end{itemize}
\end{remark}
We end this section with the following definition.

\begin{definition}
	Let $T$ be a densely defined operator. We say that $T$ is \textbf{semi-bounded from below}, or equivalently lower semi-bounded, if it exists a constant $a\geq0$ such that 
	$$\scalar{\Phi}{T\Phi}\geq -a\norm{\Phi}^2\quad\forall\Phi\in\D(T)\;.$$
	We say that $T$ is \textbf{semi-bounded from above}, or equivalently upper semi-bounded, if $-T$ is semi-bounded from below. The operator $T$ is \textbf{positive} if the lower bound satisfies $a=0$\,.
\end{definition}

Notice that, because of Theorem \ref{thesymmcases}, closed, symmetric, semi-bounded operators are automatically self-adjoint.


\section{Closable quadratic forms}

In this section we introduce the notion of closed and closable quadratic forms. Standard references are,
e.g.,~\cite[Chapter~VI]{kato95}, \cite[Section~VIII.6]{reed72} or \cite[Section~4.4]{davies95}.

\begin{definition}
Let $\D$ be a dense subspace of the Hilbert space $\H$ and denote by $Q\colon\mathcal{D}\times\mathcal{D}\to \mathbb{C}$
a sesquilinear form (anti-linear in the first entry and linear in the second entry). The quadratic form associated to $Q$
with domain $\mathcal D$ is its evaluation on the diagonal, i.e., $Q(\Phi):=Q(\Phi,\Phi)$\,,
$\Phi\in\mathcal{D}$\,. We say that the sesquilinear form is \textbf{Hermitean} if
\[
  Q(\Phi,\Psi)=\overline{Q(\Psi,\Phi)}\;,\quad \Phi,\Psi\in\mathcal{D}\;.
\]
The quadratic form is  \textbf{semi-bounded from below} if there is an $a\geq 0$ such that
\[
  Q(\Phi)\geq -a \norm{\Phi}^2\;,\; \Phi \in \mathcal{D}\;.
\]
The smallest possible value $a$ satisfying the preceding inequality is called the \textbf{lower bound} for the quadratic form $Q$. ç
In particular, if  $Q(\Phi)\geq 0$ for all $\Phi\in\mathcal{D}$ we say that $Q$ is \textbf{positive}.

\end{definition}

Note that if $Q$ is semi-bounded with lower bound $a$\,, then $Q_a(\Phi):=Q(\Phi)+a\norm{\Phi}^2$\,, $\Phi\in\D$\,,
is positive on the same domain. Moreover, the polarisation identity holds. 
\begin{equation}\label{polarization id}
	Q(\Phi,\Psi)=\frac{1}{4}\left[Q(\Phi+\Psi)-Q(\Phi-\Psi)+\mathbf{i}Q(\mathbf{i}\Phi+\Psi)-\mathbf{i}Q(\mathbf{i}\Phi-\Psi)\right]\;.
\end{equation}
We need to recall also the notions of closable and closed quadratic forms as
well as the fundamental representation theorems that relate closed, semi-bounded quadratic forms with self-adjoint, semi-bounded operators.

\begin{definition}
Let $Q$ be a semi-bounded quadratic form with lower bound $a\geq 0$ and dense domain $\D\subset\H$.
The quadratic form $Q$ is \textbf{closed} if $\D$ is closed with respect to the norm
\[
 \normm{\Phi}_Q:=\sqrt{Q(\Phi)+(1+a)\|\Phi\|^2}\;,\quad\Phi\in\D\;.\newnot{symb:graphnormQ}
\]
If Q is closed and $\D_0\subset\D$ is dense with respect to the norm $\normm{\cdot}_Q$\,, then $\D_0$ is called a
\textbf{form core} for $Q$.

Conversely, the closed quadratic form $Q$ with domain $\D$ is called an
\textbf{extension} of the quadratic form $Q$ with domain $\D_0$. A quadratic form is said to be
\textbf{closable} if it has a closed extension.
\end{definition}

\begin{remark}$\phantom{=}$\label{Remclosable}
\begin{enumerate}
\item The norm $\normm{\cdot}_Q$ is induced by the following inner product on
the domain:
\[
 \langle\Phi,\Psi\rangle_Q:= Q(\Phi,\Psi)+(1+a)\langle\Phi,\Psi\rangle\;,\quad \Phi,\Psi\in\D\;.
\]
 \item The quadratic form $Q$ is closable iff whenever a sequence $\{\Phi_n\}_n\subset\D$ satisfies
$\norm{\Phi_n}\to 0$ and $Q(\Phi_n-\Phi_m)\to 0$\,, as $n,m\to\infty$\,, then $Q(\Phi_n)\to 0$.
 \item In general it is always possible to close $\D \subset\mathcal{H}$ with respect to the norm $\normm{\cdot}_Q$.
The quadratic form is closable iff this closure is a subspace of $\H$.
 \end{enumerate}

\end{remark}

\begin{theorem}[Kato's representation theorem]\label{fundteo}
Let $Q$ be an Hermitean, closed, semi-bounded quadratic form defined on the dense domain
$\D\subset\H$. Then it exists a unique self-adjoint, semi-bounded operator $T$
with domain $\D(T)$ and the same lower bound such that
\begin{enumerate}
\item $\Psi\in\mathcal{D}(T)$ iff $\Psi\in \D$ and it exists $\chi \in \H$ such that
$$Q(\Phi,\Psi)=\langle\Phi,\chi\rangle\,,\quad\forall \Phi\in\D\;.$$
In this case we write $T\Psi=\chi$.
\item $Q(\Phi,\Psi)=\langle\Phi,T\Psi\rangle$ for any $\Phi\in\D$\,,\;$\Psi\in\D(T)$.
\item $\D(T)$ is a core for $Q$.
\end{enumerate}
\end{theorem}

Notice that in the same way as in Definition \ref{adjointoperator}, the element $\chi$ of i) is uniquely determined because $\D$ is dense in $\H$\,.

\begin{remark}\label{rem:represntability}
	A closed, semi-bounded, quadratic form can be represented using the spectral resolution of the identity $E_\lambda$ of the associated self-adjoint operator, cf., Theorem \ref{spectraltheorem},
		$$Q(\Phi,\Psi)=\int_{\mathbb{R}}\lambda\d\scalar{\Phi}{E_\lambda\Psi}\quad\forall\Phi,\Psi\in\D(T)\;.$$
\end{remark}

One of the most common uses of the representation theorem is to obtain self-adjoint extensions of symmetric, semi-bounded operators. Given a semi-bounded, symmetric operator $T$ one can consider the associated quadratic form $$Q_T(\Phi,\Psi)=\scalar{\Phi}{T\Psi}\quad \Phi,\Psi\in\D(T)\;.$$ These quadratic forms are always closable, cf., \cite[Theorem X.23]{reed75}, and therefore their closure is associated to a unique self-adjoint operator. Even if the symmetric operator has infinite possible self-adjoint extensions, the representation theorem allows to select a particular one. This extension is called the Friedrichs extension. The approach that we shall take in Chapter \ref{cha:QF} is close to this method.\\

Before we close this section we introduce one important result that is encompassed as one of the tools of what is commonly know as the variational methods, cf., \cite[Theorem~XIII.2]{reed78}. 

\begin{theorem}[min-max Principle]\label{minmaxprinciple}
	Let $Q$ be a closed, semi-bounded quadratic form with domain $\D$, let $T$ be the associated self-adjoint operator, and let $V_n$ be the subspace $$V_n=\{\Phi\in\D\bigr|\scalar{\Phi}{\xi_i}=0\,,\xi_i\in\H^0(\Omega)\,,i=1,\dots,n\}\;.$$ Define $$\lambda^{(n)}=\sup_{\xi_1,...,\xi_{n-1}}\biggl[\inf_{\Phi\in V_{n-1}} \frac{Q(\Phi,\Phi)}{\norm{\Phi}^2}\biggr]\;.$$ Then, for each fixed $n$, either:
	\begin{itemize}
		\item[a)] there are $n$ eigenvalues (counting degenerate eigenvalues a number of times equal to their multiplicity) below the bottom of the essential spectrum and $\lambda^{(n)}$ is the $n$th eigenvalue counting multiplicity;
	\end{itemize}		
		or
	\begin{itemize}	
		\item[b)] $\lambda^{(n)}$ is the bottom of the essential spectrum, i.e., $$\lambda^{(n)}=\inf\{\lambda\mid\lambda\in\sigma_{\mathrm{ess}}(T)\}$$ and $\lambda^{(n)}=\lambda^{(n+1)}=\lambda^{(n+2)}=\cdots$ and there are at most $n-1$ eigenvalues (counting multiplicity) below $\lambda^{(n)}$.
	\end{itemize}
\end{theorem}

We shall use this result to proof the convergence of the numerical scheme introduced in Chapter~\ref{cha:FEM}.


\section{Laplace-Beltrami operator on Riemannian manifolds and Sobolev spaces}

Our aim is to describe a class of closable quadratic forms related to the self-adjoint extensions of the Laplace-Beltrami operator defined
on a compact Riemannian manifold. This section is devoted to the definition of such manifold and of the different spaces of functions that will appear throughout the rest of this dissertation. Any further details can be found at \cite[Chapter 3 et seq.]{marsden01}, \cite[Chapter 3 et seq.]{adams03}, \cite[Section 3.7]{davies95}, \cite[Chapter 1]{lions72}, \cite[Chapter 3]{poor81}.

Let $(\Omega,\pO,\eta)$\newnot{symb:omega} be a smooth, orientable, compact Riemannian manifold  with metric
$\eta$\newnot{symb:metric} and smooth boundary $\partial \Omega$\,.\newnot{symb:po}
We will denote as  $\C^\infty (\Omega)$\newnot{symb:smooth} the space of smooth functions of the Riemannian manifold $\Omega$ and
by $\C_c^\infty (\Omega)$\newnot{symb:smooth compact} the space of smooth functions with compact support in the interior of $\Omega$.
The Riemannian volume form is written as $\d\mu_\eta$\,\newnot{symb:volume}.

\begin{definition}
The \textbf{Laplace-Beltrami Operator} associated to the Riemannian manifold $(\Omega,\pO,\eta)$ is the
second order differential operator $\Delta_\eta:\C^\infty(\Omega)\to\C^\infty(\Omega)$ given by
$$\Delta_\eta\Phi=\frac{1}{\sqrt{|\eta|}}\frac{\partial}{\partial x�}\sqrt{|\eta|}\eta^{ij}\frac{\partial\Phi}{\partial x^j}\;.\newnot{symb:laplacian}$$
\end{definition}

Let $(\tilde{\Omega},\tilde{\eta})$\newnot{symb:omeganotboundary} be a smooth, orientable, boundary-less, compact Riemannian manifold with metric $\widetilde{\eta}$.
The Laplace-Beltrami operator $-\Delta_{\tilde{\eta}}$ associated to the Riemannian manifold $(\tilde{\Omega},\tilde{\eta})$ defines a positive,
essentially self-adjoint, second order differential operator, cf., \cite{marsden01}. One can use it to define the following norms.

\begin{definition} \label{DefSobolev}
Let $k\in\mathbb{R}$\,. The \textbf{Sobolev norm of order} $k$ in the boundary-less Riemannian manifold $(\tilde{\Omega},\tilde{\eta})$ is defined by
$$|| \Phi ||_k^2 := \int_{\tilde{\Omega}} \overline{\Phi} (I - \Delta_{\tilde{\eta}} )^{k}\Phi \d\mu_{\tilde{\eta}}\;. $$
The closure of the smooth functions with respect to this norm $\H^k(\tilde{\Omega}) := \overline{\C^{\infty}(\tilde{\Omega})}^{\norm{\cdot}_k}$ is the \textbf{Sobolev space of class $k$} of the Riemannian manifold $(\tilde{\Omega},\tilde{\eta})$\,.
The scalar products associated to these norms are written as $\scalar{\cdot}{\cdot}_k$.
In the case $k=0$ we will denote the $\H^0(\tilde{\Omega})$ scalar product simply by
$\scalar{\Phi}{\Psi}=\int_{\tilde{\Omega}} \overline{\Phi}\Psi \d\mu_{\tilde{\eta}}$.
\end{definition}

Note that Definition~\ref{DefSobolev} holds only for Riemannian manifolds without boundary. In particular, it holds for the Sobolev spaces defined over the boundary $\pO$\,, $\H^k(\pO)$\newnot{symb:sobolevpo}\,, of the Riemannian manifold $\Omega$\,.
The construction of the Sobolev spaces of functions over a manifold $(\Omega,\pO,\eta)$
cannot be done directly like in the definition above because the Laplace-Beltrami operator does not define in general
positive differential operators. However, it is possible to construct it as a quotient of the Sobolev space of functions
over a Riemannian manifold $(\tilde{\Omega},\tilde{\eta})$ without boundary, cf., \cite[Section 4.4]{taylor96}.

\begin{definition}\label{DefSobolev2}
Let $(\Omega,\pO,\eta)$ be a Riemannian manifold and let $(\tilde{\Omega},\tilde{\eta})$ be any smooth Riemannian manifold without boundary such that
$\smash{\overset{\scriptscriptstyle\circ}{\Omega}}$\newnot{symb:intOmega}\,, i.e., the interior of $\Omega$\,, is an open submanifold of
$\tilde{\Omega}$. Let $k\in\mathbb{R}$. The \textbf{Sobolev space of class} $k$ of the Riemannian manifold $(\Omega,\pO,\eta)$ is the quotient
$$\H^k(\Omega):=\H^k(\tilde{\Omega})/\{\Phi\in\tilde{\Omega}\mid \Phi|_\Omega=0\}\;. \newnot{symb:sobolevomega}$$
The norm is denoted again as $\norm{\cdot}_k$. When there is ambiguity about the manifold, the subindex shall denote the full space, i.e.,
$$\norm{\cdot}_k=\norm{\cdot}_{\H^k(\Omega)}\;.\newnot{symb:normk}\newnot{symb:normkomega}$$
\end{definition}

\begin{remark}\quad
The Sobolev spaces $\H^k(\Omega)$ do not depend on the choice of $\;\tilde{\Omega}$\,.
\end{remark}

We shall need the following subspaces of the the Sobolev spaces.

\begin{definition}
	Let $(\Omega,\eta)$ be any smooth Riemannian manifold with or without boundary. The closure of the set of smooth functions with compact support in the interior of $\Omega$ with respect to the Sobolev norm of order $k\in\mathbb{R}$ is denoted as $\H^k_0(\Omega)$\,, i.e., $$\H^k_0(\Omega):=\overline{\C_c^{\infty}(\Omega)}^{\norm{\cdot}_k}\;.\newnot{symb:sobolevomega0}$$
\end{definition}

There are many equivalent ways to define the Sobolev norms. In particular, we shall need the following characterisation.

\begin{proposition}\label{equivalentsobolev}
	The Sobolev norm of order 1, $\norm{\cdot}_1$\,, is equivalent to the norm $$\sqrt{\norm{\d\cdot}^2_{\Lambda^1}+\norm{\cdot}^2}\;,$$
	where $\d$\newnot{symb:d} stands for the exterior differential acting on functions and $\norm{\d\cdot}_{\Lambda^1}$ is the induced norm from the natural scalar product among 1-forms $\alpha\in\Lambda^1(\Omega)$\newnot{symb:oneforms}, cf., \cite[Chapter 6]{marsden01}.
\end{proposition}

\begin{proof}
	It is enough to show it for a boundary-less Riemannian manifold $(\tilde{\Omega},\tilde{\eta})\,.$ The Laplace-Beltrami operator can be expressed in terms of the exterior differential and its formal adjoint,
	$$-\Delta_{\tilde{\eta}}=\d^\dagger\d\;,$$
	where the formal adjoint is defined to be the unique differential operator $\d^\dagger:\Lambda^1(\tilde{\Omega})\to\C^\infty(\tilde{\Omega})$  that verifies $$\scalar{\alpha}{\d\Phi}_{\Lambda^1}=\scalar{\d^\dagger\alpha}{\Phi}\quad \alpha\in\Lambda^1(\tilde{\Omega}),\Phi\in\C^\infty(\tilde{\Omega})\;.$$ Let $\Phi\in\C^{\infty}(\tilde{\Omega})$. Then we have that
	\begin{align*}
		\norm{\Phi}^2_1&=\int_{\tilde{\Omega}}\bar{\Phi}(I-\Delta_{\tilde{\eta}})\Phi\d\mu_{\tilde{\eta}}\\
			&=\int_{\tilde{\Omega}}\bar{\Phi}\Phi\d\mu_{\tilde{\eta}}+ \int_{\tilde{\Omega}}\bar{\Phi}\d^\dagger\d\Phi\d\mu_{\tilde{\eta}}\\
			&=\norm{\Phi}^2+\scalar{\d\Phi}{\d\Phi}_{\Lambda^1}=\norm{\Phi}^2+\norm{\d\Phi}^2_{\Lambda^1}\;.
	\end{align*}
\end{proof}

The subindex $\Lambda^1$ will be omitted when it is clear from the context which scalar products are considered.\\

The boundary $\pO$ of the Riemannian manifold $(\Omega,\pO,\eta)$ has itself the structure of a Riemannian manifold
without boundary $(\pO,\partial\eta)$. The Riemannian metric induced at the boundary is just the pull-back of the
Riemannian metric $\partial\eta=i^\star\eta$\newnot{symb:partialeta}\,, where $i:\pO\to\Omega$ is the inclusion map. The spaces of smooth functions
over the two manifolds verify that $\C^\infty(\Omega)\bigr|_{\pO}\simeq\C^{\infty}(\pO)$.

There is an important relation between the Sobolev spaces defined over the manifolds $\Omega$ and $\pO$ (see Definition \ref{DefSobolev} and Definition \ref{DefSobolev2}).
This is the well known Lions trace theorem (cf., \cite[Theorem 7.39]{adams03}, \cite[Theorem 9.4 of Chapter 1]{lions72}):

\begin{theorem}[Lions trace theorem]\label{LMtracetheorem}
Let $\Phi\in\C^{\infty}(\Omega)$ and let $\gamma:\C^\infty(\Omega)\to\C^{\infty}(\pO)$ be the trace map $\gamma(\Phi)=\Phi\bigr|_{\pO}$\newnot{symb:trace}. There is a unique continuous extension of the trace map such that
\begin{enumerate}
\item $\gamma:\H^{k}(\Omega)\to\H^{k-1/2}(\pO)$\,, $k > 1/2$\;.\\
\item The map is surjective\;.\\
\end{enumerate}
\end{theorem}

Finally we introduce for later use some particular operators associated to the Laplacian. Consider the symmetric operator $\Delta_0:=\Delta_\eta\bigr|_{\C^\infty_c(\Omega)}$\newnot{symb:Delta0}. Then we have the following extensions of it.

\begin{definition}\quad\newline
\begin{enumerate}
\item The \textbf{minimal closed extension} $\Delta_{\mathrm{min}}$\newnot{symb:Deltamin} is defined to be the closure of $\Delta_0$.
Its domain is $\D(\Delta_{\mathrm{min}})=\H^2_0:=\overline{\C^\infty_c(\Omega)}^{\norm{\cdot}_2}$\,.
\item The \textbf{maximal closed extension} $\Delta_{\mathrm{max}}$\newnot{symb:Deltamax} is the closed operator defined in the domain
$\D(\Delta_{\mathrm{max}})=\bigl\{ \Phi\in\H^0(\Omega)\bigr| \Delta_\eta\Phi\in\H^0(\Omega) \bigr\}$\,.\\
\end{enumerate}
\end{definition}

The trace map defined in Theorem \ref{LMtracetheorem} can be extended continuously to $\D(\Delta_{\mathrm{max}})$\,,
see for instance \cite{frey05,grubb68,lions72}:


\begin{theorem}[Weak trace theorem for the Laplacian]\label{weaktracetheorem}
The Sobolev space $\H^k(\Omega)$\,, with $k\geq2$\,, is dense in $\mathcal{D}(\Delta_{\mathrm{max}})$ and
there is a unique continuous extension of the trace map $\gamma$ such that
$$\gamma \colon \mathcal{D}(\Delta_{\mathrm{max}}) \to H^{-1/2}(\pO)\;.$$
Moreover, $\ker \gamma = H_0^2(\Omega)$\,.
\end{theorem}


\section{Scales of Hilbert spaces}

Later on we will need the theory of scales of Hilbert spaces, also known as theory of rigged Hilbert spaces.
In the following paragraph we state the main results, (see \cite[Chapter I]{berezanskii68}, \cite[Chapter 2]{koshmanenko99} for proofs and more results).

Let $\H$ be a Hilbert space with scalar product $\scalar{\cdot}{\cdot}$ and induced norm $\norm{\cdot}$. Let $\H_+$ be a dense linear
subspace of $\H$ which is a complete Hilbert space with respect to another scalar product that will be denoted by $\scalar{\cdot}{\cdot}_+$.
The corresponding norm is $\norm{\cdot}_+$ and we assume that
\begin{equation}\label{inclusion inequality}
	\norm{\Phi}\leq\norm{\Phi}_+\;,\quad \Phi\in\H_+\;.
\end{equation}

Any vector $\Phi\in\H$ generates a continuous linear functional $L_\Phi\colon\H_+\to \mathbb{C}$  as follows. For $\Psi\in\H_+$ define
\begin{equation}
	L_{\Phi}(\Psi)=\scalar{\Phi}{\Psi}\;.
\end{equation}
Continuity follows by the Cauchy-Schwartz inequality and Eq.\ \eqref{inclusion inequality}.
\begin{equation}
L_{\Phi}(\Psi) \leq \norm{\Phi}\cdot\norm{\Psi}\leq\norm{\Phi}\cdot\norm{\Psi}_+\;,
                    \quad\forall \Phi\in\H\;,\forall\Psi\in\H_+\;.	
\end{equation}
Since $L_\Phi$ represents a continuous linear functional on $\H_+$ it can be represented, according to Riesz theorem,
using the scalar product in $\H_+$. Namely, it exists  a vector $\xi\in\H_+$ such that
\begin{equation}
	\forall\Psi\in\H_+\,,\quad L_{\Phi}(\Psi)=\scalar{\Phi}{\Psi}=\scalar{\xi}{\Psi}_+\;,
\end{equation}
and the norm of the functional coincides with the norm in $\H_+$ of the element $\xi$\,, i.e.,
 $$\norm{L_\Phi}=\sup_{\Psi\in\H_+}\frac{|L_\Phi(\Psi)|}{\norm{\Psi}_+}=\norm{\xi}_+\;.$$
One can use the above equalities to define an operator
\begin{equation}
\begin{array}{c}
\hat{I}\colon \H\to\H_+\\
\hat{I}\Phi=\xi\;.
\end{array}
\end{equation}
This operator is clearly injective since $\H_+$ is a dense subset of $\H$ and therefore it can be used to define a new scalar product on $\H$
\begin{equation}
	\scalar{\cdot}{\cdot}_-:=\scalar{\hat{I}\cdot}{\hat{I}\cdot}_+\;.
\end{equation}
The completion of $\H$ with respect to this scalar product defines a new Hilbert space, $\H_-$\,, and the corresponding norm will be
denoted accordingly by $\norm{\cdot}_-$\,. It is clear that $\H_+\subset\H\subset \H_-$\,, with dense inclusions.
Since $\norm{\xi}_+=\norm{\hat{I}\Phi}_+=\norm{\Phi}_-$\,, the operator $\hat{I}$ can be extended by continuity to an isometric bijection.

\begin{definition}\label{def:scales}
The Hilbert spaces $\H_+$\,, $\H$ and $\H_-$ introduced above define a \textbf{scale of Hilbert spaces}. The extension by continuity of the
operator $\hat{I}$ is called the \textbf{canonical isometric bijection}. It is denoted by:
\begin{equation}
I\colon \H_-\to\H_+\;.\newnot{symb:I}
\end{equation}
\end{definition}

\begin{proposition}\label{proppairing}
The scalar product in $\H$ can be extended continuously to a pairing
\begin{equation}
\pair{\cdot}{\cdot}\colon\H_-\times\H_+\to\mathbb{C}\;.\newnot{symb:pair}
\end{equation}
\end{proposition}

\begin{proof}
Let $\Phi\in\H$ and $\Psi\in\H_+$. Using the Cauchy-Schwartz inequality we have the following
\begin{equation}\label{CSpairing}
	|\scalar{\Phi}{\Psi}|=|\scalar{I\Phi}{\Psi}_+|\leq \norm{I\Phi}_+\norm{\Psi}_+=\norm{\Phi}_-\norm{\Psi}_+\;.
\end{equation}
\end{proof}

\clearemptydoublepage
\chapter{Closable and Semi-bounded Quadratic Forms Associated to the Laplace-Beltrami Operator}
\markboth{Closable and Semi-bounded Quadratic Forms}{}
\label{cha:QF}

In this chapter we construct a wide class of closed, semi-bounded quadratic forms on the space of square integrable functions over a smooth Riemannian manifold with smooth boundary. Each of these quadratic forms specifies a semi-bounded self-adjoint extension of the Laplace-Beltrami operator. These quadratic forms are based on the Lagrange boundary form on the manifold and a family of domains parameterised by a suitable class of unitary operators on the Hilbert space of the boundary that will be called admissible.

There are previous results characterising the lower semi-boundedness of sesqui\-linear forms, see for instance \cite{grubb70}. However, we exploit the particularly simple form of the boundary condition \eqref{intro: asorey} in order to provide a direct characterisation of a class of self-adjoint and semi-bounded extensions of the Laplace-Beltrami operator. This class of semi-bounded extensions, as will be shown in Section \ref{QF: Examples}, covers a number of relevant known examples that can be handled easily in this manner.

In Section \ref{sec:class} we introduce the aforementioned class of quadratic forms associated to the Laplace-Beltrami operator. We give the main definitions and analyse thoroughly the structure of the boundary equation \eqref{intro: asorey}. In Section \ref{sec:closable and semibounded qf} we prove the main results of this chapter. Namely, those concerning the semi-boundedness and closability of the class of quadratic forms. The last section, Section \ref{QF: Examples}, is devoted to introduce a number of meaningful examples.


\section{A class of closable quadratic forms on a Riemannian manifold}\label{sec:class}

We begin presenting a canonical sesquilinear form that, on smooth functions over the Riemannian manifold $\Omega$\,, is associated to
the Laplace-Beltrami operator. Motivated by this quadratic form we will address questions like hermiticity, closability
and semi-boundedness on suitable domains.

Integrating once by parts the expression
$\scalar{\Phi}{-\Delta_\eta\Psi}$ we obtain, on smooth functions, the following sesquilinear form
$Q \colon \C^{\infty}(\Omega)\times\C^{\infty}(\Omega)\to\mathbb{C}$\,,
\begin{equation} \label{Q-def}
Q(\Phi,\Psi)=\scalar{\d\Phi}{\d\Psi}_{\Lambda^1}-\scalarb{\varphi}{\dot{\psi}}\;.
\end{equation}

From now on the restrictions to the boundary are going to be denoted with the corresponding small size greek letters, $\varphi:=\gamma(\Phi)$\,.\newnot{symb:varphi}
The doted small size greek letters denote the restriction to the boundary of the normal derivatives, $\dot{\varphi}:=\gamma(\d\Phi(\nu))$\,,\newnot{symb:dotvarphi}
where $\nu\in\mathfrak{X}(\Omega)$\newnot{symb:vectorfield} is any vector field such that $\mathrm{i}_\nu\d\mu_\eta=\d\mu_{\partial\eta}$\,.
Notice that in the expression above $\d \Phi\in \Lambda^1(\Omega)$ is a 1-form on $\Omega$\,, thus the inner product $\langle \cdot, \cdot \rangle_{\Lambda^1}$
is defined accordingly by using the induced Hermitean structure on the cotangent bundle (see, e.g., \cite{poor81}).
We have therefore that $$\scalar{\d\Phi}{\d\Psi}_{\Lambda^1}= \int_\Omega \eta^{-1}(\d \bar{\Phi}, \d\Psi) \diff \mu_\eta\;.$$
In the second term at the right hand side of \eqref{Q-def} $\scalarb{\cdot}{\cdot}$ stands for the induced scalar product at the boundary
given explicitly by
\begin{equation}\label{scalarboundary}
\scalarb{\varphi}{\psi}=\int_{\pO }\bar{\varphi}\,  \psi \,\diff \mu_{\partial \eta}  ,
\end{equation}\newnot{symb:scalarb}
where $\diff\mu_{\partial \eta}$ is the Riemannian volume defined by the restricted Riemannian metric $\partial\eta$\,. The subscript $\Lambda^1$ will be dropped from now on as along as there is no risk of confusion.\\\

In general, the sesquilinear form $Q$ defined above is not Hermitean. To study subspaces where $Q$ is Hermitean it is convenient to isolate the part of $Q$ related to the boundary data $(\varphi,\dot{\varphi})$\,.

\begin{definition}\label{def: lagrangebf}
Let $\Phi,\Psi\in \C^{\infty}(\Omega)$ and denote by $(\varphi,\dot{\varphi})$\,, $(\psi,\dot{\psi})$ the corresponding boundary data. The \textbf{Lagrange boundary form} is defined as:
\begin{equation}\label{lagrange}
	\Sigma\bigl(\Phi,\Psi\bigr)=\Sigma\bigl((\varphi,\dot{\varphi}),(\psi,\dot{\psi})\bigr):=\scalarb{\varphi}{\dot{\psi}}-\scalarb{\dot{\varphi}}{\psi}.
\end{equation}
Any dense subspace $\mathcal{D}\subset\H^0(\Omega)$ is said to be \textbf{isotropic with respect to $\Sigma$} if $\Sigma\bigl(\Phi,\Psi\bigr)=0\quad\forall\Phi,\Psi\in \D$\,.
\end{definition}

\begin{proposition}
The sesquilinear form $Q$ defined in Eq.~(\ref{Q-def}) on a dense subspace $\mathcal{D}\subset \H^{0}$ is Hermitian iff $\mathcal{D}$ is isotropic with respect to $\Sigma$\,.
\end{proposition}
\begin{proof}
The sesquilinear form $Q\colon\mathcal{D}\times\mathcal{D}\to\mathbb{C}$ is Hermitean if $Q(\Phi,\Psi)=\overline{Q(\Psi,\Phi)}$ for all $\Phi,\Psi\in\mathcal{D}$\,. By definition of $Q$ this is equivalent to $\Sigma\bigl(\Phi,\Psi\bigr)=0$\,, for all $\Phi,\Psi\in\mathcal{D}$\,, hence $\mathcal{D}$ is isotropic with respect to $\Sigma$\,. The reverse implication is obvious.
\end{proof}


\subsection{Isotropic subspaces}
The analysis of maximally isotropic subspaces can be handled more easily using the underlying Hilbert space structure of the Lagrange boundary form
and not considering for the moment any regularity question.
The expression \eqref{lagrange} can be understood as a sesquilinear form on the boundary Hilbert space $\H_b:=\H^0(\pO)\times \H^0(\pO)$\,,
\begin{equation*}
\Sigma\bigl(\Psi,\Phi\bigr)
	=\scalarb{\varphi}{\dot{\psi}}-\scalarb{\dot{\varphi}}{\psi}\;.\\
\end{equation*}
We will therefore focus now on the study of the sesquilinear form on the Hilbert space $\H_b$ and, while there is no risk of confusion, we will denote the scalar product in $\H^0(\pO)$ simply as $\scalar{\cdot}{\cdot}$\,,
\[
	\Sigma\left((\varphi_1,\varphi_2),(\psi_1,\psi_2)\right)
                :=\langle\varphi_1,\psi_2\rangle - \langle\varphi_2,\psi_1\rangle\;,\quad
	(\varphi_1,\varphi_2),(\psi_1,\psi_2)\in \H_b\,.
\]
Formally, $\Sigma$ is a sesquilinear, symplectic form by which we mean that it satisfies the following conditions:
\begin{enumerate}
	\item $\Sigma$ is conjugate linear in the first argument and linear in the second.
	\item $\Sigma\Big((\varphi_1,\varphi_2),(\psi_1,\psi_2)\Big)
               =-\overline{\Sigma\Big((\psi_1,\psi_2),(\varphi_1,\varphi_2)\Big)}$\,,
                 $(\varphi_1,\varphi_2),(\psi_1,\psi_2)\in\H_b$\,.
	\item $\Sigma$ is non-degenerate, i.e., $\Sigma\Big((\varphi_1,\varphi_2),(\psi_1,\psi_2)\Big)=0$
               for all $(\psi_1,\psi_2)\in\H_b$ implies $(\varphi_1,\varphi_2)=(0,0)$\,.
\end{enumerate}
The analysis of the isotropic subspaces of such sesquilinear forms is by no means new and their characterisation is well known (\cite{bruning08}, \cite{kochubei75}). However, in order to keep this dissertation self-contained, we provide in the following paragraphs independent proofs of the main results that we will need.

First we write the sesquilinear symplectic form $\Sigma$ in diagonal form. This is done introducing the unitary Cayley transformation $\C\colon\H_b\to\H_b$\,,\newnot{symb:Cayley}
\[
 \C(\varphi_1,\varphi_2):=\frac{1}{\sqrt{2}}\left( \varphi_1+\mathbf{i}\varphi_2, \varphi_1-\mathbf{i}\varphi_2 \right)\;,\;
  (\varphi_1,\varphi_2)\in\H_b\;.
\]
Putting
\[ \Sigma_c\Big((\varphi_+,\varphi_-),(\psi_+,\psi_-)\Big)
        :=-\mathbf{i}\Big( \langle\varphi_+,\psi_+\rangle - \langle\varphi_-,\psi_-\rangle\Big)
\]
with $(\varphi_+,\varphi_-),(\psi_+,\psi_-)\in\H_b$\,, the relation between $\Sigma$ and $\Sigma_c$ is given by
\begin{equation}\label{relation}
	\Sigma\Big((\varphi_1,\varphi_2),(\psi_1,\psi_2)\Big)
  	=\Sigma_c\Big(\C(\varphi_1,\varphi_2),\C(\psi_1,\psi_2)\Big)
\end{equation}
with $(\varphi_1,\varphi_2),(\psi_1,\psi_2)\in\H_b$\,.

\begin{definition}
Consider a subspace $\mathcal{W}\subset\H_b$ and define the \textbf{$\Sigma$-orthogonal subspace} by
\[
\mathcal{W}^{\perp_\Sigma}:=\left\{(\varphi_1,\varphi_2)\in\H_b\mid\Sigma\Big((\varphi_1,\varphi_2),(\psi_1,\psi_2)\Big)=0
                  \;,\;\forall(\psi_1,\psi_2)\in \mathcal{W}\right\}\;.
\]
A subspace $\mathcal{W}\subset\H_b$ is \textbf{$\Sigma$-isotropic} [resp.~\textbf{maximally $\Sigma$-isotropic}] if $\mathcal{W}\subset \mathcal{W}^{\perp_\Sigma}$ [resp.~$\mathcal{W}= \mathcal{W}^{\perp_\Sigma}$].
\end{definition}

We begin enumerating some direct consequences of the preceding definitions:
\begin{lemma}
Let $\mathcal{W}\subset\H_b$ and put $\mathcal{W}_c:=\C(\mathcal{W})$\,.
	\begin{enumerate}\label{lemma-iso}
 		\item $\mathcal{W}$ is $\Sigma$-isotropic [resp.~maximally $\Sigma$-isotropic] iff $\mathcal{W}_c$ is $\Sigma_c$-isotropic [resp. maximally $\Sigma_c$-isotropic].
		 \item If $(\varphi_1,\varphi_2)\in \mathcal{W}\subset \mathcal{W}^{\perp_\Sigma}$\,, then $\langle \varphi_1,\varphi_2\rangle=\overline{\langle \varphi_1,\varphi_2\rangle}$\,. If $(\varphi_+,\varphi_-)\in \mathcal{W}_c\subset \mathcal{W}_c^{\perp_{\Sigma_c}}$\,, then $\|\varphi_+\|=\|\varphi_-\|$\,.
	\end{enumerate}
\end{lemma}
\begin{proof}
Part (i) follows directly from Eq.(\ref{relation}) and the fact that $\C$ is a unitary transformation. To prove (ii) note that if $(\varphi_1,\varphi_2)$ is in an isotropic subspace $\mathcal{W}$\,, then
\[
 \Sigma\Big((\varphi_1,\varphi_2),(\varphi_1,\varphi_2)\Big)
         =\langle\varphi_1,\varphi_2\rangle - \langle\varphi_2,\varphi_1\rangle=0\;.
\]
One argues similarly in the other case.
\end{proof}

\begin{proposition}\label{pro:W}
Let $\mathcal{W}_\pm\subset\H^0(\pO)$ be closed subspaces and put $\mathcal{W}_c:=\mathcal{W}_+\times \mathcal{W}_-\subset\H_b$\,.
\begin{enumerate}
	\item The subspace $\mathcal{W}_c$ is $\Sigma_c$-isotropic iff it exists a partial isometry
            $$V\colon \H^0(\pO)\to \H^0(\pO)$$ with initial space $\mathcal{W}_+$
            and final space $\mathcal{W}_-$\,, i.e., $V^*V(\H^0(\pO))=\mathcal{W}_+$ and $VV^*(\H^0(\pO))=\mathcal{W}_-$
            and
\[
 \mathcal{W}_c=\{(\varphi_+,V\varphi_+)\mid \varphi_+\in \mathcal{W}_+\}=\mathrm{gra}\,V\;.
\]
	\item The subspace $\mathcal{W}_c$ is maximally $\Sigma_c$-isotropic iff it exists a unitary $$U\colon \H^0(\pO)\to \H^0(\pO)$$ such that
		\begin{equation}\label{eq:U}
			\mathcal{W}_c=\{(\varphi_+,U\varphi_+)\mid \varphi_+\in \H^0(\pO)\}=\mathrm{gra}\,U\;.
		\end{equation}
\end{enumerate}

\end{proposition}
\begin{proof}
(i) For any $(\varphi_+,\varphi_-)\in \mathcal{W}_c$ we define the mapping
$V\colon \H^0(\pO)\to \H^0(\pO)$ by $V(\varphi_+):=\varphi_-$\,, $\varphi_+\in \mathcal{W}_+$ and $V(\varphi)=0$\,, $\varphi\in \mathcal{W}_+^\perp$\,. Since $\mathcal{W}_c\subset \mathcal{W}_c^{\perp_{\Sigma_c}}$ we have from part (ii) of Lemma~\ref{lemma-iso} that $V$ is a well-defined linear map and a partial isometry. 

The reverse implication is immediate: for any $(\varphi_+,V\varphi_+)\in \mathcal{W}_c$ and $\psi_+\in \H^0(\pO)$ we have
\[
 \Sigma_c\Big((\varphi_+,V\varphi_+),(\psi_+,V\psi_+)\Big)
     =-\mathbf{i}\Big(\langle\varphi_+,\psi_+\rangle- \langle V\varphi_+,V\psi_+\rangle\Big)
     =0\;,
\]
hence $\mathcal{W}_c=\{(\varphi_+,V\varphi_+)\mid \varphi_+\in \mathcal{W}_+\}=\mathrm{gra}\,V$ is $\Sigma_c$-isotropic.\\

(ii) Suppose that $\mathcal{W}_c= \mathcal{W}_c^{\perp_{\Sigma_c}}$\,. By the previous item we have 
$$\mathcal{W}_c=\{(\varphi_+,U\varphi_+)\mid \varphi_+\in \mathcal{W}_+\}$$
 for some partial isometry $U\colon \H^0(\pO)\to \H^0(\pO)$\,. Consider the following decompositions $\H^0(\pO)=\mathcal{W}_+\oplus \mathcal{W}_+^\perp= (U\mathcal{W}_+)\oplus (U\mathcal{W}_+)^\perp$ and note that any $(\varphi_+^\perp,\varphi_-^\perp)\in \mathcal{W}_+^\perp\times (U\mathcal{W}_+)^\perp$ satisfies $(\varphi_+^\perp,\varphi_-^\perp)\in \mathcal{W}_c^{\perp_{\Sigma_c}}$\,. Since $\mathcal{W}_c= \mathcal{W}_c^{\perp_{\Sigma_c}}$ we must have $\varphi_+^\perp=\varphi_-^\perp=0$\,, or, equivalently, $\mathcal{W}_+=\H^0(\pO)=U\mathcal{W}_+$\,, hence $\ker U =\ker U^*=\{0\}$ and $U$ is a unitary map.

To prove the reverse implication consider $\mathcal{W}_c=\{(\varphi_+,U\varphi_+)\mid \varphi_+\in \H^0(\pO)\}$ with $U$ unitary and choose $(\psi_+,\psi_-)\in \mathcal{W}_c^{\perp_{\Sigma_c}}$\,. Then for any $\varphi_+\in \H^0(\pO)$ we have
\begin{align*}
 0&=\Sigma_c\Big((\varphi_+,U\varphi_+),(\psi_+,\psi_-)\Big)\\
  &=-\mathbf{i}\Big(\langle \varphi_+,\psi_+\rangle- \langle U\varphi_+,\psi_-\rangle\Big)\\
  &=-\mathbf{i}\Big(\langle \varphi_+,(\psi_+-U^*\psi_-)\rangle\Big)\;.
\end{align*}
This shows that $\psi_-=U\psi_+$ and hence $(\psi_+,\psi_-)\in \mathcal{W}_c$\,, therefore $\mathcal{W}_c$ is maximally $\Sigma_c$-isotropic.
\end{proof}

The previous analysis allows to characterise finally the $\Sigma$-isotropic subspaces of the boundary Hilbert space $\H_b$\,.
\begin{theorem}\label{teo:parametriceW}
 A closed subspace $\mathcal{W}\subset\H_b$ is maximally $\Sigma$-isotropic iff there exists a unitary $U\colon  \H^0(\pO)\to \H^0(\pO)$ such that
\[
 \mathcal{W}=\left\{\Big((\1+U)\varphi\,,\,-\mathbf{i}(\1- U)\varphi\Big)\mid \varphi\in \H^0(\pO))\right\}\;.
\]
\end{theorem}
\begin{proof}
By Lemma~\ref{lemma-iso}~(i) and Proposition~\ref{pro:W}~(ii) we have that $\mathcal{W}$ is maximally $\Sigma$-isotropic iff $\mathcal{W}=\C^{-1}\mathcal{W}_c$\,, where $\mathcal{W}_c$ is given by Eq.~(\ref{eq:U}).
\end{proof}

\begin{proposition}\label{prop: asorey}
Let $U\colon  \H^0(\pO)\to \H^0(\pO)$ be a unitary operator and consider the maximally isotropic subspace $\mathcal{W}$ given in Theorem~\ref{teo:parametriceW}. Then $\mathcal{W}$ can be rewritten as
		\begin{equation}\label{eq:asorey}
			\mathcal{W}=\Big\{(\varphi_1\,,\,\varphi_2)\in\H_b\mid \varphi_1-\mathbf{i}\varphi_2= U(\varphi_1+\mathbf{i}\varphi_2)\Big\}\;.
		\end{equation}
\end{proposition}

\begin{proof}
Let $\mathcal{W}$ be given as in Theorem~\ref{teo:parametriceW} and let $\mathcal{W}'$ be a subspace defined as in Eq.~\eqref{eq:asorey}. Put $\varphi_1:=(\1+U)\varphi$ and $\varphi_2:=-\mathbf{i}(\1-U)\varphi$\,. Then it is straightforward to verify that $(\varphi_1\,,\,\varphi_2)$ satisfy the relation defining Eq.~ \eqref{eq:asorey} and therefore $\mathcal{W}\subset \mathcal{W}'$\,.

Consider a subspace $\mathcal{W}'$ defined as in Eq.~\eqref{eq:asorey} and let $(\varphi_1\,,\,\varphi_2)\in \mathcal{W}'$\,. Then the following relation holds
\begin{equation}\label{WW'1}
	(1-U)\varphi_1-\mathbf{i}(1+U)\varphi_2=0\;.
\end{equation}
Now consider that $(\varphi_1,\varphi_2)\in \mathcal{W}^\perp$\,. Then for all $\varphi\in\H^0(\pO)$
\begin{align*}
	0	&=\scalar{\varphi_1}{(1+U)\varphi}+\scalar{\varphi_2}{-\mathbf{i}(1-U)\varphi}\\
		&=\scalar{(1+U^*)\varphi_1+\mathbf{i}(1-U^*)\varphi_2}{\varphi}
\end{align*}
and therefore
\begin{equation}\label{WW'2}
	(1+U^*)\varphi_1+\mathbf{i}(1-U^*)\varphi_2=0\;.
\end{equation}
Now we can arrange Eqs.~\eqref{WW'1} and \eqref{WW'2}
\begin{equation}\label{M:equation}
M\begin{pmatrix}
\varphi_1\\\varphi_2
\end{pmatrix}
:=
\begin{pmatrix}
1-U & -\mathbf{i}(1+U)\\1+U^* & \mathbf{i}(1-U^*)
\end{pmatrix}
\begin{pmatrix}
\varphi_1\\\varphi_2
\end{pmatrix}
=0\;,
\end{equation}
where now $M\colon \H_b\to\H_b$\,. But clearly $M$ is a unitary operator so that Eq.~\eqref{M:equation} implies that $(\varphi_1,\varphi_2)=0$ and therefore $\mathcal{W}\oplus \mathcal{W}'^\perp=(\mathcal{W}^\perp\bigcap \mathcal{W}')^\perp=\H_b$\,. This condition together with $\mathcal{W}\subset \mathcal{W}'$ implies $\mathcal{W}=\mathcal{W}'$ because $\mathcal{W}$ is a closed subspace, as it is easy to verify.
\end{proof}


\subsection{Admissible unitaries and closable quadratic forms}

In this subsection we will restrict to a family of unitaries $U\colon  \H^0(\pO)\to \H^0(\pO)$ that will allow us to describe a wide class
of quadratic forms whose Friedrichs' extensions are associated to self-adjoint extensions of the Laplace-Beltrami operator.
\begin{definition}\label{DefGap}
Let $U\colon  \H^0(\pO)\to\H^0(\pO)$ be unitary and denote by $\sigma(U)$\newnot{symb:spectrum} its spectrum. We say that the unitary $U$ on the boundary
\textbf{has gap at $-1$} if one of the following conditions hold:
\begin{enumerate}
\item $\1+U$ is invertible.
\item $-1\in\sigma(U)$ and $-1$ is not an accumulation point of $\sigma(U)$\,.
\end{enumerate}
\end{definition}

\begin{definition}\label{P,boundary}
Let $U$ be a unitary operator acting on $\H^0(\pO)$ with gap at $-1$\,.
Let $E_\lambda$ be the spectral resolution of the identity associated to the unitary $U$\,,
i.e., $$U=\int_{[0,2\pi]}e^{\mathbf{i}\lambda}\d E_\lambda\;.$$
The \textbf{invertibility boundary space} $W$\newnot{symb:W} is defined by
$W=\operatorname{Ran}E^{\bot}_{\pi}\,.$ The orthogonal projection onto $W$ is denoted by $P$\,.\newnot{symb:P}
\end{definition}

\begin{definition}\label{partialCayley}
Let $U$ be a unitary operator acting on $\H^0(\pO)$ with gap at $-1$\,. The \textbf{partial Cayley transform} $A_U:\H^0(\pO)\to W$ is
the operator
\[
A_U:=\mathbf{i}\,P (U-\mathbb{I}) (U+\mathbb{I})^{-1}\;.\newnot{symb:AU}
\]

\end{definition}

\begin{proposition}\label{prop: AU bounded}
The partial Cayley transform is a bounded, self-adjoint operator on $\H^0(\pO)$\,.
\end{proposition}

\begin{proof}
First notice that the operators $P$\,, $U$ and $A_U$ commute. That $A_U$ is bounded is a direct consequence of the Definition~\ref{DefGap},
because the operator $P(\mathbb{I}+U)$ is, under the assumption of gap at -1, an invertible bounded operator on the boundary space $W$\,.
To show that $A_U$ is self-adjoint
consider the spectral resolution of the identity of the operator $U$\,. Since $U$ has gap at $-1$\,,
either $\{e^{\mathbf{i\pi}}\}\not\in\sigma(U)$ or there
exists a neighbourhood $V$ of $\{e^{\mathbf{i}\pi}\}$ such that it does not contain any element of the spectrum $\sigma(U)$ besides
$\{e^{\mathbf{i}\pi}\}$\,. Pick $\delta\in V\cap S^1$\,. Then one can express the operator $A_U$ using the spectral resolution of the identity of the
operator $U$ as
$$A_U=\int_{-\pi+\delta}^{\pi-\delta}\mathbf{i}\frac{e^{\mathbf{i}\lambda}-1}{e^{\mathbf{i}\lambda}+1}\d E_\lambda
=\int_{-\pi+\delta}^{\pi-\delta}-\tan\frac{\lambda}{2}\d E_\lambda\;.$$
Since $\lambda\in[-\pi+\delta,\pi-\delta]$\,, then $\tan\frac{\lambda}{2}\in\mathbb{R}$\,.
Therefore the spectrum of $A_U$ is a subset of the real line, necessary and sufficient condition for a closed, symmetric operator to be self-adjoint.
\end{proof}

We can now introduce the class of closable quadratic forms that was announced at the beginning of this section.

\begin{definition}\label{DefQU}
Let $U$ be a unitary with gap at $-1$\,, $A_U$ the corresponding partial Cayley transform and $\gamma$
the trace map considered in Theorem~\ref{LMtracetheorem}.
The Hermitean quadratic form associated to the unitary $U$ is defined by
$$Q_U(\Phi,\Psi)=\scalar{\d\Phi}{\d\Psi}-\scalarb{\gamma(\Phi)}{A_U\gamma(\Phi)}\;.\newnot{symb:QU}$$
on the domain
$$\D_U=\bigl\{ \Phi\in\H^1(\Omega)\bigr|P^{\bot}\gamma(\Phi)=0 \bigr\}\;.\newnot{symb:DU}$$
\end{definition}

\begin{proposition}\label{H1bound}
The quadratic form $Q_U$ is bounded by the Sobolev norm of order 1, $$Q_U(\Phi,\Psi)\leq K\norm{\Phi}_1\norm{\Psi}_1\;.$$
\end{proposition}

\begin{proof}
That the first summand of $Q_U$
is bounded by the $\H^1(\Omega)$ norm is a direct consequence of the Cauchy-Schwartz inequality and
Proposition~\ref{equivalentsobolev}.

For the second term we have that
\begin{align*}
|\scalarb{\gamma(\Phi)}{A_U\gamma(\Psi)}|	&\leq \norm{A_U}\cdot \norm{\gamma(\Phi)}_0\,\norm{\gamma(\Psi)}_0\\
&\leq C \norm{A_U}\cdot\norm{\gamma(\Phi)}_{\frac{1}{2}}\,\norm{\gamma(\Psi)}_{\frac{1}{2}}\\
&\leq C' \norm{A_U}\cdot\norm{\Phi}_{1}\norm{\Psi}_{1}\;,
\end{align*}
where we have used Theorem~\ref{LMtracetheorem} and Proposition~\ref{prop: AU bounded} in the last inequality.
\end{proof}

Finally, we need an additional condition of admissibility on the unitaries on the boundary that will be needed
to prove the closability of $Q_U$\,.

\begin{definition}\label{def:admissible}
Let $U$ be a unitary with gap at $-1$\,. The unitary is said to be \textbf{admissible} if the partial Cayley transform
$A_U\colon\H^0(\pO)\to \H^0(\pO)$ is continuous with respect to the Sobolev norm of order $1/2$\,, i.e.,
$$\norm{A_U\varphi}_{\H^{1/2}(\pO)}\leq K \norm{\varphi}_{\H^{1/2}(\pO)}\;.$$
\end{definition}

\begin{example}
Consider a manifold with boundary given by the unit circle, i.e., $\partial\Omega=S^1$\,,
and define the unitary $(U_\beta\varphi)(z):=e^{i\beta(z)}\,\varphi(z)$\,, $\varphi\in L^2(S^1)$\,.
If $\beta\in L^2(S^1)$ and $\ran\beta\subset\{\pi\}\cup [0,\pi-\delta]\cup [\pi+\delta,2\pi)$\,, for some $\delta >0$\,, then $U_\beta$ has gap
at $-1$\,. If, in addition, $\beta\in C^\infty(S^1)$\,, then $U_\beta$ is admissible.
\end{example}


\section{Closable and semi-bounded quadratic forms}\label{sec:closable and semibounded qf}

This section addresses the questions of semi-boundedness and closability of
the quadratic form $Q_U$ defined on its domain $\D_U$
(cf., Definition~\ref{DefQU}).

\subsection{Functions and operators on collar neighbourhoods}

We will need first some technical results that refer to the functions and operators in a collar neighbourhood
of the boundary $\pO$ and that we will denote by $\Xi$\,.\newnot{symb:collar}
Recall the conventions at the beginning of Section~\ref{sec:class}: if $\Phi\in\H^1(\Omega)$\,,
then $\varphi=\gamma(\Phi)$ denotes its restriction to $\pO$ and for $\Phi$ smooth, $\dot{\varphi}=\gamma(\d\Phi(\nu))$\, is the restriction to
the boundary of the normal derivative.

\begin{lemma}\label{Lemma approxdotphi}
Let $\Phi\in\H^1(\Omega)$\,, $f\in\H^{1/2}(\pO)$\,. Then, for every $\epsilon>0$ it exists $\tilde{\Phi}\in\C^\infty(\Omega)$ such that
\begin{itemize}
\item $\norm{\Phi-\tilde{\Phi}}_1<\epsilon$\,,
\item $\norm{\varphi-\tilde{\varphi}}_{\H^{1/2}(\pO)}<\epsilon$\,,
\item $\norm{f-\dot{\tilde{\varphi}}}_{\H^{1/2}(\pO)}<\epsilon$\,.
\end{itemize}
\end{lemma}

\begin{proof}
The first two inequalities are standard (cf., Theorem~\ref{LMtracetheorem}). Moreover, it
is enough to consider $\Phi\in\C^{\infty}(\Omega)$ with
$\d\Phi(\nu)\equiv0$\,, where $\nu\in\mathfrak{X}(\Omega)$ is the normal vector field on a collar neighbourhood $\Xi$ of $\pO$\,,
(see \cite[Chapter 4]{hirsch76} for details on such neighbourhoods).
According to the proof of \cite[Theorem 7.2.1]{davies95} this is a dense subset of $\H^1(\Omega)$\,.
The compactness assumption of $\Omega$ assures that the collar neighbourhood has a minimal width $\delta$\,. Without loss of
generality we can consider that the collar neighbourhood $\Xi$ has gaussian coordinates $\mathbf{x}=(r,\bt)$\,, being $\frac{\partial}{\partial r}$
the normal vector field pointing outwards. In particular,
we have that  $\Xi\simeq[-\delta,0]\times\pO$ and $\pO\simeq\{0\}\times\pO$\,. It is enough to consider $f\in\H^{1}(\pO)$\,, because $\H^1(\pO)$ is dense in $\H^{1/2}(\pO)$\,.

Consider a smooth function $g\in\C^\infty(\mathbb{R})$ with the following properties:
\begin{itemize}
\item $g(0)=1$ and $g'(0)=-1$\,.
\item $g(s)\equiv0$\,, $s\in[2,\infty)$\,.
\item $|g(s)|<1$ and $|g'(s)|<1$\,,\quad$s\in\mathbb{R}^+_0$\,.
\end{itemize}

Define now the rescaled functions
$g_n(r):=\frac{1}{n}g(-nr)$\,. Let $\{f_n(\bt)\}_n\subset\C^\infty(\pO)$ be any sequence such that $\norm{f_n-f}_{\H^{1}(\pO)}\to 0$\,.
Now consider the smooth functions
\begin{equation}\label{smoothphi}
\tilde{\Phi}_n(\bx):=\Phi(\mathbf{x})+g_n(r)f_n(\bt)\;.
\end{equation}
Clearly we have that $\dot{\tilde{\varphi}}_n(\bt)\equiv f_n(\bt)$ and therefore $\norm{\dot{\tilde{\varphi}}_n-f}_{\H^1(\pO)}\to 0$ as needed.
Now we are going to show that $\tilde{\Phi}_n\stackrel{\H^1}{\to}\Phi$\,. According to Proposition \ref{equivalentsobolev} it is enough to show that
the functions and all their first derivatives converge in the $\H^0(\Omega)$ norm.
\begin{subequations}\label{H1convergence}
\begin{align}
	\norm{\tilde{\Phi}_n(\mathbf{x})-\Phi(\bx)}_{\H^0(\Omega)}&=\norm{g_n(r)f_n(\bt)}_{\H^0([-\frac{2}{n},0]\times\pO)}\\
	&\leq\frac{2}{n^2}\norm{f_n}_{\H^0(\pO)}\;.\notag
\end{align}
\begin{align}
	\norm{\frac{\partial}{\partial r}\tilde{\Phi}_n(\mathbf{x})-\frac{\partial}{\partial r}\Phi(\bx)}_{\H^0(\Omega)}&\leq\norm{f_n(\bt)}_{\H^0([-\frac{2}{n},0]\times\pO)}\\ 
	&\leq\frac{2}{n}\norm{f_n}_{\H^0(\pO)}\;.\notag
\end{align}
\begin{align}
	\norm{\frac{\partial}{\partial \theta}\tilde{\Phi}_n(\mathbf{x})-\frac{\partial}{\partial \theta}\Phi(\bx)}_{\H^0(\Omega)}&=\norm{g_n(r)\frac{\partial}{\partial \theta}f_n(\bt)}_{\H^0([-\frac{2}{n},0]\times\pO)}\\
	&\leq\frac{2C}{n^2}\norm{f_n}_{\H^1(\pO)}\;.\notag
\end{align}
\end{subequations}

The constant $C$ in the last inequality comes from the inequality $$\norm{\frac{\partial f_n}{\partial\theta}}_{\H^0(\pO)}\leq C \norm{f_n}_{\H^1(\pO)}\;.$$ Since $\{f_n(\bt)\}$ is a convergent sequence in $\H^1(\pO)$
the norms appearing at the right hand sides are bounded.
\end{proof}

\begin{corollary}\label{Corsemibounded}
Let $\Phi\in\H^1(\Omega)$ and $c\in\mathbb{R}$\,. Then for every $\epsilon>0$
there exists a
$\tilde{\Phi}\in\C^{\infty}(\Omega)$ with $\dot{\tilde{\varphi}}=c\,\tilde{\varphi}$ such that
$\norm{\Phi-\tilde{\Phi}}_1<\epsilon$\,.
\end{corollary}

\begin{proof}
As in the proof of the preceding lemma it is enough to approximate any smooth function
$\Phi$ with vanishing normal derivative in a collar neighbourhood.
Pick now a sequence of smooth functions $$\tilde{\Phi}_n(\bx):=\Phi(\bx)+c\Phi(0,\bt)\bigl(g_n(r)-\frac{1}{n}\bigr)\;,$$
where $g_n$ is the sequence of scaled functions defined in the proof of the preceding lemma.
This family of functions clearly verifies the boundary condition
$\dot{\tilde{\varphi}}=c\,\tilde{\varphi}$\,.
The inequalities \eqref{H1convergence} now read
\begin{align*}
\norm{\tilde{\Phi}_n(\mathbf{x})-\Phi(\bx)}_{\H^0(\Omega)}
&\leq\norm{c\Phi(0,\bt)\bigl(g_n(r)-\frac{1}{n}\bigr)}_{\H^0([-\frac{2}{n},0]\times\pO)}+\\
  &\qquad\qquad\qquad+\frac{c}{n}\vol(\Omega)\cdot\sup_{\Omega}|\Phi(0,\bt)|\\
&\leq\frac{2}{n^2}\norm{c\Phi(0,\bt)}_{\H^0(\pO)}+\frac{c}{n}\vol(\Omega)\cdot\sup_{\Omega}|\Phi(0,\bt)|\;.
\end{align*}
\begin{align*}
\norm{\frac{\partial}{\partial r}\tilde{\Phi}_n(\mathbf{x})-\frac{\partial}{\partial r}\Phi(\bx)}_{\H^0(\Omega)}
&\leq\norm{c\Phi(0,\bt)}_{\H^0([-\frac{2}{n},0]\times\pO)}\\
&\leq\frac{2c}{n}\norm{\Phi(0,\bt)}_{\H^0(\pO)}\;.
\end{align*}
\begin{align*}
\norm{\frac{\partial}{\partial \theta}\tilde{\Phi}_n(\mathbf{x})-\frac{\partial}{\partial \theta}\Phi(\bx)}_{\H^0(\Omega)}
&\leq\norm{c\bigl(g_n(r)-\frac{1}{n}\bigr)\frac{\partial}{\partial \theta}\Phi(0,\bt)}_{\H^0([-\frac{2}{n},0]\times\pO)}+\\
  &\qquad\qquad\qquad+\frac{c}{n}\vol(\Omega)\cdot\sup_{\Omega}| \frac{\partial \Phi(0,\bt)}{\partial \theta}|\\
&\leq\frac{2c}{n^2}\norm{\frac{\partial \Phi(0,\bt)}{\partial \theta}}_{\H^0(\pO)}+\\
  &\qquad\qquad+\frac{c}{n}\vol(\Omega)\cdot\sup_{\Omega}|\frac{\partial \Phi(0,\bt)}{\partial \theta}|\;.
\end{align*}
\end{proof}

\begin{corollary}\label{Corclosable}
Let $\{\Phi_n\}_n\subset\H^1(\Omega)$ and let $A_U$ be the partial Cayley transform of an admissible unitary $U$\,.
Then it exists a sequence of smooth functions $\{\tilde{\Phi}_n\}\subset\C^{\infty}(\Omega)$ such that
\begin{itemize}
\item $\norm{\Phi_n-\tilde{\Phi}_n}_{\H^1(\Omega)}<\frac{1}{n}$\,,
\item $\norm{\varphi_n-\tilde{\varphi}_n}_{\H^{1/2}(\pO)}<\frac{1}{n}$\,,
\item $\norm{\dot{\tilde{\varphi}}_n-A_U\tilde{\varphi}_n}_{\H^{1/2}(\pO)}<\frac{1}{n}$\,.
\end{itemize}
\end{corollary}

\begin{proof}
For $\Phi_{n_0}$\,, $n_0\in\mathbb{N}$\,, take the approximating smooth function
$\tilde{\Phi}_{n_0}$ as in Lemma~\ref{Lemma approxdotphi} with
$$f:=A_U\varphi_{n_0}\in\H^{1/2}(\pO)$$
(note that since $U$ is admissible we have indeed that $f\in\H^{1/2}(\pO)$\,, cf., Definition~\ref{def:admissible}).
Choose also $\epsilon>0$ such that
$$\epsilon\leq \frac{1}{(1+\norm{A_U}_{\H^{1/2}(\pO)})n_0}$$
and note that this implies $\epsilon\leq\frac{1}{n_0}$\,.
Then the first two inequalities follow directly from Lemma~\ref{Lemma approxdotphi}.
Moreover, we also have
\begin{align*}
	\norm{\dot{\tilde{\varphi}}_{n_0}-A_U\tilde{\varphi}_{n_0}}_{\H^{1/2}(\pO)}&\leq \norm{\dot{\tilde{\varphi}}_{n_0}-A_U\varphi_{n_0}}_{\H^{1/2}(\pO)}+\norm{A_U\varphi_{n_0}-A_U\tilde{\varphi}_{n_0}}_{\H^{1/2}(\pO)}\\
	&\leq \epsilon+ \norm{A_U}_{\H^{1/2}(\pO)}\norm{\varphi_{n_0}-\tilde{\varphi}_{n_0}}_{\H^{1/2}(\pO)}\\
	&\leq (1+\norm{A_U}_{\H^{1/2}(\pO)})\epsilon \leq \frac{1}{n_0}\;.
\end{align*}
\end{proof}

For the analysis of the semi-boundedness and closability of the quadratic form $(Q_U,\D_U)$ of Definition \ref{DefQU} we need to analyse first the following one-dimensional problem in an interval. The operator is defined with Neumann conditions at one end of the interval and Robin-type conditions at the other end.

\begin{definition}\label{Defunidimensional}
Consider the interval $I=[0,2\pi]$ and a real constant $c\in\mathbb{R}$\,.
Define the second order differential operator
$$  R\colon\D(R)\to\H^0([0,2\pi]) \quad\text{by}\quad
    R=-\frac{\d^2}{\d r^2}
$$
on the domain
$$\D(R):=\left\{\Phi\in\C^{\infty}(I)\;\Bigr|\;\; \frac{\partial \Phi}{\partial r}\bigr|_{r=0}=0
         \quad\text{and}\quad \frac{\partial \Phi}{\partial r}\bigr|_{r=2\pi}
=c\Phi|_{r=2\pi} \right\}\;.$$
\end{definition}

\begin{proposition}\label{prop: intervalrobin}
The symmetric operator $R$ of Definition~\ref{Defunidimensional} is essentially self-adjoint with discrete
spectrum and semi-bounded from below with lower bound $\Lambda_0$\,.
\end{proposition}

\begin{proof}
It is well known that this operator together with this boundary conditions defines an essentially self-adjoint operator
(see, e.g., \cite{asorey05, bruning08, grubb68}). We show next that its spectrum is semi-bounded from below.
Its closure is a self-adjoint extension of the Laplace operator defined on $\H_0^2[0,2\pi]$\,.
The latter operator has finite dimensional deficiency indices and its Dirichlet extension is known to have
empty essential spectrum. According to \cite[Theorem 8.18]{weidman80} all the self-adjoint extensions of a
closed, symmetric operator with finite deficiency indices have the same essential spectrum and therefore the spectrum of $R$ is discrete.

Consider now the following spectral problem:
\begin{equation}
R\Phi=\Lambda\Phi,\quad \frac{\partial \Phi}{\partial r}\Bigr|_{r=0}=0,\quad \frac{\partial \Phi}{\partial r}\Bigr|_{r=2\pi}=c\Phi|_{r=2\pi}\;,
\end{equation}
with $c$ a real constant.
On general solutions $\Phi(r)=Ae^{\mathbf{i}\lambda r}+Be^{-\mathbf{i}\lambda r}$ we impose the boundary conditions.
For nonzero solutions we obtain the following relation
\begin{equation}\label{spectral function}
	-(\mathbf{i}\lambda+c)e^{-\mathbf{i}2\pi\lambda}+(\mathbf{i}\lambda-c)e^{\mathbf{i}2\pi\lambda}=0\;,
\end{equation}
where $\Lambda=\lambda^2\in\mathbb{R}$\,. The equation is symmetric under the interchange $\lambda\to-\lambda$\,.
It is therefore enough to consider either $\lambda\geq 0$ or $\lambda=\mathbf{i}\mu$ with $\mu>0$\,.
These two choices correspond to the positive and negative eigenvalues, respectively.
The imaginary part of Eq.~\eqref{spectral function} vanishes identically.
If $\lambda\geq 0$ its real
part takes the form $$\tan 2\pi\lambda= -\frac{c}{\lambda}\;,$$ which leads to infinite solutions for each
$c\in\mathbb{R}$ and therefore there are infinitely many positive eigenvalues.
If $\lambda=\mathbf{i}\mu$ we obtain from Eq.~\eqref{spectral function}
$$e^{-4\pi\mu}=\frac{\mu-c}{\mu+c}\;,$$
which has either no solution for $c<0$\,, the trivial solution $\mu=0$ for $c=0$
and exactly one negative solution for $c>0$\,. So the operator $R$ is positive for $c\leq0$ and semi-bounded from below for $c>0$\,.
We denote the lowest possible eigenvalue by $\Lambda_0$\,.
\end{proof}

\begin{definition}\label{Deftensorproduct}
Consider the interval $I=[0,2\pi]$ and let $\{\Gamma_i(\bt)\}\subset\H^0(\pO)$ be an orthonormal basis.
Consider the following operator $A$ on the tensor product
$\H^0(I)\otimes\H^0(\pO)\simeq \H^0(I\times\pO)$
given by
$$A\colon\D(A)\to \H^0(I)\otimes\H^0(\pO) \quad\text{where}\quad A:=R\otimes\mathbb{I}\;,$$
on its natural domain
$$\D(A)=\Bigl\{\Phi\in\H^0(I)\otimes\H^0(\pO)\;\bigr|\;\; \Phi=\sum_{i=1}^n\Phi_i(r)\Gamma_i(\bt)\,,\;n\in\mathbb{N}\,,\;
               \Phi_i\in\D(R) \Bigr\}\;.$$
\end{definition}

\begin{proposition}\label{Asemibounded}
The operator $A$ is essentially self-adjoint, semi-bounded from below and has the same lower bound $\Lambda_0$
as the operator $R$ of Proposition~\ref{prop: intervalrobin}.
\end{proposition}

\begin{proof}
Let $\Psi\in\ker(A^\dagger\mp\mathbf{i})$ and consider its decomposition in terms of the orthonormal basis
$\{\Gamma_i(\theta)\}\subset\H^0(\pO)$ such that
$\Psi=\sum_{i=0}^\infty\Psi(r)_i\Gamma_i(\bt)$\,. We have that $\scalar{\Psi}{(A\pm\mathbf{i})\Phi}=0$\, for all $\Phi\in\D(A)\;.$
In particular his is true for any $\Phi=\Phi_{i_0}\Gamma_{i_0}\in\D(A)$\,. Then
\begin{align*}
0=\scalar{\Psi}{(A\pm\mathbf{i})\Phi_{i_0}\Gamma_{i_0}}&=\sum_i^{\infty}\scalar{\Psi_i}{(R\pm\mathbf{i})
  \Phi_{i_0}}_{\H^0(I)}\scalar{\Gamma_i}{\Gamma_{i_0}}_{\H^{0}(\pO)}\\
&=\scalar{\Psi_{i_0}}{(R\pm\mathbf{i})\Phi_{i_0}}_{\H^0(I)}\quad\forall \Phi_{i_0}\in\D(R)\;.
\end{align*}
This implies that $\Psi_{i_0}=0$ because, by Proposition~\ref{prop: intervalrobin},
$R$ is essentially self-adjoint. Recall Definition \ref{def:deficiencyspaces} and Remark \ref{rem: sa cases}. Since the choice of $\Phi_{i_0}$ was arbitrary this implies that $\Psi=0$ and therefore $A$ is essentially self-adjoint.

Finally we show the semi-boundedness condition.
Using the orthonormality of the basis $\{\Gamma_i(\bt)\}$
and for any $\Phi\in\D(A)$
we have that
\begin{align*}
	\scalar{\Phi}{A\Phi}_{\H^0(I\times\pO)}&=\sum_{i=1}^n\scalar{\Phi_i}{R\Phi_i}_{\H^0(I)}\\
	&\geq\Lambda_0\sum_{i=1}^n\scalar{\Phi_i}{\Phi_i}_{\H^0(I)}\\
	&=\Lambda_0\scalar{\Phi}{\Phi}_{\H^0(I\times\pO)}\;.
\end{align*}
\end{proof}

\subsection{Quadratic forms and extensions of the minimal Laplacian}

We begin associating quadratic forms to some of the operators on a collar neighbourhood of the precedent subsection.

\begin{lemma}\label{AdomainH1}
Denote by $Q_A$ the closed quadratic form represented by the closure of $A$\,. Then its domain $\D(Q_A)$
contains the Sobolev space of class 1. For any $\Phi\in\H^1(I\times\pO)\subset\D(Q_A)$ we have the expression
$$Q_A(\Phi)=
\int_{\pO}\Bigl[\int_I \frac{\partial \bar{\Phi}}{\partial r}\frac{\partial \Phi}{\partial r}\d r- c|\gamma(\Phi)|^2 \Bigr]\d\mu_{\partial \eta}
\;.$$
\end{lemma}

\begin{proof}
Let $\Phi\in\D(A)$\,. Then, recalling the boundary conditions specified in the domain $\D(R)$, we have that
\begin{align}
Q_A(\Phi)&=\scalar{\Phi}{A\Phi}_{\H^0(I\times\pO)}=\sum_i\scalar{\Phi_i}{R\Phi_i}_{\H^0(I)}\notag\\
&=\sum_i\scalar{\frac{\partial \Phi_i}{\partial r}}{\frac{\partial \Phi_i}{\partial r}}_{\H^0(I)}-c\bar{\Phi}_i(0)\Phi_i(0)\notag\\
&=\int_{\pO}\Bigl[\int_I \frac{\partial \bar{\Phi}}{\partial r}\frac{\partial \Phi}{\partial r}\d r -c|\varphi|^2\Bigr]\d\mu_{\partial \eta}\;.
   \label{QFA}
\end{align}
Now it is easy to check that the graph norm of this quadratic form is dominated by the Sobolev norm of order 1, $\H^1(I\times\pO)$ \,.
\begin{align*}
\normm{\Phi}^2_{Q_A}&=(1+|\Lambda_0|)\norm{\Phi}^2_{\H^0(I\times\pO)}+Q_A(\Phi)\\
&\leq(1+|\Lambda_0|)\norm{\Phi}^2_{\H^0(I\times\pO)}+\int_{\pO}\int_I
     \frac{\partial \bar{\Phi}}{\partial r}\frac{\partial \Phi}{\partial r}\d r\d\mu_{\partial \eta} +c\norm{\varphi}^2_{\H^0(\pO)}\\
&\leq (1+|\Lambda_0|)\norm{\Phi}^2_{\H^0(I\times\pO)}+C\norm{\Phi}^2_{\H^1(I\times\pO)}\\
&\leq C'\norm{\Phi}^2_{\H^1(I\times\pO)}\;,
\end{align*}
where in the second step we have used again the equivalence appearing in Proposition \ref{equivalentsobolev} and Theorem \ref{LMtracetheorem}.
The above inequality shows that $\overline{\D(A)}^{\norm{\cdot}_1}\subset \D(Q_A)$\,. Moreover, Corollary~\ref{Corsemibounded}
states that $\D(A)$ is dense in $\H^1(I\times\pO)$\,. Hence the expression Eq.~\eqref{QFA} holds also on $\H^1(I\times\pO)$\,.
\end{proof}

\begin{theorem}\label{maintheorem1}
Let $U\colon\H^0(\pO)\to\H^0(\pO)$ be a unitary operator with gap at $-1$\,.
Then the quadratic form $Q_U$ of Definition \ref{DefQU} is semi-bounded from below.
\end{theorem}

\begin{proof}
Let $(\Omega,\pO,\eta)$ be a compact, Riemannian manifold with boundary. One can always select a collar neighbourhood $\Xi$ of the boundary with coordinates $(r,\bt)$ such that $\Xi\simeq[-L,0]\times\pO$ and where $$\eta(r,\bt)=\begin{bmatrix}1 & 0 \\ 0 & g(r,\bt) \end{bmatrix}\;.$$ The normal vector field to the boundary is going to be $\frac{\partial}{\partial r}$\,. With this choice the induced Riemannian metric at the boundary becomes $\partial\eta(\bt)\equiv g(0,\bt)$\,. The thickness $L$ of the collar neighbourhood $\Xi$ can be also selected such that it exists $\delta\ll1$ that verifies
\begin{equation}\label{compactcollar}
	(1-\delta)\sqrt{|g(0,\bt)|}\leq \sqrt{|g(r,\bt)|}\leq (1+\delta)\sqrt{|g(0,\bt)|}\;.
\end{equation}
The quadratic form $Q_U$ can be adapted to this splitting. Let $\Phi\in\D_U\subset\H^1(\Omega)$\,. Obviously $\Phi|_\Xi\in\H^1(\Xi)\simeq\H^1(I\times\pO)$\,. In what follows, to simplify the notation and since there is no risk of confusion, the symbol $\Phi$ will stand for both $\Phi\in\H^1(\Omega)$ and $\Phi|_{\Xi}\in\H^1(\Xi)$\,.
\begin{subequations}\label{equationssemibounded}
\begin{align}
Q_U(\Phi)&=\int_\Omega\eta^{-1}(\d\bar{\Phi},\d \Phi)\d\mu_\eta-\int_\pO\bar{\varphi}A\varphi\d\mu_{\partial \eta}\\
&=\int_\Xi\eta^{-1}(\d\bar{\Phi},\d \Phi)\d\mu_\eta+\int_{\Omega\backslash\Xi}\eta^{-1}(\d\bar{\Phi},\d \Phi)\d\mu_\eta
  -\int_\pO\bar{\varphi}A\varphi\d\mu_{\partial \eta}\\
&\geq \int_\Xi\eta^{-1}(\d\bar{\Phi},\d \Phi)\d\mu_\eta-\int_\pO\bar{\varphi}A\varphi\d\mu_{\partial \eta}\label{semibounded1}\\
&=\int_\pO\int_I\Bigl[\frac{\partial \bar{\Phi}}{\partial r}\frac{\partial \Phi}{\partial r}+ g^{-1}(\d_{\bt}\Phi,\d_{\bt}\Phi)\Bigr]
  \sqrt{|g(r,\bt)|}\d r\wedge\d\bt-\int_\pO\bar{\varphi}A\varphi\d\mu_{\partial \eta}\\
&\geq \int_\pO\int_I\frac{\partial \bar{\Phi}}{\partial r}\frac{\partial \Phi}{\partial r}\sqrt{|g(r,\bt)|}\d r\wedge\d\bt
  -\int_\pO\bar{\varphi}A\varphi\d\mu_{\partial \eta}\label{semibounded2}\\
&\geq (1-\delta)\int_\pO\int_I\frac{\partial \bar{\Phi}}{\partial r}\frac{\partial \Phi}{\partial r}
  \sqrt{|g(0,\bt)|}\d r\wedge\d\bt-\int_\pO\bar{\varphi}A\varphi\sqrt{|g(0,\bt)|}\d\bt\label{semibounded3}\\
&\geq(1-\delta)\int_\pO\Bigl[\int_I\frac{\partial \bar{\Phi}}{\partial r}\frac{\partial \Phi}{\partial r}\d r
  -\frac{\norm{A}}{(1-\delta)}|\varphi|^2\Bigr]\sqrt{|g(0,\bt)|}\d\bt\\
&\geq -|\Lambda_0|(1-\delta)\norm{\Phi}^2_{\H^0(I\times\pO)}\geq -|\Lambda_0|\frac{1-\delta}{1+\delta}\norm{\Phi}^2_{\H^0(\Xi)}\\ &\phantom{\geq -|\Lambda_0|(1-\delta)\norm{\Phi}^2_{\H^0(I\times\pO)}}\geq
  -|\Lambda_0|\frac{1-\delta}{1+\delta}\norm{\Phi}^2_{\H^0(\Omega)}\;.\label{semibounded4}
\end{align}\\
\end{subequations}	
In the step leading to \eqref{semibounded1} we have used the fact that the second term is positive. In the step leading to
\eqref{semibounded2} we have used that the second term in the first integrand is positive. Then \eqref{semibounded3} follows using
the bounds \eqref{compactcollar}. The last chain of inequalities follows by Proposition~\ref{Asemibounded} and Lemma~\ref{AdomainH1},
taking $c=\norm{A}/(1-\delta)$\,.
Notice that the semi-bound of Proposition \ref{Asemibounded} is always negative in this case because $c=\norm{A}/(1-\delta)>0$\,.
In Definition \ref{Defunidimensional} the interval $I$ was taken of length $2\pi$ whereas in this case it has length $L$\,.
This affects only in a constant factor that can be absorbed in the constant $c$ by means of a linear transformation of the
manifold $T\colon [0,2\pi]\to I$\,.
\end{proof}

\begin{theorem}\label{maintheorem2}
Let $U:\H^0(\pO)\to\H^0(\pO)$ be an admissible, unitary operator. Then the quadratic form $Q_U$ of Definition \ref{DefQU} is closable.
\end{theorem}

\begin{proof}
According to Remark \ref{Remclosable} a quadratic form is closable iff for any $\Phi\in\overline{\D_U}^{\normm{\cdot}_{Q_U}}$
such that the corresponding Cauchy sequence $\{\Phi_n\}$ verifies $\norm{\Phi_n}\to0$ then
$Q(\Phi)=0$\,. Let $\Phi\in\overline{\D_U}^{\normm{\cdot}_{Q_U}}$\,. \\

(a) Lets show that it exist $\{\tilde{\Phi}_n\}\in\C^\infty(\Omega)$ such that $\normm{\Phi-\tilde{\Phi}_n}_{Q_U}\to 0$ and $\norm{\dot{\tilde{\varphi}}_n-A_U\tilde{\varphi}_n}_{\H^{1/2}(\pO)}\to0$\,. It exists $\{\Phi_n\}\in\D_U\subset\H^1(\Omega)$ such that $\normm{\Phi-\Phi_n}_{Q_U}\to 0$\,. For the sequence $\{\Phi_n\}$ take $\{\tilde{\Phi}_n\}\in\C^{\infty}(\Omega)$ as in Corollary \ref{Corclosable}. Then we have that
\begin{align*}
\normm{\Phi-\tilde{\Phi}_n}_{Q_U}&\leq \normm{\Phi-\Phi_n}_{Q_U}+\normm{\Phi_n-\tilde{\Phi}_n}_{Q_U}\\
&\leq \normm{\Phi-\Phi_n}_{Q_U}+K\norm{\Phi_n-\tilde{\Phi}_n}_1\;,
\end{align*}
where we have used Proposition \ref{H1bound}.\\

(b) Lets assume that $\norm{\Phi_n}\to0$\,. This implies that $\norm{\tilde{\Phi}_n}\to0$\,.
For every $\Psi\in\H^2_0=\D(\Delta_{\mathrm{min}})$ we have that
$$|\scalar{\Delta_{\mathrm{min}}\Psi}{\tilde{\Phi}_n}|\leq \norm{\Delta_{\mathrm{min}}\Psi}\norm{\tilde{\Phi}_n}\to 0\;.$$
Hence $\lim\tilde{\Phi}_n\in\D(\Delta_{\mathrm{min}}^\dagger)=\D(\Delta_{\mathrm{max}})$\,. According to
Theorem~\ref{weaktracetheorem} the traces of such functions exist and are elements of $\H^{-1/2}(\pO)$\,,
i.e., $\tilde{\varphi}_n\stackrel{\H^{-1/2}(\pO)}{\to}\tilde{\varphi}$\,.\\

(c) Finally we have that
\begin{align*}
Q_U(\Phi)&=\lim_{m\to\infty}\lim_{n\to\infty}\left[\scalar{\d\tilde{\Phi}_n}{\d\tilde{\Phi}_m}
           -\scalarb{\tilde{\varphi}_n}{A_U\tilde{\varphi}_m}\right]\\
&=\lim_{m\to\infty}\lim_{n\to\infty}\left[\scalar{\tilde{\Phi}_n}{-\Delta_\eta\tilde{\Phi}_m}+\scalarb{\tilde{\varphi}_n}{\dot{\tilde{\varphi}}_m}
   -\scalarb{\tilde{\varphi}_n}{A_U\tilde{\varphi}_m}\right]\\
&=\lim_{m\to\infty}\pairb{\tilde{\varphi}}{\dot{\tilde{\varphi}}_m-A_U\tilde{\varphi}_m}=0\;.
\end{align*}
Notice that in the last step we have used the continuous extension given in Proposition \ref{proppairing} of the scalar
product of the boundary $\scalarb{\cdot}{\cdot}$ to the pairing $$\pairb{\cdot}{\cdot}:\H^{-1/2}(\pO)\times\H^{1/2}(\pO)\to\mathbb{C}$$
associated to the scale of Hilbert spaces $\H^{1/2}(\pO)\subset\H^0(\pO)\subset\H^{-1/2}(\pO)$\,, and the fact that the unitary operator is admissible.
\end{proof}

Theorem \ref{maintheorem1} and Theorem \ref{maintheorem2} ensure that Theorem \ref{fundteo} applies and that the closure of the quadratic
form $Q_U$ for an admissible unitary $U$ is representable by means of a unique self-adjoint operator $T$\,,
with domain $\D(T)\subset\D(\overline{Q}_U):=\overline{\D_U}^{\normm{\cdot}_{Q_U}}$\,, i.e.,
$$\overline{Q}_U(\Psi,\Phi)=\scalar{\Psi}{T\Phi}\quad\Psi\in\D(\overline{Q}_U),\Phi\in\D(T)\;.$$
The following theorem establishes the relation between this operator $T$ and the Laplace-Beltrami operator.

\begin{theorem}\label{DeltaUextDeltamin}
Let $T$ be the self-adjoint operator with domain $\D(T)$ representing the closed quadratic form $\overline{Q}_U$ with domain $\D(\overline{Q}_U)$\,. The operator $T$ is a self-adjoint extension of the closed symmetric operator $-\Delta_{\mathrm{min}}$\,.
\end{theorem}

\begin{proof}
By Theorem \ref{fundteo} we have that $\Phi\in\D(T)$ iff $\Phi\in\D(\overline{Q}_U)$ and it exists $\chi\in\H^0(\Omega)$
such that $$Q_U(\Psi,\Phi)=\scalar{\Psi}{\chi}\quad\forall\Psi \in \D(\overline{Q}_U)\;.$$
Let $\Phi\in\H^2_0(\Omega)\subset\D_U$ and $\Psi\in\D_U$\,. Then
\begin{align*}
Q(\Psi,\Phi)&=\scalar{\d\Psi}{\d\Phi}-\scalarb{\psi}{A_U\varphi}\\
&=\scalar{\Psi}{-\Delta_{\mathrm{min}}\Phi}+\scalarb{\psi}{\dot{\varphi}}-\scalarb{\psi}{A_U\varphi}\\
&=\scalar{\Psi}{-\Delta_{\mathrm{min}}\Phi}\;.
\end{align*}
Since $\D_U$ is a core for $\overline{Q}_U$ and $\D(\overline{Q}_U)\subset\H^0(\Omega)$
the above equality holds also for every $\Psi\in\D(\overline{Q}_U)$\,.
Therefore $\D(\Delta_{\mathrm{min}})=\H^2_0(\Omega)\subset\D(T)$ and moreover $T|_{\D(\Delta_{\mathrm{min}})}=-\Delta_{\mathrm{min}}$\,.
\end{proof}


\section{Examples}\label{QF: Examples}

In this section we introduce some examples that show that the characterisation of the quadratic forms of Section~\ref{sec:class} and Section \ref{sec:closable and semibounded qf} include a large class of possible self-adjoint extensions of the Laplace-Beltrami operator. This section also illustrates the simplicity in the description of extensions using admissible unitaries at the boundary.

As the boundary manifold $\pO$ is an $(n-1)$-dimensional, smooth manifold, there always exist a $(n-1)$-simplicial complex  $\mathcal{K}$ and a smooth diffeomorphism $f:\mathcal{K}\to\pO$ such that $f(\mathcal{K})=\pO$\, cf., \cite{whitehead40,whitney57}. Any simplex in the complex is diffeomorphic to a reference polyhedron $\Gamma_0\subset\mathbb{R}^{n-1}$\,. The simplicial complex $\mathcal{K}$ defines therefore a triangulation of the boundary $\pO=\cup_{i=1}^N\Gamma_i$\,, where $\Gamma_i:=f(A_i)$\,, $A_i\in\mathcal{K}$\,. For each element of the triangulation $\Gamma_i$ it exists a diffeomorphism $g_i:\Gamma_0\to\Gamma_i$\,. Consider a reference Hilbert space $\H^0(\Gamma_0,\d\mu_0)$ where $d\mu_0$ is a fixed smooth volume element. Each diffeomorphism $g_i$ defines a unitary transformation as follows:
\begin{definition}\label{def unitary transformation}
	Let $|J_i|$ be the Jacobian determinant of the transformation of coordinates given by the diffeomorphism $g_i:\Gamma_0\to\Gamma_i$\,. Let $\mu_i\in\C^\infty(\pO)$ be the proportionality factor $g_i^\star\d\mu_{\partial\eta}=\mu_i\d\mu_0$\,, where $g_i^\star$ stands for the pull-back of the diffeomorphism. The unitary transformation $T_i:\H^0(\Gamma_i,\d\mu_{\partial\eta})\to\H^0(\Gamma_0,\d\mu_0)$ is defined by
	\begin{equation}\label{unitary transformation}
		T_i\Phi:=\sqrt{|J_i|\mu_i}(\Phi\circ g_i)\;.
	\end{equation}
\end{definition}
We show that the transformation above is unitary. First note that $T$ is invertible. It remains to show that $T$ is an isometry:
\begin{align*}
	\scalar{\Phi}{\Psi}_{\Gamma_i}&=\int_{\Gamma_i}\overline{\Phi}\Psi\d\mu_{\partial \eta}\\
		&=\int_{\Gamma_0}(\overline{\Phi\circ g_i})(\Psi\circ g_i)|J_i|g_i^\star\d\mu_{\partial \eta}\\
		&=\int_{\Gamma_0}(\overline{\Phi\circ g_i})(\Psi\circ g_i)|J_i|\mu_i \d\mu_0=\scalar{T_i\Phi}{T_i\Psi}_{\Gamma_0}\;.
\end{align*}

\begin{example}\label{periodic}
	Consider that the boundary of the Riemannian manifold $(\Omega,\pO,\eta)$ admits a triangulation of two elements, i.e., $\pO=\Gamma_1\cup\Gamma_2$\,. The Hilbert space of the boundary satisfies $\H^0(\pO)=\H(\Gamma_1\cup\Gamma_2)\simeq \H^0(\Gamma_1)\oplus\H^0(\Gamma_2)$\,. The isomorphism is given explicitly by the characteristic functions $\chi_i$ of the submanifolds $\Gamma_i$\,, $i=1,2$\,. Modulo a null measure set we have that $$\Phi=\chi_1\Phi+\chi_2\Phi\;.$$ We shall define unitary operators $U=\H^0(\pO)\to\H^0(\pO)$ that are adapted to the block structure induced by the latter direct sum:
	$$U=\begin{bmatrix} U_{11} & U_{12} \\ U_{21} & U_{22} \end{bmatrix}\;,$$
where $U_{ij}:\H^0(\Gamma_{j})\to\H^0(\Gamma_{i})$\,. Hence, consider the following unitary operator
	\begin{equation}\label{Eqperiodic}
		U=\begin{bmatrix} 0 & T_1^*T_2 \\ T_2^*T_1 & 0 \end{bmatrix}\;,
	\end{equation}
where the unitaries $T_i$ are defined as in Definition \ref{def unitary transformation}. Clearly, $U^2=\mathbb{I}$\,, and therefore the spectrum of $U$ is $\sigma(U)=\{-1,1\}$ with the corresponding orthogonal projectors given by $$P^{\bot}=\frac{1}{2}(\mathbb{I}-U)\;,$$ $$P=\frac{1}{2}(\mathbb{I}+U)\;.$$ The partial Cayley transform $A_U$ is in this case the null operator, since $P(\mathbb{I}-U)=0$\,. The unitary operator is therefore admissible and the corresponding quadratic form will be closable. The domain of the corresponding quadratic form $Q_U$ is given by all the functions $\Phi\in\H^1(\Omega)$ such that $P^\bot\gamma(\Phi)=0$\,, which in this case becomes
\begin{equation}
	P^\bot\gamma(\Phi)=\frac{1}{2}
	\begin{bmatrix}
		\mathbb{I}_1 & -T_1^*T_2 \\ -T_2^*T_1 & \mathbb{I}_2
	\end{bmatrix}\begin{bmatrix} \chi_1\gamma(\Phi) \\ \chi_2 \gamma(\Phi) \end{bmatrix}=
	\begin{bmatrix}
		\chi_1\gamma(\Phi)-T_1^*T_2\chi_2\gamma(\Phi)\\
		-T_2^*T_1\chi_1\gamma(\Phi)+\chi_2\gamma(\Phi)
	\end{bmatrix}=0\;.
\end{equation}
We can rewrite the last condition as
\begin{equation}
	T_1(\chi_1\gamma(\Phi))=T_2(\chi_2\gamma(\Phi))\;.
\end{equation}	
More concretely, this boundary conditions describe generalised periodic boundary conditions identifying the two triangulation elements of the boundary with each other. The unitary transformations $T_i$ are necessary to make the triangulation elements congruent. In particular, if $(\Gamma_1,\partial\eta_1)$ and $(\Gamma_2,\partial\eta_2)$ are isomorphic as Riemannian manifolds then one can recover the standard periodic boundary conditions.
\end{example}


\begin{example}\label{ex:quasiperiodic}
	Consider the same situation as in the previous example but with the unitary operator replaced by
	\begin{equation}\label{Eqquasiperiodic}
		U=\begin{bmatrix} 0 & T_1^*e^{\mathrm{i}\alpha}T_2 \\ T_2^*e^{-\mathrm{i}\alpha}T_1 & 0 \end{bmatrix}\,,\quad \alpha\in \C^\infty(\Gamma_0)\;.
	\end{equation}
In this case we have also that $U^2=\mathbb{I}$ and the calculations of the previous example can be applied step by step. More concretely $P^\bot=(\mathbb{I}-U)/2$ and the partial Cayley transform also vanishes. The boundary condition becomes in this case
\begin{equation}\label{quasiperiodic}
	T_1(\chi_1\gamma(\Phi))=e^{\mathrm{i}\alpha}T_2(\chi_2\gamma(\Phi))\;.
\end{equation}
This boundary conditions can be called generalised quasi-periodic boundary conditions. For simple geometries and constant function $\alpha$ these are the boundary conditions that define the periodic Bloch functions.
\end{example}

The condition $\alpha\in \C^\infty(\Gamma_0)$ in the example above can be relaxed. First we will show that the isometries $T_i$  do preserve the regularity of the function.

\begin{proposition}\label{prop regularity}
	Let $T_i$ be a unitary transformation as given by Definition \ref{def unitary transformation}. Let $\Phi\in\H^k(\Gamma_i)$\,, $k\geq0$\,. Then $T_i\Phi\in\H^k(\Gamma_0)$\,.
\end{proposition}
	
\begin{proof}
	It is well known, cf., \cite[Theorem 3.41]{adams03} or \cite[Lemma 7.1.4]{davies95}, that the pull-back of a function under a smooth diffeomorphism $g:\Omega_1\to\Omega_2$ preserves the regularity of the function, i.e., $g^\star\Phi\in\H^k(\Omega_1)$ if $\Phi\in\H^k(\Omega_2)$\,, $k\geq0$\,. It is therefore enough to prove that multiplication by a smooth positive function also preserves the regularity. According to Definition \ref{DefSobolev2} it is enough to prove it for a smooth, compact, boundary-less Riemannian manifold $(\tilde{\Omega},\tilde{\eta})$ and to consider that $\Phi\in\C^\infty(\tilde{\Omega})$\,, since this set is dense in $\H^k(\tilde{\Omega})$\,. Let $f\in\C^\infty(\tilde{\Omega})$\,.
	\begin{align*}
		\int_{\tilde{\Omega}} \overline{f\Phi}(I-\Delta_{\tilde{\eta}})^k(f\Phi)\d\mu_{\tilde{\eta}}&\leq \sup_{\tilde{\Omega}}|f|\int_{\tilde{\Omega}}\overline{\Phi}(I-\Delta_{\tilde{\eta}})^k(f\Phi)\d\mu_{\tilde{\eta}}\\
		&\leq \sup_{\tilde{\Omega}}|f|\int_{\tilde{\Omega}}\overline{(I-\Delta_{\tilde{\eta}})^k\Phi}f\Phi\d\mu_{\tilde{\eta}}\\
		&\leq (\sup_{\tilde{\Omega}}|f|)^2\int_{\tilde{\Omega}}\overline{(I-\Delta_{\tilde{\eta}})^k\Phi}\Phi\d\mu_{\tilde{\eta}}<\infty\;.
	\end{align*}
We have used Definition \ref{DefSobolev} directly and the fact that the operator $(I-\Delta_{\tilde{\eta}})^k$ is essentially self-adjoint over the smooth functions. 
\end{proof}

According to Proposition \ref{prop regularity} we have that $T_i(\chi_i\gamma(\Phi))\in\H^{1/2}(\Gamma_0)$\,, $i=1,2$\,. Therefore, to get nontrivial solutions for the expression \eqref{quasiperiodic}, the function $\alpha:\Gamma_0\to[0,2\pi]$ can be chosen such that $e^{i\alpha}T_2(\chi_2\gamma)\in\H^{1/2}(\Gamma_0)$\,. Since $\C^0(\Gamma_0)$ is a dense subset in $\H^{1/2}(\Gamma_0)$\,, and pointwise multiplication is a continuous operation for continuous functions it is enough to consider $\alpha\in\C^0(\Gamma_0)$\,.

\begin{example}\label{generalized Robin}
Consider that the boundary of the Riemannian manifold $(\Omega,\pO,\eta)$ admits a triangulation of two elements like in the Example \ref{periodic}.
So we have that $\pO=\Gamma_1\cup\Gamma_2$\,. Consider the following unitary operator $U:\H^0(\pO)\to\H^0(\pO)$ adapted to the block structure
defined by this triangulation
\begin{equation}\label{EqRobin}
U=\begin{bmatrix} e^{\mathrm{i}\beta_1}\mathbb{I}_1 & 0\\ 0 & e^{\mathrm{i}\beta_2}\mathbb{I}_2 \end{bmatrix}\;,
\end{equation}
where $\C^0(\Gamma_i)\ni\beta_i:\Gamma_i\to [-\pi+\delta,\pi-\delta]$ with $\delta>0$\,. The latter condition guaranties that the unitary matrix
has gap at $-1$\,. Since the unitary is diagonal in the block structure, it is clear that $P^\bot=0$\,.
The domain of the quadratic form $Q_U$ is given in this case by all the functions $\Phi\in\H^1(\Omega)$ \,. The partial Cayley transform is
in this case the operator $A_U=\H^0(\pO)\to\H^0(\pO)$ defined by
\begin{equation}\label{partial Cayley Robin}
A_U=\begin{bmatrix} -\tan\frac{\beta_1}{2} & 0 \\ 0 & -\tan{\frac{\beta_2}{2}} \end{bmatrix}\;.
\end{equation}
A matrix like the one above will lead to self-adjoint extensions of the Laplace-Beltrami operator that verify generalised Robin type boundary
conditions $\chi_i\dot{\varphi}=-\tan\frac{\beta_i}{2}\chi_i\varphi$\,. Unfortunately, the partial Cayley transform does not satisfy the
admissibility condition in this case. Nevertheless, we will show that the quadratic form above is indeed closable.
\end{example}

Given a triangulation of the boundary $\pO=\cup_{i=1}^N\Gamma_i$ we can consider the Hilbert space that results of the direct sum of the corresponding Sobolev spaces. We will denote it as $$\oplus\H^k:=\oplus_{i=1}^{N}\H^{k}(\Gamma_i)\;.$$ Assuming that the partial Cayley transform verifies the condition $$\norm{A_U\gamma(\Phi)}_{\oplus\H^{1/2}}\leq K \norm{\gamma(\Phi)}_{\oplus\H^{1/2}}\;,$$ we can generalise Lemma \ref{Lemma approxdotphi} and Corollary \ref{Corclosable} as follows.

\begin{lemma}[Lemma $\ref{Lemma approxdotphi}^*$]\label{Lemma bis}
	Let $\Phi\in\H^1(\Omega)$\,, $f\in\left(\oplus\H^{1/2}\right)$\,. Then, for every $\epsilon>0$ it exists $\tilde{\Phi}\in\C^\infty(\Omega)$ such that
	\begin{itemize}
	\item $\norm{\Phi-\tilde{\Phi}}_1<\epsilon$\,,
	\item $\norm{\varphi-\tilde{\varphi}}_{\H^{1/2}(\pO)}<\epsilon$
	\item $\norm{f-\dot{\tilde{\varphi}}}_{\oplus\H^{1/2}}<\epsilon$\,.
	\end{itemize}
\end{lemma}

\begin{proof}
 The proof of this lemma follows exactly the one for the Lemma \ref{Lemma approxdotphi}. It is enough to notice that the space $\H^1(\pO)$ is dense in $\oplus\H^{1/2}$\,.
\end{proof}

\begin{corollary}[Corollary $\ref{Corclosable}^*$]\label{Corollary bis}
Let $\{\Phi_n\}\subset\H^1(\Omega)$ and let $A_U$ be the partial Cayley transform of a unitary operator with gap at $-1$
such that $$\norm{A\gamma(\Phi)}_{\oplus\H^{1/2}}\leq K \norm{\gamma(\Phi)}_{\oplus\H^{1/2}}\;.$$
Then it exists a sequence of smooth functions $\{\tilde{\Phi}_n\}\in\C^{\infty}(\Omega)$ such that 
\begin{itemize}
	\item $\norm{\Phi_n-\tilde{\Phi}_n}_{\H^1(\Omega)}<\frac{1}{n}$\,,
	\item $\norm{\varphi_n-\tilde{\varphi}_n}_{\H^{1/2}(\pO)}<\frac{1}{n}$\,,
	\item $\norm{\dot{\tilde{\varphi}}_n-A_U\tilde{\varphi}_n}_{\oplus\H^{1/2}}<\frac{1}{n}$\,.
\end{itemize}	
\end{corollary}

\begin{proof}
The proof is the same as for Corollary \ref{Corclosable} but now we take $\tilde{\Phi}_{n_0}$ as in Lemma \ref{Lemma bis} with
$f=A_U\varphi_{n_0}\in \left(\oplus\H^{1/2}\right)$\,.
\end{proof}

Now we can show that the quadratic forms $Q_U$ defined for unitary operators of the form appearing in Example \ref{generalized Robin} are closable.
Let us show first that the partial Cayley transform of Equation \eqref{partial Cayley Robin} verifies the conditions of the
Corollary~\ref{Corollary bis} above. We have that
\begin{align*}
\norm{A_U\varphi}_{\oplus\H^{1/2}}^2&=\norm{A_U\chi_1\varphi}^2_{\H^{1/2}(\Gamma_1)}+\norm{A_U\chi_2\varphi}^2_{\H^{1/2}(\Gamma_2)}\\
&=\norm{\tan{\frac{\beta_1}{2}}\chi_1\varphi}^2_{\H^{1/2}(\Gamma_1)}+\norm{\tan{\frac{\beta_2}{2}}\chi_2\varphi}^2_{\H^{1/2}(\Gamma_2)}\\
&\leq K \left[\norm{\chi_1\varphi}^2_{\H^{1/2}(\Gamma_1)}+\norm{\chi_2\varphi}^2_{\H^{1/2}(\Gamma_2)}\right]=K\norm{\varphi}^2_{\oplus\H^{1/2}}\;.
\end{align*}
The last inequality follows from the discussion after Example \ref{ex:quasiperiodic} because the functions
$\beta_i\colon\Gamma_i\to[-\pi+\delta,\pi-\delta]$ are continuous. Take the sequence $\{\Phi_n\}\in\D_U$ as in the
proof of Theorem \ref{maintheorem2} and accordingly take $\{\tilde{\Phi}_n\}\in\C^{\infty}(\Omega)$ as in Corollary~\ref{Corollary bis}. Then we have that
\begin{align*}
|Q(\Phi)|&=\lim_{m\to\infty}\lim_{n\to\infty}\left|\scalar{\d\tilde{\Phi}_n}{\d\tilde{\Phi}_m}-\scalarb{\tilde{\varphi}_n}{A\tilde{\varphi}_m}\right|\\
&\leq\lim_{m\to\infty}\lim_{n\to\infty}\left[|\scalar{\tilde{\Phi}_n}{-\Delta_\eta\tilde{\Phi}_m}|+|\scalarb{\tilde{\varphi}_n}{\dot{\tilde{\varphi}}_m-A_U\tilde{\varphi}_m}|\right]\\
&=\lim_{m\to\infty}\lim_{n\to\infty}\left[|\scalar{\tilde{\Phi}_n}{-\Delta_\eta\tilde{\Phi}_m}|+|\smash{\sum_{i=1}^N}\scalar{\tilde{\varphi}_n}{\dot{\tilde{\varphi}}_m-A_U\tilde{\varphi}_m}_{\Gamma_i}|\right]
\end{align*}
\begin{align*}
&\leq\lim_{m\to\infty}\lim_{n\to\infty}\sum_{i=1}^N|\scalar{\tilde{\varphi}_n}{\dot{\tilde{\varphi}}_m-A_U\tilde{\varphi}_m}_{\Gamma_i}|\\
&\leq\lim_{m\to\infty}\lim_{n\to\infty}\sum_{i=1}^N\norm{\chi_i\tilde{\varphi}_n}_{\H^{-1/2}(\Gamma_i)}\norm{\chi_i\dot{\tilde{\varphi}}_m-\chi_iA_U\tilde{\varphi}_m}_{\H^{1/2}(\Gamma_i)}\\
&\leq \lim_{m\to\infty}\lim_{n\to\infty} \norm{\tilde{\varphi}_n}_{\H^{-1/2}(\pO)} \smash{\sum_{i=1}^N}\norm{\chi_i\dot{\tilde{\varphi}}_m-\chi_iA_U\tilde{\varphi}_m}_{\H^{1/2}(\Gamma_i)}=0\;.
\end{align*}
We have used Definition \ref{DefSobolev2} and the structure of the scales of Hilbert spaces
$\H^{1/2}(\Gamma_i)\subset\H^0(\Gamma_i)\subset\H^{-1/2}(\Gamma_i)$\,. Hence, the unitary operators of Example \ref{generalized Robin} are closable.
In particular, this class of closable quadratic forms defines generalised Robin type boundary conditions $\dot{\varphi}=-\tan{\frac{\beta}{2}}\varphi$
where $\beta$ is allowed to be a piecewise continuous function with discontinuities at the vertices of the triangulation.

\begin{example}\label{generalized Robin2}
Consider a unitary operator at the boundary of the form
\begin{equation}\label{Eq mixed}
	U=\begin{bmatrix} -\mathbb{I}_1 & 0\\ 0 & e^{\mathrm{i}\beta_2}\mathbb{I}_2 \end{bmatrix}\;,
\end{equation}
with $\beta_2:\Gamma_2 \to [-\pi+\delta,\pi-\delta]$ continuous. Again we need the condition $\delta>0$
in order to guaranty that the unitary matrix $U$ has gap at $-1$\,. In this case it is clear that
$$P^\bot=\begin{bmatrix} \mathbb{I}_1 & 0 \\ 0 & 0 \end{bmatrix}\;,$$ and that the partial Cayley transform becomes
$$A_U=\begin{bmatrix} 0 \\ -\tan{\frac{\beta_2}{2}} \end{bmatrix}\;.$$ This partial Cayley transform verifies the weaker
admissibility condition of the previous example and therefore defines a closable quadratic form too. This one defines a
boundary condition of the mixed type where
$$\chi_1\varphi=0\,,\quad\chi_2\dot{\varphi}=-\tan{\frac{\beta_2}{2}}\chi_2\varphi\;.$$
In particular when $\beta_2=0$ this mixed type boundary condition defines the boundary conditions of the so called
\emph{Zaremba problem} with $$\chi_1\varphi=0\,,\quad\chi_2\dot{\varphi}=0\;.$$
\end{example}

\begin{example}
Let $(\Omega,\pO,\eta)$ be a smooth, compact, Riemannian manifold.
Suppose that the boundary manifold admits a triangulation $\pO=\cup_{i=1}^N\Gamma_i$\,. Any unitary matrix that has
blockwise the structure of any of the above examples, i.e., Equations \eqref{Eqperiodic}, \eqref{Eqquasiperiodic},
\eqref{EqRobin} or \eqref{Eq mixed} leads to a closable, semi-bounded quadratic form $Q_U$\,.
\end{example}

\clearemptydoublepage
\chapter{Numerical Scheme to Solve the Spectral Problem of the Laplace-Beltrami Operator}\label{cha:FEM}
\markboth{Numerical Scheme to Solve the Spectral Problem of the L.-B. Operator}{}

In this chapter we develop a class of numerical algorithms that can be used to approximate the spectral problem for the self-adjoint extensions of the Laplace-Beltrami operator described in chapter \ref{cha:QF}. The numerical scheme that we use is based in the finite element method. A standard reference for this method is \cite{brenner08}.

In what follows let $-\Delta_U$ denote the self-adjoint operator associated to the closure of the quadratic form $Q_U$ of Definition \ref{DefQU}. Theorem \ref{maintheorem1} and Theorem \ref{maintheorem2} ensure that this closure exists and that the operator $-\Delta_U$\,, with domain $\D(\Delta_U)$\,, is semi-bounded below. Moreover $-\Delta_U$\newnot{symb:DeltaU} is a self-adjoint extension of $-\Delta_\mathrm{min}$\,, cf., Theorem \ref{DeltaUextDeltamin}. We are interested in obtaining numerical approximations of pairs $(\Phi,\lambda)\in\H^0(\Omega)\times\mathbb{R}$ that are solutions of the spectral problem
\begin{equation}\label{spectralproblem}
	-\Delta_U\Phi=\lambda\Phi\quad\Phi\in\D(\Delta_U)\;.
\end{equation}
\begin{proposition}
	A pair $(\Phi,\lambda)\in\H^0(\Omega)\times\mathbb{R}$ is a solution of the spectral problem \eqref{spectralproblem} iff it is a solution of the weak spectral problem 
	\begin{equation}\label{weakspectralproblem}
		Q_U(\Psi,\Phi)=\lambda\scalar{\Psi}{\Phi}\quad\forall\Psi\in\D_U\;.
	\end{equation}
\end{proposition}

\begin{proof}
	The only if part is trivial. Suppose that $(\Phi,\lambda)\in\overline{\D_U}^{\normm{\cdot}_{Q_U}}\times\mathbb{R}$ is a solution of the weak spectral problem \eqref{weakspectralproblem}. Let $\{\Phi_n\}\in\D(\Delta_U)$ be a sequence such that $\normm{\Phi_n-\Phi}_{Q_U}\to 0$ and let $\Psi\in\D(\Delta_U)$\,. Then we have that
	\begin{align*}
		\scalar{-\Delta_U\Psi}{\Phi}&=\scalar{-\Delta_U\Psi}{\lim_{n\to\infty}\Phi_n}\\
			&=\lim_{n\to\infty}\scalar{-\Delta_U\Psi}{\Phi_n}\\
			&=\lim_{n\to\infty} Q_U(\Psi,\Phi_n)\\
			&=Q_U(\Psi,\Phi)\\
			&=\scalar{\Psi}{\lambda\Phi}\;.
	\end{align*}
According to Definition \ref{adjointoperator} the equality above ensures that $\Phi\in\D(\Delta^\dagger_U)=\D(\Delta_U)$. Moreover it ensures that $$-\Delta_U\Phi=\lambda\Phi\;.$$
\end{proof}
Hence, in order to approximate the solutions of \eqref{spectralproblem} it is enough to approximate the solutions of the weak spectral problem \eqref{weakspectralproblem}. This is very convenient for two reasons. First, the domain of the quadratic form is bigger than the domain of the associated operator and we can use a wider space to look for the solutions. In our particular case this means that we can look for solutions in $\H^1(\Omega)$ instead of looking for them in the more regular space $\H^2(\Omega)$\,. Second, the finite-dimensional approximation of the weak spectral problem is explicitly Hermitean. Numerical tools for computing the eigenvalues of Hermitean finite-dimensional matrices are much faster than the corresponding ones for the non-Hermitean case.

The way that we shall use to approximate the solution of the weak spectral problem \eqref{weakspectralproblem} is to find a family of finite-dimensional subspaces $\{S^N\}_N\subset\D_U$ and to look for solutions $(\Phi_N,\lambda_N)\in S^N\times\mathbb{R}$ of the approximate spectral problem
	\begin{equation}\label{approximatespectralproblem}
		Q_U(\Psi_N,\Phi_N)=\lambda_N\scalar{\Psi_N}{\Phi_N}\quad\forall{\Psi_N\in S^N}\;.
	\end{equation}\\

Section \ref{convergenceofthenumericalscheme} is devoted to show that the solutions of the approximate spectral problem \eqref{approximatespectralproblem} converge to the solutions of the weak spectral problem \eqref{weakspectralproblem} provided that the family $\{S^N\}_N$ verifies the appropriate conditions. In Section \ref{FEM} we construct explicitly a family $\{S^N\}_N$ for the one-dimensional case. Even if this situation may appear simple, the fact that we want to compute any possible self-adjoint extension of $\Delta_{\mathrm{min}}$ introduces complications that are not dealt with in the common algorithms based on the finite element method. More concretely, one needs to introduce a nonlocal subspace of functions that is able to encode all the possible boundary conditions. In Section \ref{Numerical_experiments_and_conclusions} we study the stability of the method, treat some particular cases, and compare the results with other available algorithms. In particular we show that this procedure is more reliable in some cases and that it can implement any boundary condition leading to a self-adjoint operator over a one-dimensional manifold.


\section{Convergence of the numerical scheme}\label{convergenceofthenumericalscheme}

In what follows we assume that the unitary operator $U$ describing the quadratic form $Q_U$ is admissible. Hence we are under the conditions of Theorem \ref{maintheorem1} and Theorem \ref{maintheorem2} and therefore the quadratic form
\begin{equation}\label{QFcha4}
	Q_U(\Phi,\Psi)=\scalar{\d\Phi}{\d\Psi}-\scalarb{\gamma(\Phi)}{A_U\gamma(\Phi)}
\end{equation}
with domain
\begin{equation}\label{domaincha4}
	\D_U=\bigl\{ \Phi\in\H^1(\Omega)\bigr|P^{\bot}\gamma(\Phi)=0 \bigr\}\;,
\end{equation}
where $A_U$ is the partial Cayley transform of Definition \ref{partialCayley}, is closable and semi-bounded below.

\begin{definition}\label{def: approx family}
	Let $\{S^N\}_N$ be a family os subspaces of $\H^0(\Omega)$\,. We will say that $\{S^N\}_N$ is an \textbf{approximating family} of $Q_U$ if 
	\begin{equation}\label{SNdense}
		\overline{\bigcup_{N>0}S^N}^{\norm{\cdot}_1}=\D_U\;.
	\end{equation}
\end{definition}

The convergence of the eigenvalues and eigenvectors follows by standard arguments.\\

\begin{theorem}\label{ConvergenciaSoluciones}
Let $\{S^N\}_N$ be an approximating family of $Q_U$ and let $$\{(\Phi_N,\lambda_N)\}_N\in \{S^N\}_N \times \mathbb{R}$$ be the sequence of solutions corresponding to the the $n$th lowest eigenvalue of the approximate spectral problems \eqref{approximatespectralproblem}. Then a solution $(\Phi,\lambda)$ of the weak spectral problem exists such that $$\normm{\Phi-\Phi_N}_{Q_U}\to 0\quad \text{and}\quad \lim_{N\to\infty}\lambda_N=\lambda\;.$$
\end{theorem}

\begin{proof}
The admissibility condition for the unitary operator $U$ ensures that the quadratic form $Q_U$ is semi-bounded and closable. The convergence of the eigenvalues is proved as follows.

Let $V_n$ and $V^N_n$ be the subspaces $$V_n=\{\Phi\in\D_U\bigr|\scalar{\Phi}{\xi_i}=0\,,\xi_i\in\H^0(\Omega)\,,i=1,\dots,n\}\;,$$ $$V^N_n=\{\Phi_N\in S^N \bigr|\scalar{\Phi_N}{\xi_i}=0\,,\xi_i\in\H^0(\Omega)\,,i=1,\dots,n\}\;.$$ Then, applying the \textrm{min-max} Principle, Theorem~\ref{minmaxprinciple}, to the quadratic forms on the domains of the weak and the approximate spectral problems and subtracting them we get
\begin{equation}
\sup_{\xi_1,...,\xi_{n-1}}\biggl[\inf_{\Phi_N\in V^N_{n-1}} \frac{Q(\Phi _N,\Phi _N)}{\norm{\Phi _N}^2}\biggr]-\sup_{\xi_1,...,\xi_{n-1}}\biggl[\inf_{\Phi\in V_{n-1}} \frac{Q(\Phi ,\Phi )}{\norm{\Phi}^2}\biggr]=\lambda_N-\lambda\;,
\end{equation}
where $\lambda_N$ is the $n$th eigenvalue of the approximate problem and $\lambda$ is either the $n$th lowest eigenvalue of the weak problem or the bottom of the essential spectrum of $-\Delta_U$\,. The left-hand side tends to zero in the limit $N\to\infty$ provided that $S^N$ is an approximating family, cf., Definition \ref{def: approx family}, and therefore
\begin{equation}\label{Convergenceeigenvalue}
	\lim_{N\to\infty}\lambda_N=\lambda\;.
\end{equation}
Now let $\{(\Phi_N,\lambda_N)\}_N$ be a sequence of solutions of \eqref{approximatespectralproblem} corresponding to the $n$th lowest eigenvalue $\lambda_N$\,. We can assume that $\normm{\Phi_N}_{Q_U}=1$ for every $N$ and therefore the Banach-Alaoglu theorem ensures that subsequences $\{\Phi_{N_j}\}$ and accumulation points $\Phi$ exist such that for every $\Psi\in\D_U$ 
\begin{equation}\label{BanachAlaoglu}
\lim_{N_j\to\infty}Q_U(\Psi,\Phi-\Phi_{N_j})=0\;.
\end{equation}
If we denote by $P_N\Psi$ the orthogonal projection of $\Psi$ onto $S^N$\,, we have that for every $\Psi\in\D_U$
\begin{align*}
	|Q(\Psi,\Phi)-\lambda\scalar{\Psi}{\Phi}|&\leq|Q(\Psi,\Phi-\Phi_{N_j})|+|\lambda||\scalar{\Psi}{\Phi-\Phi_{N_j}}|+\\
	&\phantom{\leq}+|\lambda_{N_j}-\lambda||\scalar{\Psi}{\Phi_{N_j}}|+|Q(\Psi-P_N\Psi,\Phi_{N_j})|+\\
	&\phantom{\leq}+|\lambda_{N_j}||\scalar{\Psi-P_N\Psi}{\Phi_{N_j}}|\;.
\end{align*}
The first two terms go to zero because of the weak convergence in $\normm{\cdot}_{Q_U}$\,. The third one because of \eqref{Convergenceeigenvalue} and the fact that $$\norm{\Phi_{N_j}}\leq\normm{\Phi_{N_j}}_{Q_U}=1\;.$$ The last two terms go to zero because the assumption that $\{S^N\}_N$ is an approximating family ensures that $\norm{\Psi-P_N\Psi}_1\to0$ and both terms are continuous with respect to the Sobolev norm of order 1, cf., Proposition \ref{H1bound}. Since the left-hand side does not depend on $N_j$ this shows that $(\Phi,\lambda)$ is a solution of the weak spectral problem.

\end{proof}


\section{Finite element method for the eigenvalue problem in dimension 1}\label{FEM}

As it was stated at the beginning of this chapter, we will restrict our attention to the case of one-dimensional compact manifolds. Let us discuss first the particularities of this situation. 

Notice first that a compact one-dimensional manifold $\Omega$ consists of a finite number of closed intervals $I_\alpha$\,, $\alpha = 1,\ldots,n$\,.  Each interval will have the form $I_\alpha = [a_\alpha, b_\alpha] \subset \mathbb{R}$ and the boundary of the manifold $\Omega = \coprod_{\alpha=1}^n [a_\alpha, b_\alpha]$ is given by the family of points $\pO = \{ a_1, b_1, \ldots, a_n,b_n\}$\,.     Functions $\Psi$ on $\Omega$ are determined by vectors $(\Psi_1, \ldots, \Psi_n)$ of complex valued functions $\Psi_\alpha \colon I_\alpha \to \mathbb{C}$\,.
A Riemannian metric $\eta$ on $\Omega$ is given by specifying a Riemannian metric $\eta_\alpha$ on each interval $I_\alpha$\,, this is, by a positive smooth function $\eta_\alpha(x) > 0$ on the interval $I_\alpha$\,, i.e., $\eta|_{I_\alpha} = \eta_\alpha (x)\d x\otimes\d x$\,.    Then the $\H^0(\Omega)$ inner product on $I_\alpha$ takes the form $\langle \Psi_\alpha , \Phi_\alpha \rangle = \int_{a_\alpha}^{b_\alpha} \overline{\Psi}_\alpha (x) \Phi_\alpha (x) \sqrt{\eta_\alpha (x)} dx$ and the Hilbert space of square integrable functions on $\Omega$ is given by $\H^0(\Omega) = \bigoplus_{\alpha= 1}^n \H^0(I_\alpha, \eta_\alpha)$\,.

On each subinterval $I_\alpha = [a_\alpha, b_\alpha]$ the differential operator $\Delta_{\eta_\alpha} = \Delta_\eta|_{I_\alpha}$ takes the form 
of a Sturm--Liouville operator 
\begin{equation}\label{sturm}
\Delta_\alpha = - \frac{1}{W_\alpha} \frac{\d}{\d x} p_\alpha(x) \frac{\d}{\d x}
\end{equation}
with smooth coefficients $W_\alpha = 1/(2\sqrt{\eta_\alpha}) > 0$ and $p_\alpha(x) = 1/\sqrt{\eta_\alpha}$\,.

The boundary $\pO= \{ a_1, b_1, \ldots, a_n,b_n\}$ is a discrete set of cardinality $\#(\pO)=2n$\,. Hence the Hilbert space at the boundary is finite-dimensional and therefore $$\H^0(\pO)\simeq\mathbb{C}^{2n}\;.$$ 
One particular property of the one-dimensional situation is that all the Sobolev spaces at the boundary are isomorphic to each other $$\H^k(\pO)\simeq \mathbb{C}^{2n}\quad \forall k\in \mathbb{R}\;.$$ These observations lead to the following proposition.

\begin{proposition}
	Let $\Omega$ be a one-dimensional, compact, Riemannian manifold. Any unitary operator acting on the Hilbert space of the boundary has gap at $-1$ and is admissible.
\end{proposition}

\begin{proof}
	Any unitary operator $U\in\mathcal{U}(2n)$ has gap at -1 since it can be represented by a unitary matrix and therefore has pure discrete spectrum. Recalling Definition \ref{def:admissible} we need only to show that it exists a constant $K>0$ such that $$\norm{U\varphi}_{\H^{1/2}(\pO)}\leq K\norm{\varphi}_{\H^{1/2}(\pO)}\quad \forall \varphi\in\H^{1/2}(\pO)\;,$$ but this is clearly satisfied since $$\norm{\cdot}_{\H^{1/2}(\pO)}\sim \norm{\cdot}_{\H^0}(\pO)$$ and $U$ is a bounded operator in $\H^0(\pO)$\,.
\end{proof}

The Lagrange boundary form in dimension 1, cf., Definition \ref{def: lagrangebf}, becomes now a sesquilinear form $$\Sigma:\mathbb{C}^{2n}\times \mathbb{C}^{2n}\to\mathbb{C}\;.$$
According to Proposition \ref{prop: asorey} their maximally isotropic subspaces $\mathcal{W}$ can be described in terms of a unitary operator $U\in\mathcal{U}(2n)$ as $$\mathcal{W}=\{(\varphi_1,\varphi_2)\in\mathbb{C}^{2n}\times \mathbb{C}^{2n}\mid \varphi_1-\mathbf{i}\varphi_2=U(\varphi_1+\mathbf{i}\varphi_2)\}\;.$$
By the previous proposition the unitary operator $U$ has gap at -1 and the condition 
\begin{subequations}\label{eq: 2n linear equations}
\begin{equation}\label{eq: asorey1}
	\varphi_1-\mathbf{i}\varphi_2=U(\varphi_1+\mathbf{i}\varphi_2)
\end{equation}
 is equivalent to the conditions
\begin{equation}\label{eq: asorey2}
	P\varphi_2= A_U\varphi_1\quad\text{and}\quad  P^\bot \varphi_1=0\;,
\end{equation}
\end{subequations}
where, as in Definition \ref{P,boundary} and Definition \ref{DefQU}, $P$ is the projection onto $W=\operatorname{Ran}E^\bot_{\{\pi\}}$\,, the orthogonal complement of the eigenspace associated to the eigenvalue -1 of the operator $U$\,,  and $P^\bot$ is the projection onto $W^\bot$\,. The linear system \eqref{eq: asorey2} constitutes a system of $2n$ equations.

\subsection{Finite elements for general self-adjoint boundary conditions}\label{finiteelementsforgeneralsabc}

As stated at the beginning of this chapter, to approximate the solutions of \eqref{weakspectralproblem} we shall construct a family of finite-dimensional subspaces of functions $\{S^N\}_N$ of $\H^0(\Omega)\subset \D_U$\,. Such finite-dimensional subspaces are constructed using finite elements.  The way in which we are going to  force our family of functions $\{S^N\}_N$ to be in $D_U$ is imposing the boundary conditions given by equations \eqref{eq: 2n linear equations}. In this case we will take 
$$\varphi_1=\gamma(\Phi_N)=\Phi_N|_{\pO}\;,$$ 
$$\varphi_2=\gamma(\d\Phi_N(\nu))=\frac{\partial\Phi_N}{\partial\nu}\Bigr|_{\pO}\;,$$
with $\Phi_N\in S^N$\,. Notice that with this choice the boundary term of the quadratic form \eqref{QFcha4} can be written as 
\begin{equation}
	\scalarb{\gamma(\Phi_N)}{A_U\gamma(\Phi_N)}=\scalarb{\gamma(\Phi_N)}{\gamma(\d\Phi_N(\nu))}=\scalarb{\varphi}{\dot{\varphi}}\;.
\end{equation}	
Recall that small greek letters denote the restriction to the boundary of the corresponding capital greek letters and that doted, small greek letters denote restrictions to the boundary of the normal derivatives. The subindex $N$ labelling the subspace $\{S^N\}$ will be dropped in the small greek letter notation as long as there is no risk of confusion.

The finite element model $(K,\mathcal{P}, \mathcal{N})$ that we use is given by the unit interval $K = [0,1]$\,, $\mathcal{P}$ the space of linear polynomials on $K$\,, and $\mathcal{N}$ the vertex set $\{ 0, 1\}$\,.     

The domain of our problem is the manifold $\Omega$ which consists of the disjoint union of the intervals $I_\alpha = [a_\alpha, b_\alpha]$\,, $\alpha = 1, \ldots, n$\,.   For each $N$ we will construct a non-degenerate subdivision $M^N$ as follows. Let $r_\alpha$ be the integer defined as $r_\alpha = \left[ L_\alpha N /L \right] + 1$\,, where $[ x ]$ denotes the integer part of $x$\,, $L_\alpha = |b_\alpha - a_\alpha|$\,, and $L = L_1 + \ldots + L_n$\,.   We will assume that each $r_\alpha \geq 2$\,, and $N \geq 2n$\,.
  Let us denote by $r$ the multi index $(r_1, \ldots, r_n)$\,.  Then $|r| = r_1+ \cdots + r_n$ satisfies
\begin{equation}\label{eq: Nr}
	 N \leq	|r | \leq N + n \;.
 \end{equation}
Now we will subdivide each interval $I_\alpha$ into $r_\alpha + 1$ subintervals of length: 
$$h_\alpha = \frac{L_\alpha} {r_\alpha + 1}\; .$$
The non-degeneracy condition supposes that it exists $\rho>0$ such that for all $I_{\alpha,a}\in M^N$ and for all $N>1$ $$\operatorname{diam}B_{I_{\alpha,a}}\geq \rho\operatorname{diam} I_{\alpha,a}\;,$$ where $B_{I_{\alpha,a}}$ is the largest ball contained in $I_{\alpha,a}$ such that $I_{\alpha,a}$ is star shaped with respect to $B_{I_{\alpha,a}}$\,. In our particular case this is satisfied trivially since $$\operatorname{diam}B_{I_{\alpha,a}}=\operatorname{diam} I_{\alpha,a}\;.$$
(It would be possible to use a set of independent steps $h_\alpha$\,, one for each interval; however, this could create some
technical difficulties later on that we prevent in this way.)
Each subinterval $I_\alpha$ contains $r_\alpha + 2$ nodes that will be denoted as
$$x^{(\alpha)}_k = a_\alpha + h_\alpha k, \quad  k = 0, \ldots, r_\alpha +1 \;.$$ 

\subsubsection{Bulk functions}
Consider now the family $\mathcal{F}_r$ of  $|r|-2n$ piecewise linear functions $$\{f_k^{(\alpha)}(x)\}_{k=2}^{r_\alpha- 1}\,,\quad\alpha = 1, \ldots, n\,,\; 2 \leq k \leq r_\alpha - 1\;,$$ that are zero at all nodes except at the $k$th node of the interval $I_\alpha$\,, where it has value $1$\,, or more explicitly,
$$ f^{(\alpha)}_k (x) = \left\{   \begin{array}{ll}
s\,, & x = x^{(\alpha)}_{k-1} + s h_\alpha, \quad 0 \leq s \leq 1\,, \\
1-s\,, & x = x^{(\alpha)}_{k} + s h_\alpha, \quad 0 \leq s \leq 1\,, \\
0\,, & \mathrm{otherwise}\;.
\end{array}
\right. $$
Notice that these functions are differentiable on each subinterval.   All these functions satisfy trivially the boundary conditions \eqref{eq: 2n linear equations} because they and their normal derivatives vanish at the endpoints of each interval.  They are localised around the inner nodes of the intervals.  We will call these functions \emph{bulk functions}.

\subsubsection{Boundary functions}\label{sectionboundaryfunctions}
We will add to the set of bulk functions a family of functions that implement nontrivially the boundary conditions determining the self-adjoint extension. These functions will be called \emph{boundary functions} and the collection of all of them will be denoted by $\mathcal{B}_r$\,.   
Contrary to bulk functions, boundary functions need to be ``delocalised'' so that they can fulfil any possible self-adjoint extension's boundary condition. 


\begin{figure}[ht]
	\centering
	\includegraphics[height=3cm]{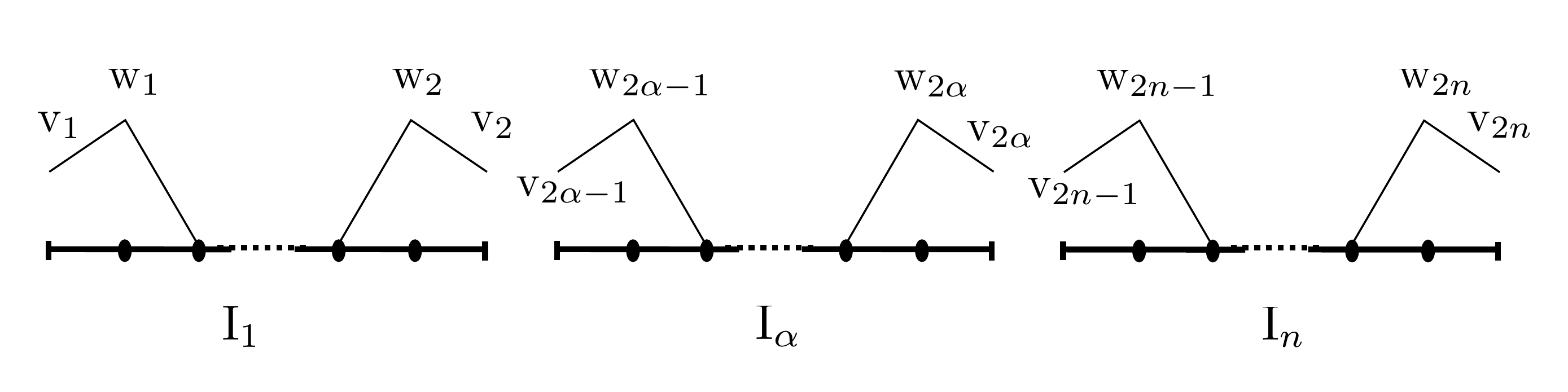}
	\caption{Boundary function $\beta^{(w)}$\,.}\label{function_beta}
\end{figure}


Because the endpoints $x_0^{(\alpha)} = a_\alpha$\,, $x_{r_\alpha + 1}^{(\alpha)} = b_\alpha$ of the intervals $I_\alpha$ and the adjacent nodes, $x_1^{(\alpha)}$ and $x_{r_\alpha}^{(\alpha)}$\,, are going to play a prominent role in what follows, we introduce some notation that takes care of them.   We will consider an index $l = 1, \ldots, 2n$ that labels the endpoints of the intervals. Now for each vector $w = (w_l) \in \mathbb{C}^{2n}$ consider the following functions (see Figure \ref{function_beta}):
$$ \beta^{(w)} (x) = \left\{   \begin{array}{ll}
v_{2\alpha -1} + s(w_{2\alpha -1} - v_{2\alpha -1})\,, & x = x^{(\alpha)}_{0} + s h_\alpha\,,  \\
w_{2\alpha -1}(1-s)\,, & x = x^{(\alpha)}_{1} + s h_\alpha\,, \\
s \,w_{2\alpha} \,, & x = x^{(\alpha)}_{r_\alpha -1} + s h_\alpha\,,  \\
w_{2\alpha} + s(v_{2\alpha} - w_{2\alpha}) \,, & x = x^{(\alpha)}_{r_\alpha} + s h_\alpha, \,,\\
0\,, & x_2^{(\alpha)} \leq x \leq x_{r_\alpha -1}^{(\alpha)} \,,
\end{array}\quad \begin{array}{c} 0 \leq s \leq 1\,,\\1 \leq \alpha \leq n\;.   \end{array}
\right. $$
Each function of the previous family is determined (apart from the vector $w$) by the vector $v = (v_l) \in \mathbb{C}^{2n}$ that collects the values of $\beta^{(w)}$ at the endpoints of the subintervals.    If we denote by $w^{(k)}$ the vectors such that $w^{(k)}_l = \delta_{lk}$\,, $k = 1, \ldots, 2n$\,, the $2n$ vectors $w^{(k)}$ are just the standard basis for $\mathbb{C}^{2n}$\,.   The corresponding functions $\beta^{(w)}$ will now be denoted simply by $\beta^{(k)}$\,. Notice that each boundary function $\beta^{(k)}$ is completely characterised by the unique non-extremal node where it does not vanish\footnote{If $k = 2\alpha-1$\,, it is the node $x^{(\alpha)}_1$\,, and if $k = 2\alpha$\,, it corresponds to the node $x^{(\alpha)}_{r_\alpha}$\,.}  and the values at the endpoints.  We denote by $v_l^{(k)}$\,, $l = 1, \ldots, 2n$\,, the boundary values of the functions $\beta^{(k)}$ above.

\subsubsection{The boundary matrix}
The $2n$ extremal values $v^{(k)}_{l}$ of the  boundary functions $\beta^{(k)}$ are undefined, but we are going to show that the $2n$ conditions \eqref{eq: 2n linear equations} imposed on the boundary functions constitute a determinate system of linear equations for them. 
Because the boundary functions are constructed to be piecewise linear, the normal derivatives of these functions at the boundary can be obtained easily.   For the left boundaries of the intervals, i.e., at the points $a_\alpha$\,, we have

\begin{equation} \label{normal_left}
\left. \frac{d\beta^{(k)}}{d\nu}\right|_{x = a_{\alpha}} = - \frac{1}{h_\alpha}(w^{(k)}_{2\alpha-1} - v^{(k)}_{2\alpha -1})\;,
\end{equation}
and respectively, for the right boundaries,

\begin{equation}\label{normal_right}
\left. \frac{d\beta^{(k)}}{d\nu}\right|_{x=b_{\alpha}} = \frac{1}{h_\alpha} (v^{(k)}_{2\alpha} - w^{(k)}_{2\alpha}) = -\frac{1}{h_\alpha} (w^{(k)}_{2\alpha} - v^{(k)}_{2\alpha}) \;.
\end{equation}
Thus the vector containing the normal derivatives of the function $\beta^{(k)}$\,, consistently denoted by $\dot{\beta}^{(k)}$\,, is given by
\begin{equation}\label{derivatives}
\dot{\beta}^{(k)}_l = - \frac{1}{h_l}(w^{(k)}_l - v^{(k)}_l) =  - \frac{1}{h_l}(\delta_{lk} - v^{(k)}_l) \;,
\end{equation}
where we use again the consistent notation $h_l = h_\alpha$\,, if $l = 2\alpha -1$\,, or $l = 2\alpha$\,.
For each boundary function $\beta^{(k)}$\,, the boundary conditions \eqref{eq: asorey1} read simply as the system of $2n$ equations on the components of the vector $v^{(k)}$\,,
$$ v^{(k)} - i \dot{\beta}^{(k)} = U(v^{(k)}  + i \dot{\beta}^{(k)} )  \;,$$
or, componentwise,
$$ v^{(k)}_l \left(1-\frac{i}{h_l} \right) + \frac{i}{h_l} w^{(k)}_l = \sum_{j = 1}^{2n}U_{l j} \left[v^{(k)}_j \left(1 + \frac{i} {h_j}\right) - \frac{i}{h_j}w^{(k)}_j \right]\;,$$ 
with $l = 1, \ldots, 2n$\,. Collecting coefficients and substituting the expressions $w^{(k)}_l=\delta_{lk}$ we get
\begin{equation}\label{boundaryvalues}
\sum_{j=1}^{2n}\left[\left(1-\frac{i}{h_j}\right)\delta_{lj}-U_{lj}\left(1+\frac{i}{h_j}\right) \right] v^{(k)}_{j} = \left[ -\frac{i}{h_k}\delta_{lk}-U_{lk}\frac{i}{h_k}\right]\;.
\end{equation}
This last equation can be written as the matrix linear system:
\begin{equation}\label{boundary_equation}
 F V = C
 \end{equation}
with $V$ a $2n \times 2n$ matrix whose entries are given by $V_{jk} = v^{(k)}_j$\,, $j,k = 1, \ldots, 2n$\,. The $k$th column of $V$ contains the boundary values of the boundary function $\beta^{(k)}$\,.  The $2n\times 2n$ matrix $F$ with entries 
$$F_{lj} = \left(1-\frac{i}{h_j}\right)\delta_{lj}-U_{lj}\left(1+\frac{i}{h_j}\right)\; ,$$
will be called the boundary matrix  of the subdivision of the domain $M$ determined by the integer $N$\,,
and 
$$ C_{lk} = -\frac{i}{h_k} (\delta_{lk}  +  U_{lk}) $$
defines the inhomogeneous term of the linear system \eqref{boundary_equation}.
Using a compact notation we get
$$ F = \mathrm{diag} (\mathbf{1} - i/\mathbf{h}) - U \mathrm{diag} (\mathbf{1} + i/\mathbf{h}) \,, \quad \quad C = -i\,  (I  + U) \mathrm{diag}(1/\mathbf{h})  \;,$$
where $1/\mathbf{h}$ denotes the vector whose components are $1/h_l$\,.
Notice that $F$ depends just on $U$ and the integer $|r|$ defining the discretisation of the manifold. 

\subsection{Conditioning of the boundary matrix}\label{conditioningoftheboundarymatrix}
Before addressing the construction of the approximate spectral problem we will study the behaviour of system \eqref{boundary_equation} under perturbations; in other words, we will compute the condition number of the boundary matrix $F$ and show that it is small enough to ensure the accuracy of the numerical determination of our family of boundary functions $\beta^{(i)}$\,.   The relative condition number we want to compute is
 $$ \kappa (F) = \norm{F}\norm{F^{-1}}\;.$$ 
In our case, the boundary matrix $F$ can be expressed as 
$$F = \bar{D} - U D = (I - U D \bar{D}^{-1}) \bar{D}$$ 
with $D_{jk} (\mathbf{h}) = D_{jk} = (1+\frac{i}{h_j})\delta_{jk}$\,.    Notice that the product $UD\bar{D}^{-1}$ is a unitary matrix which we will denote as $U_0(\mathbf{h})$ or simply $U_0$ if we do not want to emphasise the $\mathbf{h}$ dependence of $U_0$\,.  Thus, $F = (I - U_0)\bar{D}$ and
\begin{equation*}
\norm{F} = \norm{(I-U_0)\bar{D}} \leq \norm{I-U_0}  \norm{\bar{D}} \leq 2\norm{D}\;.
\end{equation*}
On the other hand, 
\begin{equation*}
\norm{F^{-1}} = \norm{\bar{D}^{-1}(I-U_0)^{-1}} \leq \norm{\bar{D}^{-1}}  \norm{(I-U_0)^{-1}} = \frac{\norm{D^{-1}}} {\min_{\lambda \in \operatorname{spec}(U_0)} \{ |1-\lambda| \} }
\end{equation*}
and thus we obtain
\begin{equation}\label{KF1}
 \kappa (F) \leq \kappa (D) \frac{2}{\min_{\lambda \in \operatorname{spec}(U_0)} \{ |1-\lambda| \} } \;.
 \end{equation}
As $D$ is a diagonal matrix its condition number is given by
$$\kappa (D) = \frac{\sqrt{\frac{1}{h^2_{\mathrm{min}}}+1}}{\sqrt{\frac{1}{h^2_{\mathrm{max}}}+1}} \leq \frac{h_{\max}}{h_{\min}}$$ 
with $h_{\max}$ ($h_{\min}$) the biggest (smallest) step of the discretisation determined by $N$\,.  We get finally,
\begin{equation}\label{KF}
\kappa (F) \leq \frac{h_{\max}}{h_{\min}}\frac{2}{|1-\lambda|}
\end{equation}
with $\lambda$ the closest element of the spectrum of $U_0$ to $1$\,.   Of course, because $U_0$ is unitary, it may happen that $1$ is in its spectrum, so that the condition number is not bounded.  
Because the matrix $U_0$ depends on $\mathbf{h}$\,, its eigenvalues will depend on $\mathbf{h}$ too.   
We want to study the dependence of the closest eigenvalue to 1, or 1 for that matter, with respect to perturbations of the vector $\mathbf{h}$\,. \\

\begin{lemma}\label{perturbation_1}
Suppose that $X_0$ is an eigenvector with eigenvalue $1$ of $U_0$ and that the perturbed matrix $\hat{U} = U_0+\delta U$\,, for $\norm{\delta U}$ small enough, is such that $1 \in \sigma (\hat{U})$\,. Then $\bar{X}_0^T\delta U X_0 = 0$ to first order in $\delta U$\,.
\end{lemma}
\begin{proof} Clearly, if $1 \in \sigma (\hat{U})$ and $\delta U$ is small enough, there exist a vector $\hat{X} = X_0 + \delta X$\,,
with $\norm{\delta X} \leq C \norm{\delta U}$\,, such that $\hat{U} \hat{X} = 1 \hat{X}$\,.  Then we have
\begin{equation}
U_0\delta X +\delta U X_0+\delta U \delta X = \delta X\notag .
\end{equation}
Because $U_0 X_0 = X_0$ and $U_0$ is unitary, $\bar{X}_0^T U_0 = \bar{X}_0^T$\,, and then
multiplying on the left by $\bar{X}_0^T$ and keeping only first order terms, we get the desired condition:
$\bar{X}_0^T\delta U X_0 = 0.$
\end{proof}

\medskip

Because of the previous lemma, if $\hat{U} = U_0 + \delta U$ is a unitary perturbation of $U_0$ such that $\bar{X}_0^T\delta U X_0 \neq 0$\,, 
for any eigenvector $X_0$ with eigenvalue 1, then $1 \notin \sigma(\hat{U})$\,.   Now if we consider a unitary perturbation $\hat{U}$ of $U_0$ such that $1 \notin \sigma (\hat{U})$ we want to estimate how far 1 is from the spectrum of $\hat{U}$\,.
Consider the eigenvalue equation for the perturbed matrix.  The perturbed eigenvalue $\hat{\lambda} = 1 + \delta \lambda$ will satisfy
\begin{equation}\label{eigenvalueunitary}
(U_0+\delta U)(X_0+\delta X)=(1+\delta \lambda)(X_0+\delta X)\;.
\end{equation}
Multiplying on the left by $\bar{X}_0^T$ and solving for $|\delta\lambda|$ it follows that 
\begin{equation}\label{bounddelta}
|\delta \lambda|\geq \frac{|X^H_0\delta U X_0 + X^H_0\delta U \delta X|}{|1+X^H_0\delta X|}\geq  \frac{|X^H_0\delta U X_0|-|X^H_0\delta U \delta X|}{1+\norm{\delta X}} 
\end{equation}
for $\norm{\delta U}$ small enough.
Taking into account the particular form of the matrix $U_0 = U D \bar{D}^{-1}$ we have that $\delta U = U\delta(D\bar{D}^{-1})$ and therefore $\norm{\delta U}=\norm{\delta (D\bar{D}^{-1})}$\,.   Moreover, because $\delta(D\bar{D}^{-1})$ is a diagonal matrix,
$\delta(D\bar{D}^{-1})_{kk} = \frac{-2i}{ (h_k - i)^2}\delta h_k$ and its singular values are the modulus of its diagonal entries.  Hence we have the following proposition.\\

\begin{proposition}\label{lower_bound} Given the matrix $U_0 = UD\bar{D}^{-1}$ with eigenvalue 1, then 1 is not an eigenvalue of any unitarily perturbed matrix $U + \delta U = U_0 (D \bar{D}^{-1} +  \delta(D \bar{D}^{-1}))$ with $\norm{\delta(D\bar{D}^{-1})}$ small enough and $\delta h = \min\{ |\delta h_k| \} > 0$\,.  Moreover there exists a constant $C>0$ such that the perturbation $\delta \lambda$ of such eigenvalue satisfies the lower bound,
$$|\delta \lambda|\geq \frac{\sigma_{\min}(\delta (D\bar{D}^{-1}))-C\sigma^2_{\max}(\delta (D\bar{D}^{-1}))}{1+C\sigma_{\max}(\delta (D\bar{D}^{-1}))} > 0.$$
\end{proposition}

\begin{proof}
Because of Lemma \ref{perturbation_1}  it is sufficient to show that $\bar{X}_0^T \delta(D\bar{D}^{-1})) X_0 \neq 0$\,.
But this is an easy consequence of the fact that $|\sum_{i= 1}^{2n} (|X_{0,i}|^2 \delta(D\bar{D}^{-1}))_{ii}| \geq 2  \delta h$\,.   Furthermore there exists a constant $C>0$ such that $\norm{\delta X}\leq C\norm{\delta U}$; hence taking $\norm{\delta(D\bar{D}^{-1})}$ small enough we get $\sigma_{\min}(\delta (D\bar{D}^{-1}))-C\sigma^2_{\max}(\delta (D\bar{D}^{-1})) > 0$ and the bound follows from \eqref{bounddelta}.
\end{proof}

Now we can apply Proposition \ref{lower_bound} to (\ref{KF}) and if we neglect terms $|h_i^2|\ll 1$ and $|\delta h_i|\ll |h_i|$ we finally get the desired bound for the condition number
\begin{equation}\label{kappa_h}
\kappa (F)  \leq \frac{h_{\max}}{h_{\min}}\frac{1}{\delta h} .
\end{equation}
 Then, if for a given $N$ we obtain a boundary matrix $F$ which is bad conditioned, it suffices to change the size of the discretisation, i.e., to increase $N$\,,  to improve the condition number.  Of course, if $N$ is already quite large, then the bound \eqref{kappa_h} could be useless.   For typical values $h\approx 10^{-2}\sim 10^{-3}$\,, it can be taken as $\delta h \approx 10^{-4} \sim 10^{-5}$ to provide condition numbers $\kappa(F) \approx 10^4\sim 10^5$\,.
 
\subsection{The spectral pencil}  For any $N > 2n$ we define the finite-dimensional approximation space $S^N$ as the linear span of the bulk and boundary functions, i.e., $S^N = \mathrm{span} \{ f_k^{(\alpha)}, \beta^{(l)} \mid \alpha = 1, \ldots, n, k = 2, \ldots, r_\alpha-1, l = 1, \ldots 2n  \}$\,.   All functions $f^{(\alpha)}_k$ and $\beta^{(l)}$ are linearly independent; thus the dimension of $S^N$ will be $|r| = r_1 + \ldots + r_n$\,. The relation of the index $N$ and the value $|r|$ is given by \eqref{eq: Nr}. It is convenient to rearrange the elements of the basis above as follows:
$$ \beta^{(1)}, f^{(1)}_1, \ldots, f^{(1)}_{r_1}, \beta^{(2)}, \beta^{(3)}, f^{(2)}_1, \ldots, f^{(2)}_{r_2}, \beta^{(4)}, \ldots, \beta^{(2n-1)}, f^{(n)}_1, \ldots, f^{(n)}_{r_n}, \beta^{(2n)} .$$
Using this ordering, we will rename the elements of this basis as $f_a$\,, with $a = 1, \ldots, |r|$\,, and an arbitrary element $\Phi_N \in S^N$ will be written as $ \Phi_N (x) = \sum_{a=1}^{|r|}\Phi _a f_a(x) .$
We consider now the approximate spectral problem \eqref{approximatespectralproblem}:
\begin{align*}
	Q(\Psi_N,\Phi_N)&=\scalar{\d\Psi_N}{\d\Phi_N}-\scalarb{\psi}{\dot{\varphi}}\\
		&=\scalar{\d\Psi_N}{\d\Phi_N}-\scalarb{\gamma(\Psi_N)}{\gamma(\d\Phi_N(\nu))}=\lambda\scalar{\Psi_N}{\Phi_N}
\end{align*}
%
Introducing the expansion above we get
\begin{equation}\label{spectral_bilinear}
\sum_{a,b}^{|r|}\Psi _a \Bigl[\scalar{\d f_a}{\d f_b} - \scalarb{\gamma(f_a)}{\gamma(\d f_b(\nu))} - \lambda \scalar{\d f_a}{\d f_b} \Bigr] \Phi _b = 0 . 
\end{equation}
As \eqref{spectral_bilinear} holds for every $\Psi _N\in S^N$\,, this equation is equivalent to the eigenvalue equation of the matrix pencil $A-\lambda B$
\begin{equation}\label{pencil}
A=\lambda B,
\end{equation}
where 
\begin{align*}
A_{ab} &= \scalar{\d f_a}{\d f_b} - \scalarb{\gamma(f_a)}{\gamma(\d f_b(\nu))}\\
	&=\scalar{\d f_a}{\d f_b} - \scalarb{\gamma(f_a)}{A_U\gamma(f_b)}\;,
\end{align*}
$$B_{ab} = \scalar{\d f_a}{\d f_b}\;.$$
Notice that $A$ and $B$ are both Hermitian matrices, which improves the numerical algorithms used to compute the eigenvalues of the pencil.
In fact, when solving numerically \eqref{pencil}, it is relevant to preserve its Hermitian character. Notice that
the boundary functions $\beta^{(l)}(x)$ satisfy 
\begin{equation}\label{constraint}
\scalarb{\beta^{(l)}}{\dot{\beta}^{(m)}} - \scalarb{\beta^{(m)}}{\dot{\beta}^{(l)}} = 0\;,
\end{equation}
because their boundary values are elements of a maximally isotropic subspace of the Lagrange boundary form, i.e., they satisfy \eqref{eq: asorey1}.  Using \eqref{derivatives} and the definition of the boundary values of the boundary functions codified in the matrix $V$ we have
$$\scalarb{\beta^{(l)}}{\dot{\beta}^{(m)}}=\sum_{k=1}^{2n}\frac{1}{h_{k}}\bar{V}_{kl}(V_{km}-\delta_{km})
= \sum_{k=1}^{2n}\frac{1}{h_k}\bar{V}_{kl}V_{km}-\frac{1}{h_m}\bar{V}_{ml}\;.$$
%
This identity together with \eqref{constraint} leads to 
\begin{equation}\label{constraint2}
 \frac{1}{h_j}\bar{V}_{jk}=\frac{1}{h_k}V_{kj}.
\end{equation}
The Hermitian relation \eqref{constraint2} is satisfied by the numerical solutions of \eqref{boundary_equation} up to roundoff errors and consequently the pencil \eqref{pencil} is Hermitian only up to these roundoff errors.   We will force the numerical solution of matrix $V$ to satisfy \eqref{constraint2} so that the Hermiticity of the pencil is preserved exactly.  This is convenient  because the algorithms for solving the general eigenvalue problem are much better behaved in the Hermitian case \cite[Chapter 5]{demmel97}. 

To end this discussion we must realise that, with the basis $f_a$ for $S^N$ we have just constructed, the matrices $A$ and $B$ are almost tridiagonal and the unique elements different from zero, besides the tridiagonal ones, are those related to the matrix elements of the boundary functions.  In fact, we can consider a number of cases.  If the function $f_a$ is an interior bulk function, i.e., not corresponding to the node $x^{\alpha}_2$ or $x^{(\alpha)}_{r_\alpha -1}$\,, it is obvious that the only nontrivial inner products $\langle f_a, f_b \rangle$ and $\langle f'_a , f'_b \rangle$ will correspond to $b = a-1,a, a+1$\,.   If the function $f_a$ is an extreme bulk function, for instance, $f_2^{(\alpha)}$\,, then it has nontrivial inner products only with $\beta^{(2\alpha-1)}$ and $f_3^{(\alpha)}$\,, and if $f_a$ is now a boundary function $\beta^{(l)}$\,, then the only non-vanishing inner products will be with the other boundary functions and an extreme bulk function, namely, $f_2^{(\alpha)}$ if $l=2\alpha-1$ or $f_{r_\alpha-1}^{(\alpha)}$ if $l=2\alpha$\,.


\subsection{The family $\{S^N\}_{_N}$ is an approximating family of $Q_U$} 

We conclude this section by showing that the family of subspaces $\{S^N\}_N$ constructed in the previous subsections is indeed an approximating family of $Q_U$\,. The results of section \ref{convergenceofthenumericalscheme} show that the solutions of the spectral pencil \eqref{pencil} do approximate the solutions of the spectral problems for the self-adjoint extensions of the Laplace-Beltrami operator $-\Delta_U$\,.

\begin{theorem}\label{RG}
The closure of the union of all the finite-dimensional spaces $S^N$ in the Sobolev norm of order 1 is the domain $\D_U$ of the quadratic form $Q_U$\,, 
$$\overline{ \displaystyle{\cup_{N>1} S^N} }^{\norm{\cdot}_{1}}=\D_U\;.$$
\end{theorem}

\begin{proof} 
First notice that for every $N$ the functions in $S_N$ verify the same boundary condition than the functions in $\D_U$\,, namely $$P^\bot\gamma(\Phi_N)=0\;.$$ Hence $S^N\subset\D_U$\,.

Now let $P_N:\H^0(\Omega)\to S^N$ be the orthogonal projection onto $S^N$\,. It is enough to show that $\cup_{N>1} \operatorname{Ran} P_N$ is dense in $\H^1(\Omega)$\,.
Let $\{\Omega^N\}_N$\,, $0\leq\frac{1}{N}\leq 1$\,, be the non-degenerate family of subdivisions of the manifold $\Omega\subset\mathbb{R}$ defined in subsection \ref{finiteelementsforgeneralsabc}. It consists of a collection of $|r|+n$ closed subintervals $I_{\alpha,a}$\,, $a=1,	\dots,r_\alpha+1$\,, of the real line. Each subinterval is a representation of the reference element $K=[0,1]$ with nodal set $\mathcal{N}=\{0,1\}$\,. According to subsection \ref{finiteelementsforgeneralsabc}, the family of functions $\mathcal{P}$ consists of the space of linear polynomials in $K$\,. Let $\mathcal{P}_m$ be the space of polynomials of degree $m$\,. Then the reference element $(K,\mathcal{P},\mathcal{N})$ satisfies the following for $m=1,2$\, and $l=0$:
\begin{itemize}
	\item $K$ is star-shaped with respect to some ball.
	\item $\mathcal{P}_{m-1}\subseteq \mathcal{P}\subseteq W^{m,\infty}(K)$\,.
	\item The nodal variables $\mathcal{N}$ involve derivatives up to order $l$\,.
\end{itemize}
For all $I_{\alpha,a}\in M^N$ let $(K,\mathcal{P}_{\alpha,a},\mathcal{N}_{\alpha,a})$ be the affine-equivalent element.   Suppose that $1<p<\infty$ and $m-l-1/p>0$\,. Then, according to  \cite[Theorem 4.4.20]{brenner08}, there exists a constant $C$\,, depending on the reference element, $m$\,, and $p$ such that for $0\leq s\leq m$\,,
\begin{equation}\label{inequality_gen}
	\left(\sum_{I_{\alpha,a}\in M^N}\norm{\Psi-P_N\Psi}^p_{\H^{s}(I_{\alpha,a})}\right)^{1/p}\leq C N^{s-m}\norm{\Psi}_{\H^{m}(\Omega)}.
\end{equation}
Particularising for the case $s=1$\,, $m=2$\,, $p=2$\,, inequality (\ref{inequality_gen}) becomes
\begin{equation}\label{bound}
	\norm{\Psi -P_N\Psi}_{\mathcal{H}^1(\Omega)}\leq \frac{C}{N}\norm{\Psi}_{\mathcal{H}^2(\Omega)} .
\end{equation}
Since $\H^2(\Omega)$ is a dense subset in $\H^1(\Omega)$ the inequality above ensures that for every $\Psi\in\D_U$ the sequence $P_N\Psi\subset S^N$  satisfies that $$\lim_{N\to\infty}\norm{\Psi-P_N\Psi}_{\H^1(\Omega)}=0\;.$$
\end{proof}
Notice that the proof of this theorem also provides an estimation for the convergence rate of the numerical scheme. In fact the result in \cite[Theorem 4.4.20]{brenner08}, quoted in the proof of the theorem above, shows that the error in the discretisation measured in the $\H^1$-norm decreases as $1/N$\,. Among other issues, in the next section we are going to show how this bound is satisfied for a particular instance of the numerical algorithms described in this section.


\section{Numerical experiments}\label{Numerical_experiments_and_conclusions}

The numerical scheme described in Section \ref{FEM} results in the finite dimensional eigenvalue problem of \eqref{pencil}.  Furthermore, we have that the error of the solution measured in the $\H^1$-norm decreases as $1/N$\,.  Hence, if we were able to solve \eqref{pencil} for increasing grid size $N$\,, we would get better and better approximations to the eigenvalue problem. As remarked in the previous section, the pencil is almost tridiagonal and both matrices $A$ and $B$ are Hermitian.  Therefore the resulting problem is algebraically well behaved and it should lead to accurate results.  

We will now discuss some numerical experiments that illustrate the stability and the convergence of the algorithm.  We will also compare these results with those obtained by using two other algorithms (not based on the finite element method). In the latter cases we will use a particular choice of boundary conditions close to the singular case described in \cite{berry08}.  
We will consider the case of the Laplacian in $\Omega=[0,2\pi]$ subject to different boundary conditions:
\begin{equation}\label{spinlessfree}
-\frac{d^2}{dx^2}\Psi=\lambda \Psi\;.
\end{equation}
Notice that in this simple case there is only one interval and $|r|=N+1$\,, so we can use $|r|$ or $N$ interchangeably.  The number $h=\frac{2\pi}{r+1}$ is going to be the length of each subinterval.
After some straightforward computations we find that the matrices $A$ and $B$ defining the spectral pencil $A -\lambda B$ associated to (\ref{spinlessfree}) are given by

\begin{equation}
			\begin{array}{l} 
				a_{kk}=2,\quad k=2,\dots, r-1\,,\\
				a_{k \, k+1}=a_{k+1\, k}=-1,\quad k=1,\dots,r-1\,,\\
				a_{11}=2-V_{11}\,,\\
				a_{rr}=2-V_{22}\,,\\
				a_{1r}=-V_{12}\,,\\
				a_{r1}=\overline{a_{1r}}=-V_{21}\,,\\
				a_{kj}=0,\quad j\geq k+2,\quad j\leq k-2\,,\\
				A_{kj}=\frac{1}{h}a_{kj}\,,
			\end{array}  
\end{equation}
\begin{equation}			
			\begin{array}{l}
				b_{kk}=4,\quad k=2,\dots, r-1\,,\\
				b_{k\, k+1}=b_{k+1\, k}=1,\quad k=1,\dots,r-1\,,\\
				b_{11}=4+2[|V_{11}|^2+|V_{21}|^2]+2V_{11}\,,\\
				b_{rr}=4+2[|V_{22}|^2+|V_{12}|^2]+2V_{22}\,,\\
				b_{1r}=2[\bar{V}_{11}V_{12}+\bar{V}_{21}V_{22}]+2\,,\\
				b_{r1}=\overline{b_{1r}}\,,\\
				b_{kj}=0,\quad j\geq k+2,\quad j\leq k-2\,,\\
				B_{kj}=\frac{h}{6}b_{kj}\;.
			\end{array}
\end{equation}

In each case, in order to obtain the matrix $V$\,, one needs to solve previously the corresponding system of equations \eqref{boundary_equation} for the given self-adjoint extension. The solutions of these generalised eigenvalue problems have been obtained using the \textrm{Octave} built-in function $\mathbf{eig}$\,. The details of this routine can be found in \cite{anderson92}. This built-in function is a general-purpose diagonalisation routine that does not exploit the particularly simple algebraic structure of this problem. One could adapt the diagonalisation routine to the algebraic structure of the problem  at hand (there are only two elements outside the main diagonals) to improve the efficiency. Moreover, one could also use the $p$-version of the finite element method that, considering that the solutions are smooth, would be more adaptive.   However our main objective here is to show that the computation of general self-adjoint extensions by using non-localised finite elements at the boundary, as explained in Subsection \ref{sectionboundaryfunctions}, is reliable and accurate.  We consider that the results explained below account for this, and we leave these improvements for future work.

First we will test the stability of the method against variations of the input parameters.  The parameters of this procedure are the matrix $U$ determining the self-adjoint extension whose eigenvalue problem we want to solve.  We will perturb an initial self-adjoint extension, described by a unitary matrix $U$\,, and we will observe the behaviour of the eigenvalues. In other words, we are interested now in studying the relation 
				\begin{equation}
					|\Delta \lambda|=K(\varepsilon)\norm{\Delta U}\;,
				\end{equation}
where $\epsilon$ is the parameter measuring the size of the perturbation. If the algorithm were stable one would expect that the condition number $K(\varepsilon)$ would grow at most polynomially with the perturbation $\epsilon$\,.  However, we must be careful in doing so since the exact eigenvalue problem presents divergences under certain circumstances (explained below) which could lead to wrong conclusions.  In fact, as a consequence of Theorem \ref{maintheorem1}, we see that when a self-adjoint extension is parameterised by a unitary matrix $U$ that has eigenvalues close but not equal to $-1$\,, it happens that some eigenvalues of the considered problem take very large negative values. However, matrices with $-1$ in the spectrum can lead to self-adjoint extensions that are positive definite, for example, Dirichlet or Periodic self-adjoint extensions.  Thus, following a path in the space of self-adjoint extensions, it could happen that a very small change in the arc parameter leads to an extremely large jump in the exact eigenvalues.  Such self-adjoint extensions are precisely the ones that lead in higher dimensions to the problem identified by M.~Berry as a Dirichlet singularity\footnote{Notice again that $U=-\mathrm{1}$ is the unitary matrix describing Dirichlet boundary conditions.} \cite{berry08},\cite{berry09}, \cite{marletta09}. In fact these boundary conditions can be described by unitary operators that do not have gap at -1. The results discussed in chapter \ref{cha:QF} show in fact that such boundary conditions can lead to truly unbounded self-adjoint extensions of the Laplace-Beltrami operator, i.e., not semi-bounded, and they will be the target of our latter tests. For proving the stability it is therefore necessary to perturb the unitary matrix along a direction of its tangent space such that the gap condition is not a jeopardy.  A path in the space of self-adjoint extensions where the aforementioned jumps in the spectrum do not occur is, for instance, the one described by the so-called quasi-periodic boundary conditions \cite{asorey83}. These correspond to the unitary operator described in Example \ref{ex:quasiperiodic}. In this case, the self-adjoint domain is described by functions that satisfy the boundary conditions $\Psi(0)=e^{i2\pi\epsilon}\Psi(2\pi)$ and $\Psi'(0)=e^{i2\pi\epsilon}\Psi'(2\pi)$\,, which correspond to the family of unitary matrices 
\begin{equation}\label{unitariaquasiperiodicas}
	U(\epsilon)=\begin{pmatrix} 0& e^{i2\pi\epsilon}\\e^{-i2\pi\epsilon}&0 \end{pmatrix}\;.
\end{equation}
Notice that this particular choice of boundary conditions, which are non-local in the sense that they mix the boundary data at both endpoints of the interval, can naturally be treated by the discretisation procedure introduced in Section \ref{FEM} and that goes beyond the ones usually addressed by most approximate methods \cite{aceto09}, \cite{chanane99}, \cite{farrington57}, \cite{pruess73}, described by local equations of the form $\alpha\cdot \Psi(a)+\beta\cdot \Psi'(a)=0$\,, $\gamma\cdot \Psi(b)+\delta\cdot \Psi'(b)=0$\,.

Let us then consider perturbations of the periodic case in the quasi-periodic direction, i.e., we consider 
$$U(\varepsilon)\simeq U+i\varepsilon A=\begin{pmatrix} 0 & 1\\1 & 0 \end{pmatrix}+i\varepsilon \begin{pmatrix} 0 & 1\\ -1 & 0\end{pmatrix}\;.$$  We have calculated the numerical solutions for values of $\varepsilon$ between $10^{-4}$ and $10^{-1}$ in steps of $10^{-4}$ and the discretisation size used was $N=250$\,. In the latter case the perturbation is $\norm{\Delta U}=\norm{i\varepsilon A}=\varepsilon$\,, hence the absolute error ratio is $$K(\varepsilon)=\frac{|\Delta \lambda|}{\norm{\Delta U}}=\frac{1}{\varepsilon}|\Delta \lambda|\;.$$ The results are plotted in figure \ref{estabilidad}. 
\begin{figure}[h]\centering
\includegraphics[height=3in]{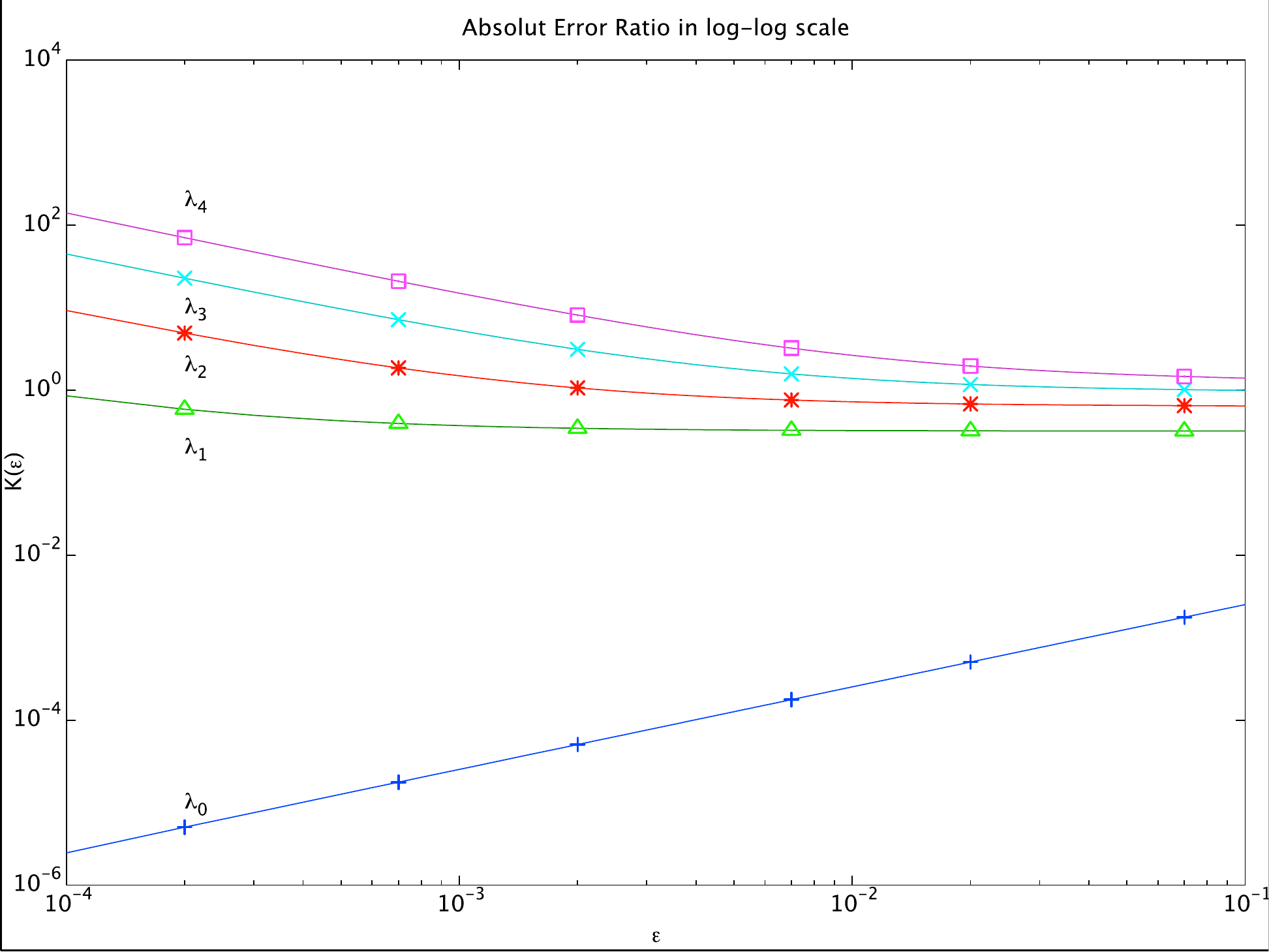}
\caption{Absolute error ratio $K(\varepsilon)$ of the five lowest levels for the periodic boundary problem plotted against $\varepsilon$ in log-log scale.}\label{estabilidad}
\end{figure}
As can be seen, $K(\varepsilon)$ is a decreasing function of $\varepsilon$ except for the fundamental level. In the latter case the absolute ratio is an increasing function of $\varepsilon$\,. However, in this case it can be seen clearly that the growth is linear. This result shows that the procedure is stable under perturbations of the input matrix $U$\,.\\

We will consider now the Laplacian \eqref{spinlessfree} subjected to quasi-periodic boundary conditions:
\begin{equation}
	\Psi(0)=e^{i2\pi\epsilon}\Psi(2\pi)\,, \quad
	\Psi'(0)=e^{i2\pi\epsilon}\Psi'(2\pi) \;.
\end{equation}
These are codified by the unitary matrix \eqref{unitariaquasiperiodicas}. This is a meaningful example since it demonstrates that this algorithm provides a new way to compute the Bloch decomposition of a periodic Schr\"odinger operator. This problem is addressed usually by considering the unitary equivalent problem $-\left( \frac{d}{d x}+i\epsilon \right)^2 \Psi = \lambda\Psi $ 
with periodic boundary conditions. However, our procedure is able to deal with it directly in terms of the original boundary condition.

The analytic solutions for this particular eigenvalue problem can be obtained explicitly \cite{asorey83}. They are $\lambda_n=(n+\epsilon)^2$\,, \smash{$\psi_n=\frac{1}{\sqrt{2\pi}}e^{-i(n+\epsilon)x}$}, $n=0,\,\pm1,\,\dots$\,, and we can compare the approximate solution obtained by the procedure described in Section \ref{FEM}. In particular we show that the bound \eqref{bound} is satisfied, i.e., that the error between the approximate solutions and the analytic ones measured in the $\mathcal{H}^1$-norm is of the form $1/N$\,. The results for the five lowest eigenvalues corresponding to $\epsilon=0.25$ are shown in figure \ref{convergencia}.\footnote{In order to avoid errors coming from numerical quadratures, the solutions of the integrals appearing in the $\mathcal{H}^1$-norms of each subinterval, comparing the analytic solutions and the linear approximations, have been computed explicitly.} As is known for the finite element approximations, the error grows with the order of the eigenvalue; however, it is clearly seen that for all the cases the decay law is of the form $1/N$\,, therefore satisfying the bound.\\

\begin{figure}[h]\centering
\includegraphics[height=3in,width=4.5in]{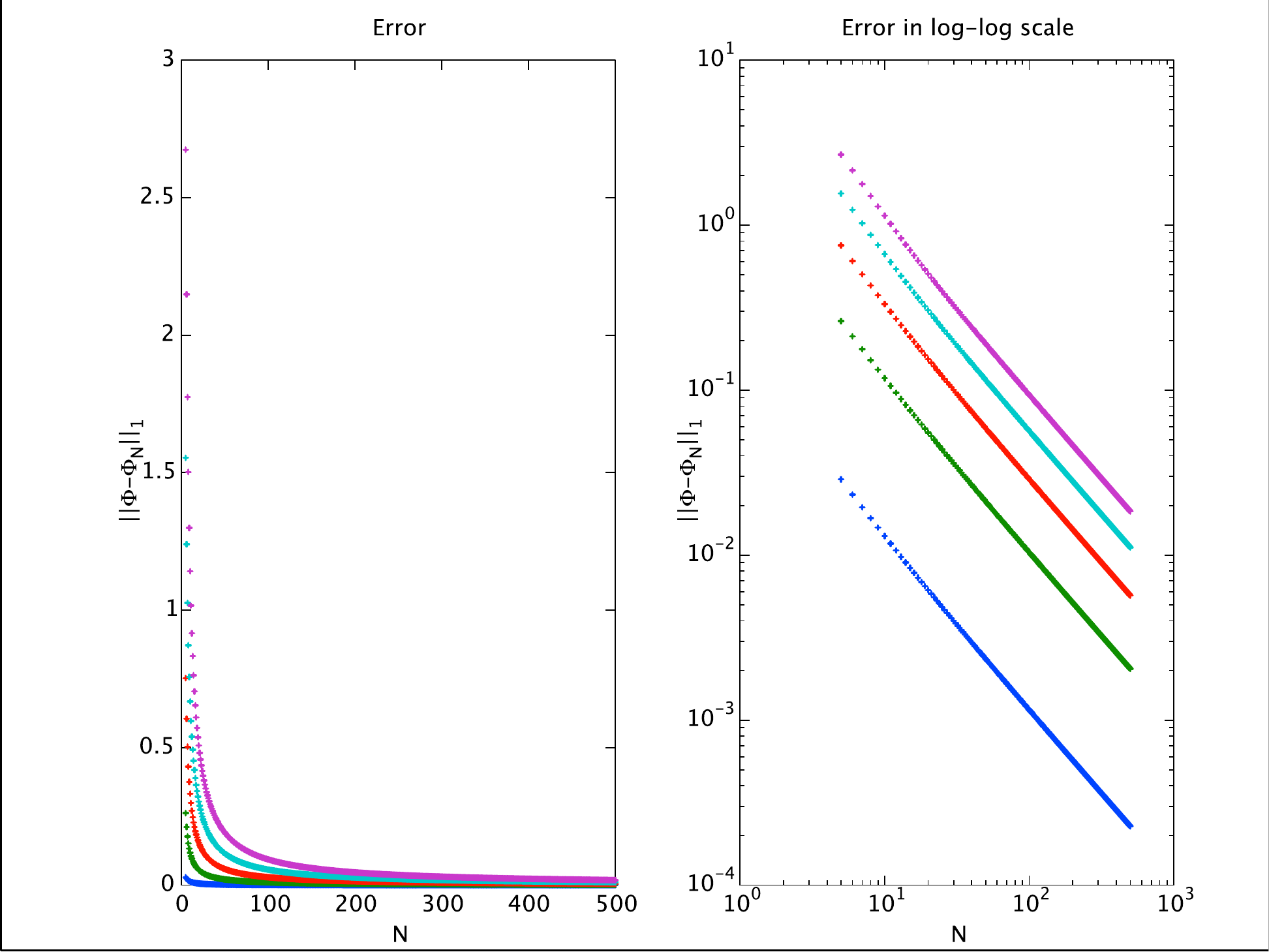}
\caption{Evolution of the error, measured in the $\mathcal{H}^1$-norm, for the five lowest levels of the quasi-periodic Laplacian problem ($\epsilon=0.25$) with increasing lattice size plotted in normal scale and log-log scale.}\label{convergencia}
\end{figure}

Finally we consider again the Laplacian, but in this case subjected to the local boundary conditions
\begin{equation}\label{condicionesDsingularity}
		\begin{array}{c}
			\Psi'(0)=0\,,\\
			\tan(-\theta/2)\cdot \Psi(2\pi)+\Psi'(2\pi)=0\;,
		\end{array}
\end{equation}
which are determined by the unitary matrix
\begin{equation}\label{unitariaDsingularity}
	U=\begin{pmatrix} 1 & 0 \\ 0 & e^{-i\theta} \end{pmatrix}\;.
\end{equation}
This is a particular case of Example \ref{generalized Robin}. Moreover, this is the problem of the interval addressed in Definition \ref{Defunidimensional}. These boundary conditions, unlike the previous ones, can be handled by most of the software packages available for the integration of Sturm-Liouville problems.   As is stated in Proposition \ref{prop: intervalrobin}, when $\theta=\pi$ the Laplace operator is positive, but for values of $\theta<\pi$ the fundamental level is negative with increasing absolute value as $\theta\to\pi$\,. It happens that for values of $\theta\simeq\pi$ the absolute value of this fundamental eigenvalue is several orders of magnitude bigger than the closest eigenvalue.

The solutions of the discrete problem, according to Theorem \ref{ConvergenciaSoluciones}, are guaranteed to converge to solutions of the exact problem for increasing lattice size. However, it is not necessary that the sequence of eigenvalues obtained in the approximate solution is in correspondence with the sequence of eigenvalues of the exact problem. It may happen that for some threshold $N$ some \emph{new} eigenvalues appear that were not detected for smaller $N$\,. The big gap between the two lowest levels in the self-adjoint extensions described above is a good example of this feature. We have computed the spectrum for a fixed value $\theta=0.997\pi$ for $N$ in increasing steps. For each value of $N$ we show the lowest five eigenvalues returned by the procedure. The results are plotted in Figure \ref{ConvergenciaAutovaores}. Notice that for $N\leq1300$ the negative fundamental level is not detected. However, for $N=1400$ an approximation of the negative eigenvalue is returned ($\sim-6000$). Now, in this situation and for $N>1400$\,, the second to fifth lowest eigenvalues coincide with the lowest four returned for the situations below the threshold. The convergence of the negative eigenvalue is also visible. The scale in the negative $y$-axis has been rescaled so that the performance could be better appreciated.
\begin{figure}[h]\centering
\includegraphics[height=3.25in]{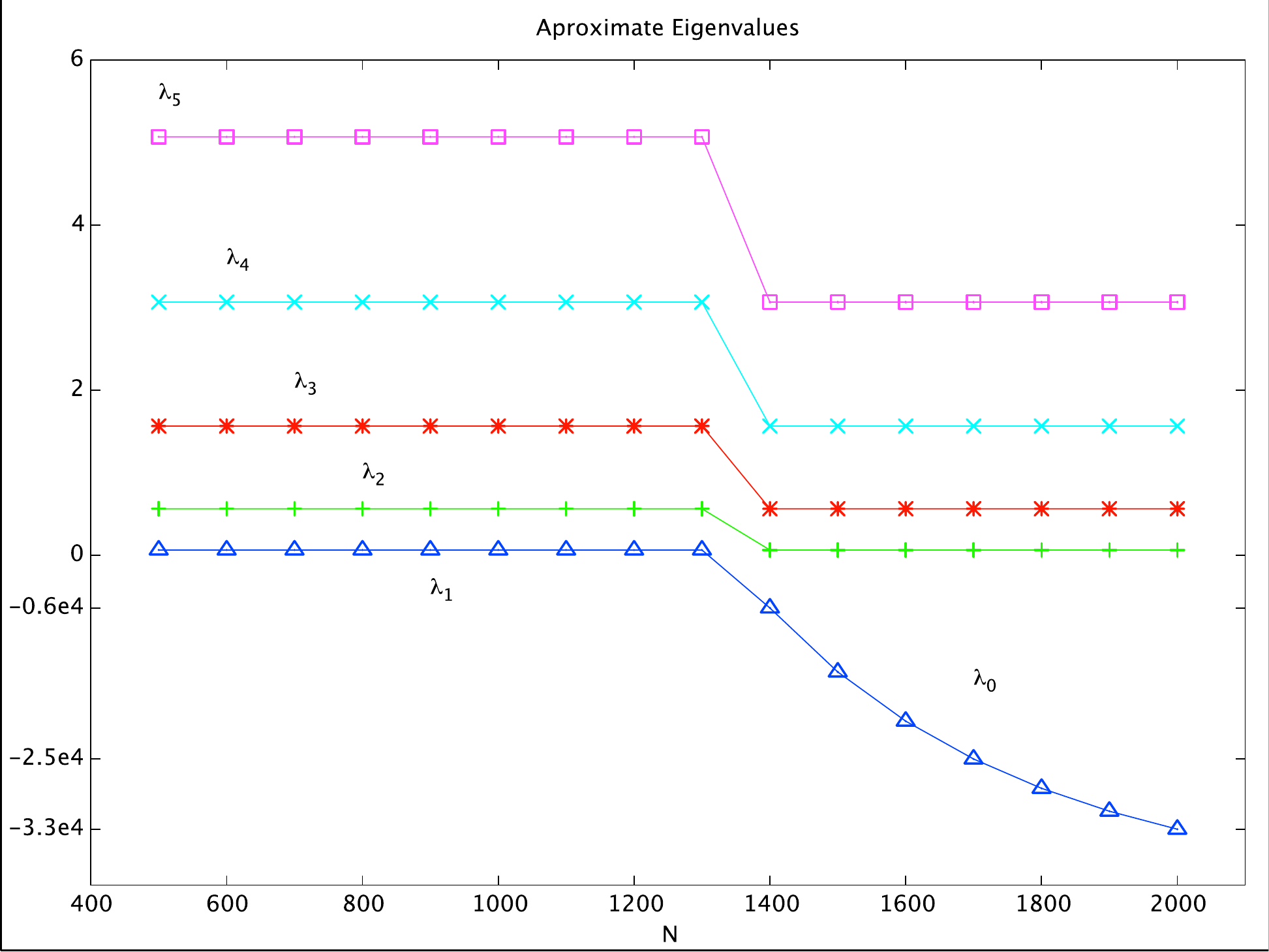}
\caption{Five lowest eigenvalues for $\theta=0.997\pi$ and increasing lattice size.}\label{ConvergenciaAutovaores}
\end{figure}
\\

As stated at the beginning of this section, we are now going to compare the results obtained with the numerical scheme described so far, from now on referred to as FEM, with two other algorithms that are not based on the finite element method, namely the line-based perturbation method (LPM) and the constant reference potential perturbation method (CPM) proposed in \cite{ledoux06} and \cite{ledoux05}, respectively. The software used for the calculations is, in the latter cases, a ready-to-run version provided by the authors.\footnote{http://www.twi.ugent.be.} The objective is now to find the solutions for the Schr\"odinger problem with the boundary conditions given by \eqref{condicionesDsingularity}.

The LPM routine reported errors for the lowest order eigenvalues for $\theta\geq 2.6$ and was not able to produce any output to compare. The results comparing the CPM and FEM routines for the Laplacian are shown in Figures \ref{laplaciano2} to \ref{laplaciano4}. The discretisation size used for all the FEM calculations was $N=5000$\,. Both methods provide almost the same results for the first excited states (indices 1 to 5); however, this is not the case for the fundamental level (Figure \ref{laplaciano2}), where it is expected that the eigenvalues take increasingly large negative values for $\theta$ approaching $\pi$\,. Note that the lowest value achieved with FEM is $\lambda_0\simeq -7\cdot 10^4$\,. Clearly the CPM routine fails to reproduce the correct eigenvalues in this situation, although it provides good approximations for $\theta\leq 0.989\pi$\,. In all these cases the CPM routine issued warnings on the low reliability of the results and an initial estimate for the fundamental eigenvalue was necessary.  Although the numerical eigenvalues are in agreement with the expected ones, except for the fundamental level, the eigenfunctions plotted in figures \ref{laplaciano3} and \ref{laplaciano4} for the case $\theta=3.1$ show clearly that the solutions provided by FEM are more accurate. For instance, the solutions obtained with CPM present discontinuities near the middle of the interval, although they are expected to be smooth functions. The eigenfunction obtained for the fundamental level (Figure \ref{laplaciano4}) is especially remarkable. This function does not present a singularity in the boundary; however, the $x$-axis has been enlarged so that the localisation at the boundary could be better appreciated. The FEM solution is a clear example of an \emph{edge state}. These are eigenfunctions that are associated with negative eigenvalues and that are strongly localised at the boundary of the system, while they vanish in the interior of the manifold (the interior of the interval in this case). These \textrm{edge states} are important in the understanding of certain physical phenomena like the quantum Hall effect (see \cite{wen92} and references therein). The bad behaviour of the eigenfunctions obtained with the CPM routine could be due to the fact that, after an initial estimation of the corresponding eigenvalue, it uses a numerical integrator to propagate the solution from one end of the interval to the other. Then it uses the differences between the obtained boundary conditions and the given ones to provide a new starting point and propagate it back. This procedure continues iteratively until convergence while continuity is imposed in the centre of the interval. This last statement could explain why the bad behaviour always appears in the middle of the interval. It needs to be said that the results using CPM do not depend on the special form that one selects to implement the boundary conditions. If one uses one of the equivalent forms
\begin{equation}
		\begin{array}{c}
			\Psi'(0)=0\\
			\Psi(2\pi)+\cot(-\theta/2)\cdot \Psi'(2\pi)=0
		\end{array} 
\end{equation}
		or
\begin{equation}		
		\begin{array}{c}
			\Psi'(0)=0\\
			\sin(-\theta/2)\cdot \Psi(2\pi)+\cos(-\theta/2)\cdot\Psi'(2\pi)=0
		\end{array} 
\end{equation}
the results remain the same.
The CPM routine in general worked much faster than the FEM routine in all the performed calculations.\footnote{All the numerical calculations of this section were performed with a laptop computer with an Intel Core i5 processor at $2.53$ GHz with 4 GB DDR3 RAM.} However, as stated earlier, the algebraic routine used to solve the generalised eigenvalue problem is not adapted to the structure of the problem and other finite element schemes could be used to improve convergence and efficiency.

\begin{figure}[h]\centering
\includegraphics[height=3.3in,width=4.3in]{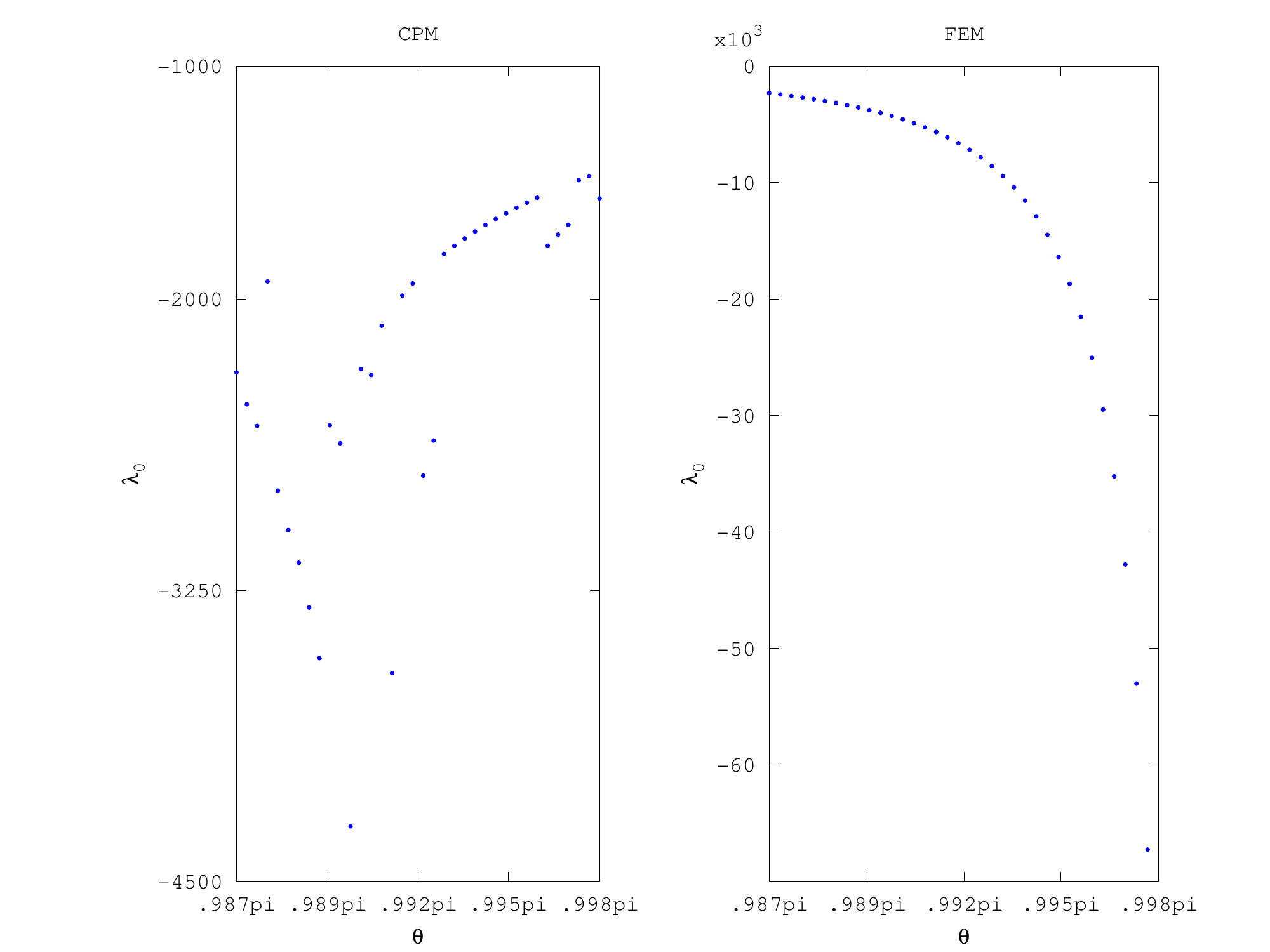}
\caption{Groundlevel eigenfunctions of the Laplacian for increasing values of $\theta$\,.}\label{laplaciano2}
\end{figure}

\begin{figure}[h]\centering
\includegraphics[height=2.6in]{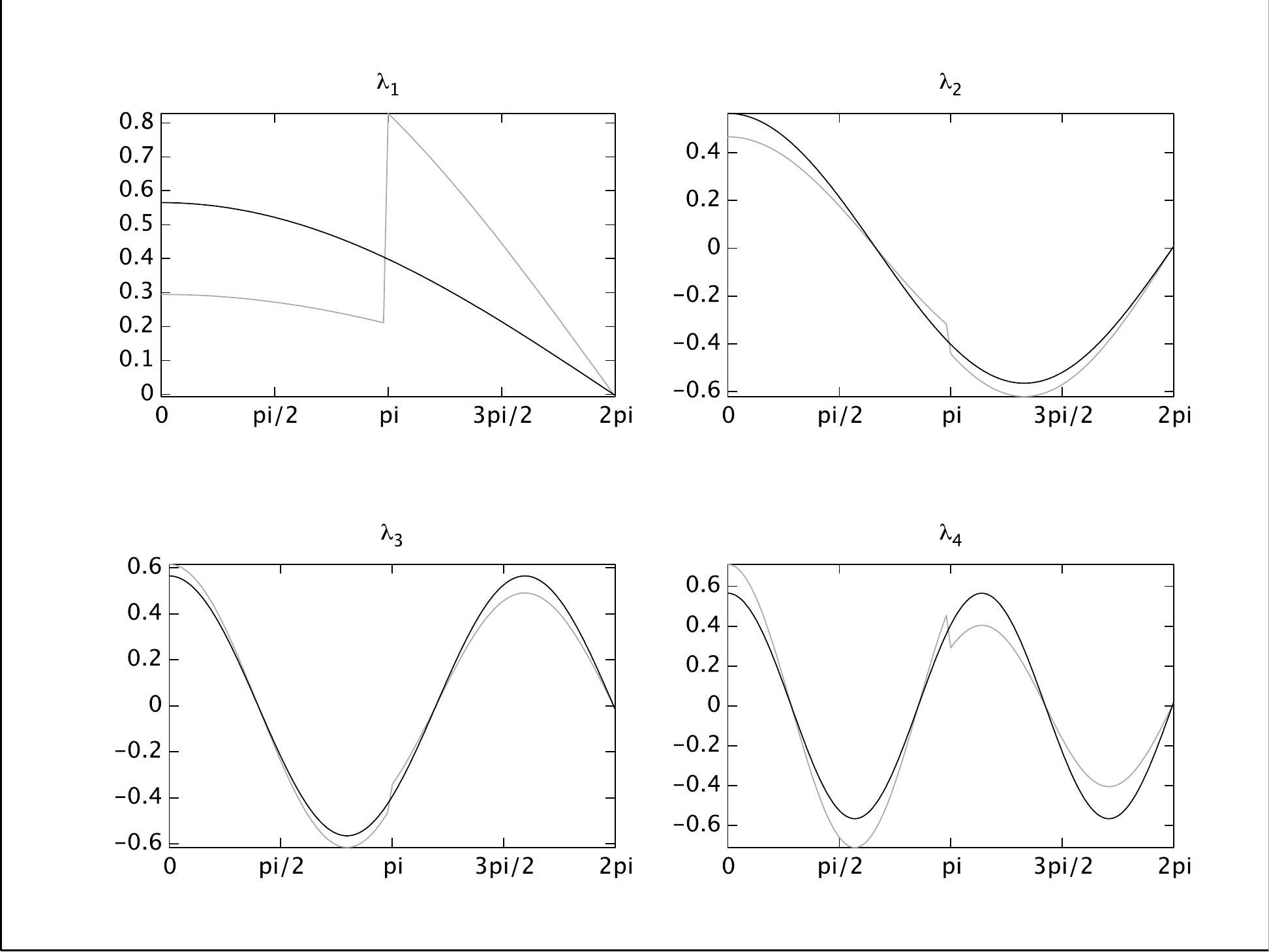}
\caption{First excited eigenfunctions for the Laplacian for $\theta=3.1$\,. FEM functions are plotted in black. CPM functions are plotted in grey.}\label{laplaciano3}
\end{figure}

\begin{figure}[h]\centering
\includegraphics[height=2.7in]{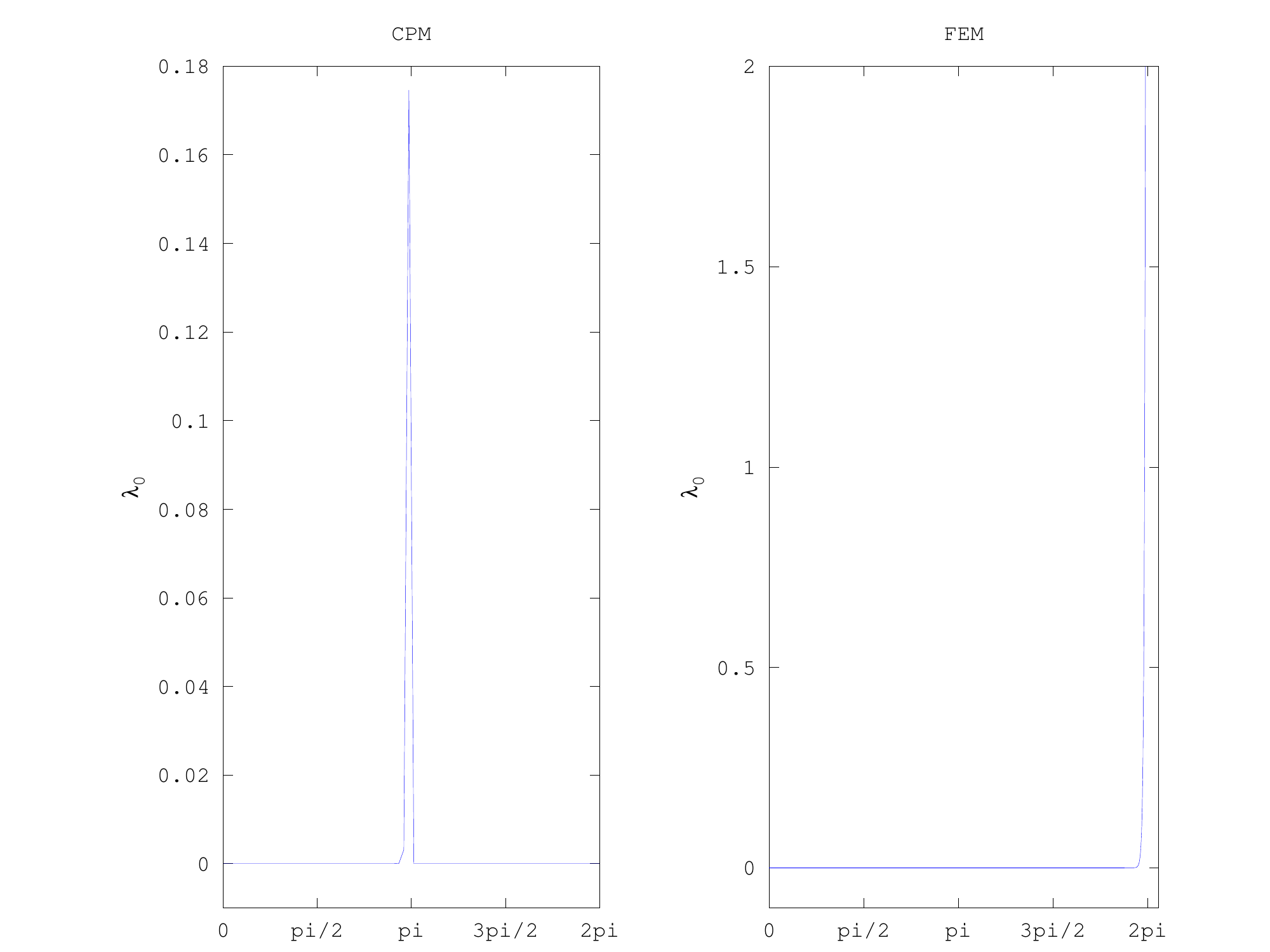}
\caption{Groundlevel eigenfunctions of the Laplacian problem for $\theta=3.1$\,.}\label{laplaciano4}
\end{figure}

\FloatBarrier

\clearemptydoublepage
\chapter{Self-Adjoint Extensions with Symmetry and Representation of Quadratic Forms}
\markboth{S.-A. Extensions with Symmetry and Representation of Q.F.}{}

\label{cha:Symmetry}


As we have discussed in the introduction, in most occasions quantum systems are constructed starting with a densely defined symmetric operator $T$ whose self-adjoint extensions will define either physical observables or unitary dynamical evolution.   It is necessary then to understand under what conditions a group acting on the Hilbert space by unitary transformations will become a symmetry group for the corresponding extensions of the operator, or which self-adjoint extensions will have $G$ as a symmetry group.   We will discuss these problems in Secs. 5.1 and 5.2 both from the point of view of the general theory of self-adjoint extensions of symmetric operator and the theory of quadratic forms.  

In the particular instance that our quantum system is of mechanical type, we are led necessarily to the study of self-adjoint extensions of the Laplace-Beltrami operator as we have discussed exhaustively.    In many interesting examples and applications, there is a group of transformations $G$ acting on the configuration manifold $\Omega$ by isometries, this is, preserving the Riemannian metric $\eta$.   Then, this group of transformations has a canonical unitary representation $V$ on all Sobolev spaces defined on the manifold by using the Laplace-Beltrami operator. Moreover, the operator $\Delta_{\mathrm{min}}$ satisfies:
$$ V(g) \Delta_{\mathrm{min}} V(g)^\dagger = \Delta_{\mathrm{min}} \;.$$
We say that $\Delta_{\mathrm{min}}$ is $G$-invariant. We shall introduce this notion more precisely in Section \ref{sec:vonNeumannSymmetry}. The question is whether or not the different self-adjoint extensions of $\Delta_{\mathrm{min}}$ will share this property. It is a common error in the literature to assume that all the self-adjoint extensions will also be $G$-invariant. We will show that this is not true in general.

In general, a densely defined self-adjoint operator $T$ does not have to be semi-bounded from below. Hence, its associated quadratic form $Q$ is not semi-bounded from below either. It is an open problem to characterise those Hermitean quadratic forms that are not semi-bounded and that do admit a representation in terms of a self-adjoint operator. We will introduce a family of not necessarily semi-bounded quadratic forms, that we will call partially orthogonally additive, and that can be represented in terms of self-adjoint operators. We also will introduce the notion of sector of a quadratic form. This notion will play the role for quadratic forms that the invariant subspaces do in the context of self-adjoint operators.

This chapter is organised as follows. In Section \ref{sec:vonNeumannSymmetry} we introduce the main definitions and give an explicit characterisation of the self-adjoint extensions that are $G$-invariant in the most general setting, i.e., using the characterisation due to von Neumann described  in Section \ref{sec:vonNeumann}. In Section \ref{sec:InvariantQF} we introduce the notion of $G$-invariant quadratic form and relate them with the notion of $G$-invariant operator. In Section \ref{sec:InvariantLB} and Section \ref{sec:SymmetryExamples} we analyse the quadratic forms associated to the Laplace-Beltrami operator, as introduced in Chapter \ref{cha:QF}, when there is a compact Lie Group acting on the manifold. Thus, we provide a characterisation of the self-adjoint extensions of the Laplace-Beltrami operator that are $G$-invariant. Finally, Section \ref{sec:Representation} is devoted to the generalisation of Kato's representation theorem, Theorem \ref{fundteo}, for quadratic forms that are not semi-bounded.


\section{General theory of self-adjoint extensions with symmetry}
\label{sec:vonNeumannSymmetry}

In this section we shall assume that we have a Hilbert space $\H$ and a group $G$\,, not necessarily compact, that possesses a unitary representation $V$ acting on $\H$\,, $$V:G\to\mathcal{U}(\H)\;.\newnot{symb:Vg}$$

\begin{definition}
	Let $T$ be a linear operator with dense domain $\D(T)$ and let $V:G\to\mathcal{U}(\H)$ be a unitary representation of the group $G$\,. The operator $T$ is said to be \textbf{$G$-invariant} if\newline $V(g)\D(T)\subset\D(T)$ for all $g\in G$ and $$\comm{T}{V(g)}\Psi=0\quad \forall g\in G,\;\forall \Psi\in\D(T)\;,$$ where $\comm{\cdot}{\cdot}$ stands for the formal commutator of operators, i.e., $$\comm{A}{B}=AB-BA\;.$$
\end{definition}

\begin{proposition}\label{adjointinvariant}
	Let $T$ be a  G-invariant, symmetric operator. Then the adjoint operator $T^{\dagger}$ is $G$-invariant.
\end{proposition}

\begin{proof}
	Let $\Psi\in\D(T^\dagger)$\,. Then, according to Definition \ref{adjointoperator}, it exists $\chi\in\H$ such that $$\scalar{\Psi}{T\Phi}=\scalar{\chi}{\Phi}\quad \forall \Phi\in\D(T)\;.$$ Now we have that 
	\begin{align*}
		\scalar{V(g)\Psi}{T\Phi}&=\scalar{\Psi}{V(g^{-1})T\Phi}\\
			&=\scalar{\Psi}{TV(g^{-1})\Phi}\\
			&=\scalar{\chi}{V(g^{-1})\Phi}\\
			&=\scalar{V(g)\chi}{\Phi}\;.
	\end{align*}
	The above equalities are true $\forall\Phi\in\D(T)$ and hence $V(g)\Psi\in\D(T^\dagger)$\,. Moreover, we have that $T^\dagger V(g)\Psi=V(g)\chi=V(g)T^\dagger \Psi\;.$
\end{proof}

\begin{corollary}
 Let $T$ be a $G$-invariant and symmetric operator on $\mathcal H$. Then its closure
 $\overline T$ is also $G$-invariant.
\end{corollary}
\begin{proof}
 The operator $T$ is symmetric and therefore closable and we have $\overline{T}=T^{**}$\,. 
 Since $T$ is $G$-invariant, by the preceding proposition we have that $T^*$ is also $G$-invariant,
 hence also $\overline{T}=(T^*)^*$\,.
\end{proof}

This lemma shows that we can always assume without loss of generality that the $G$-invariant
symmetric operators are closed.

\begin{corollary}\label{cor:invariantdeficiency}
	Let $T$ be a symmetric, $G$-invariant operator. Then, the deficiency spaces $\mathcal{N}_{\pm}$, cf., Definition \ref{def:deficiencyspaces}, are invariant under the action of the group, i.e., $$V(g)\mathcal{N}_{\pm}=\mathcal{N}_{\pm}\;.$$
\end{corollary}

\begin{proof}
	Let $\xi\in\mathcal{N}_+\subset \D(T^\dagger)$\,. Then $(T^\dagger-i)\xi=0$ and we have that
	\begin{equation*}
		(T^\dagger-i)V(g)\xi=V(g)(T^\dagger-i)\xi=0\;.
	\end{equation*}
	This shows that $V(g)\mathcal{N}_{+}\subset \mathcal{N}_{+}$\,, for all $g \in G$\;. Now $$\mathcal{N}_+ = V(g) (V(g^{-1}) \mathcal{N}_+) \subset V(g)\mathcal{N}_+$$ which shows the equality.
	Similarly for $\mathcal{N}_-$\,.
\end{proof}

With the same arguments as in the final part of the preceding proof we have the following general result:

\begin{lemma}\label{subsetequal}
	Let $\D$ be any subspace of $\H$ that remains invariant under the action of the group $G$\,, i.e., $$V(g)\D\subset\D\,, \forall g\in G\;.$$ Then $V(g)\D=\D$\,.
\end{lemma}

\begin{theorem}\label{Ginvariantoperator}
	Let $T$ be a symmetric, $G$-invariant operator with equal deficiency indices, cf., Definition \ref{def:deficiencyspaces}. Let $T_K$ be the self-adjoint extension of $T$ defined by the unitary $K:\mathcal{N}_+\to\mathcal{N}_-$\,. Then $T_K$ is $G$-invariant iff $\comm{V(g)}{K}\xi=0$ for all $\xi\in\mathcal{N_+}$\,, $g\in G$\,.
\end{theorem}

\begin{proof}\noindent
	\begin{itemize}
		\item[$(\Leftarrow)$]According to Theorem \ref{thmvonNeumann} the domain of $T_K$ is given by $$\D(T_K)=\D(T)+\bigl(\mathbb{I}+K \bigr)\mathcal{N}_+\;.$$ Let $\Psi\in\D(T)$ and $\xi\in\mathcal{N}_+$\,. Then we have that 
		\begin{align*}
			V(g)\bigl( \Psi + \bigl(\mathbb{I}+K\bigr)\xi \bigr)&=V(g)\Psi+ \bigl(V(g)+V(g)K\bigr)\xi\\
				&=V(g)\Psi + \bigl(V(g)+KV(g)\bigr)\xi\\
				&=V(g)\Psi+ \bigl(\mathbb{I}+K\bigr)V(g)\xi\;.
		\end{align*}
		By assumption $V(g)\Psi\in\D(T)$, and by Corollary \ref{cor:invariantdeficiency}, $V(g)\xi\in\mathcal{N_+}$\,. Hence $V(g)\D(T_K)\subset\D(T_K )$\,.
		Moreover, we have that 
		\begin{align*}
			T_KV(g)\Psi&=T^\dagger V(g)\bigl(\Psi+\bigl(\mathbb{I}+K\bigr)\xi\bigr)\\
				&=TV(g)\Psi+T^\dagger V(g)\bigl(\mathbb{I}+K\bigr)\xi\\
				&=V(g)T\Psi + V(g)T^\dagger\bigl(\mathbb{I}+K\bigr)\xi=V(g)T_K\Psi\;,
		\end{align*}
		where we have used Proposition \ref{adjointinvariant}.
		
		\item[$(\Rightarrow)$] Suppose that we have the self-adjoint extension defined by the unitary $$K'=V(g)KV(g)^\dagger\;.$$
		If we consider the domain $\D(T_{K'})$ defined by this unitary we have that
		\begin{align*}
			\D(T_{K'})&=\D(T) + \bigl(\mathbb{I}+V(g)KV(g)^\dagger\bigr)\mathcal{N}_+\\
				&=V(g)\D(T) + V(g)\bigl(\mathbb{I}+K\bigr)V(g)^\dagger\mathcal{N}_+\\
				&=V(g)\D(T_K)=\D(T_K)\;,
		\end{align*}
		where we have used again Proposition \ref{adjointinvariant}, Corollary \ref{cor:invariantdeficiency} and Lemma \ref{subsetequal}. Now von Neumann's theorem, Theorem \ref{thmvonNeumann}, establishes a one-to-one correspondence between isometries $K:\mathcal{N}_+\to\mathcal{N}_-$ and self-adjoint extensions of the operator $T$\,. Therefore $K=K'=V(g)KV(g)^\dagger$ and the statement follows.
	\end{itemize}
\end{proof}

We want to finish this section by pointing out that the important condition for a group to be a symmetry of a self-adjoint extension, i.e., for a self-adjoint extension $T_K$ to be $G$-invariant, is that the unitary representation of the group  leaves the domain $\D(T_K)$ invariant. The commutation property follows from this latter one.


\section{Invariant quadratic forms}
\label{sec:InvariantQF}

It is time now to introduce the analogue concept to $G$-invariant operators for quadratic forms. As before we  assume that there is a group $G$ that has a unitary representation in the Hilbert space $\H$\,.

\begin{definition}
	Let $Q$ be a quadratic form with domain $\D$ and let $V:G\to\mathcal{U}(\H)$ be a unitary representation of the group $G$\,. We will say that the quadratic form is \textbf{$G$-invariant} if $V(g)\D\subset\D$  for all $g\in G$ and $$Q(V(g)\Phi)=Q(\Phi)\quad\forall \Phi\in\D, \forall g\in G\;.$$
\end{definition}

It is clear by the polarisation identity ,cf., Eq. \eqref{polarization id}, that if $Q$ is $G$-invariant, then $Q(V(g)\Phi,V(g)\Psi)=Q(\Phi,\Psi)$\,, $g\in G$\,. We can now relate the two notions of $G$-invariance, for operators and for quadratic forms.

\begin{theorem}\label{QAinvariant}
	Let $Q$ be a closed, semi-bounded quadratic form with domain $\D$ and let $T$ be the associated semi-bounded, self-adjoint operator. The quadratic form $Q$ is $G$-invariant iff the operator $T$ is $G$-invariant.
\end{theorem}

\begin{proof}\noindent
	\begin{itemize}
		\item[$(\Rightarrow)$]  By Theorem \ref{fundteo}, $\Psi\in\D(T)$ iff $\Psi\in\D$ and it exists $\chi\in\H$ such that $$Q(\Phi,\Psi)=\scalar{\Phi}{\chi}\quad\forall \Phi\in\D\;.$$ Let $\Psi\in\D(T)$\,. Now we have that
			\begin{align*}
				Q(\Phi,V(g)\Psi)&=Q(V(g)^\dagger\Phi,\Psi)\\
					&=\scalar{V(g)^\dagger\Phi}{\chi}=\scalar{\Phi}{V(g)\chi}\;.
			\end{align*}
			This implies that $V(g)\Psi\in\D(T)$\,. Moreover, it implies that $$TV(g)\Psi=V(g)\chi=V(g)T\Psi\,,\quad \Psi\in\D(T)\,,g\in G\;.$$
		\item[$(\Leftarrow)$] To prove this implication we are going to use he fact that $\D(T)$ is a core for the quadratic form. Let $\Phi,\Psi\in\D(T)$\,. Then we have that $$Q(\Phi,\Psi)=\scalar{\Phi}{T\Psi}=\scalar{V(g)\Phi}{V(g)T\Psi}=Q(V(g)\Phi,V(g)\Psi))\;.$$ These equalities show that the $G$-invariance condition is true at least for the elements in the domain of the operator. Now let $\Psi\in\D$\,. Then it exists $\{\Psi_n\}\in\D(T)$ such that $\normm{\Psi_n-\Psi}_Q\to 0$\,. This, together with the equality above, implies that $\{V(g)\Psi_n\}$ is a Cauchy sequence with respect to $\normm{\cdot}_{Q}$\,. Since $Q$ is closed, the limit of this sequence is in $\D$\,. Moreover it is clear that $\H-\lim_{n\to\infty}V(g)\Psi_n=V(g)\Psi$\,, which implies that $\normm{V(g)\Psi_n-V(g)\Psi}_Q\to0$\,. 
		
		So far we have proved that $V(g)\D\subset\D$\,. Now let $\Phi,\Psi\in\D$ and let $\{\Phi_n\},\{\Psi_n\}\subset\D(T)$ respectively be sequences converging to them in the norm $\normm{\cdot}_{Q}$\,. Then
		\begin{align*}
		Q(\Phi,\Psi)&=\lim_{n\to\infty}\lim_{m\to\infty}Q(\Phi_n,\Psi_m)\\
			&=\lim_{n\to\infty}\lim_{m\to\infty}Q(V(g)\Phi_n,V(g)\Psi_m)=Q(V(g)\Phi,V(g)\Psi)\;.
		\end{align*}
	\end{itemize}
\end{proof}

This result allows us to reformulate Kato's representation theorem, Theorem \ref{fundteo}, in the following way.

\begin{theorem}\label{InvariantRepresentation}
	Let $Q$ be a $G$-invariant, semi-bounded quadratic form with lower bound $a$ and domain $\D$\,. The following statements are equivalent:
	\begin{enumerate}
		\item There is a $G$-invariant, lower semi-bounded, self-adjoint operator $T$ on $\H$ that represents the quadratic form, i.e., $$Q(\Phi,\Psi)=\scalar{\Phi}{T\Psi}\quad\forall \Phi\in\D,\forall \Psi\in \D(T)\;.$$
		\item The quadratic form is closed and its domain is a Hilbert space with respect to the inner product
		$$\scalar{\Phi}{\Psi}_Q=(1+a)\scalar{\Phi}{\Psi}+Q(\Phi,\Psi)\;.$$ 
	\end{enumerate}
\end{theorem}

\begin{proof}
	The equivalence is direct application of Kato's representation theorem (Theorem \ref{fundteo}) and Theorem \ref{QAinvariant}\,. 
\end{proof}

Recall the scales of Hilbert spaces introduced in Section 2.4.

\begin{theorem}
Let $Q$ be a closed, semi-bounded, $G$-invariant quadratic form with lower bound $a$\,. Then
	\begin{enumerate}
		\item $V$ restricts to a unitary representation on $\mathcal{H}_+:=\mathcal{D}(Q)\subset\mathcal{H}$ with scalar product given by
$$\scalar{\Phi}{\Psi}_+:=\scalar{\Phi}{\Psi}_Q=(1+a)\scalar{\Phi}{\Psi}+Q(\Phi,\Psi)\,,\quad \Phi,\Psi\in\H_+\;.$$

		\item $V$ extends to a unitary representation on $\mathcal{H}_-$ and we have, on $\mathcal{H}_-$\,,
\begin{equation}\label{eq:commute-I}
	V(g) I=I V(g)\;,\quad g\in G\;,
\end{equation}
where $I\colon \mathcal{H}_-\to \mathcal{H}_+$ is the canonical isometric bijection of Definition \ref{def:scales}.
	\end{enumerate}
\end{theorem}
\begin{proof}
i) To show that the representation $V$ is unitary with respect to $\scalar{\cdot}{\cdot}_+$
    note that by definition of $G$-invariance of the quadratic form we have 
    for any $g\in G$ that $V(g)\colon\mathcal{H}_+\to \mathcal{H}_+$ and 
$$
 \scalar{V(g)\Phi}{V(g)\Psi}_+=\scalar{\Phi}{\Psi}_+\,.
$$
   Since any $V(g)$ is invertible we conclude that $V$ restricts to a unitary representation on $\mathcal{H}_+$\,.


(ii) To show that $V$ extends to a unitary representation on $\H_-$ consider first the following representation of $G$ on $\H_-$\,:

$$\widehat{V}(g)\alpha:=I^{-1}V(g)I\alpha\,,\quad\alpha\in\H_-\;.$$
This representation is unitary. Indeed, 
\begin{align*}
	\scalar{\widehat{V}(g)\alpha}{\widehat{V}(g)\beta}_-&=\scalar{I^{-1}V(g)I\alpha}{I^{-1}V(g)I\beta}_-\\
		&=\scalar{V(g)I\alpha}{V(g)I\beta}_+=\scalar{I\alpha}{I\beta}_+=\scalar{\alpha}{\beta}_-\;.
\end{align*}
The restriction of $\widehat{V}(g)$, $g\in G$\,, to $\H$ coincides with $V(g)$. Let $\Phi\in\H$ and consider the pairing of Proposition \eqref{proppairing}. Then, 
\begin{align*}
	\pair{\widehat{V}(g)\Phi}{\Psi}&=\pair{I^{-1}V(g)I\Phi}{\Psi}\\
		&=\scalar{V(g)I\Phi}{\Psi}_+\\
		&=\scalar{I\Phi}{V(g^{-1})\Psi}_+\\
		&=\pair{\Phi}{V(g^{-1})\Psi}\\
		&=\scalar{\Phi}{V(g^{-1})\Psi}\\
		&=\scalar{V(g)\Phi}{\Psi}=\pair{V(g)\Phi}{\Psi}\;,\quad\forall \Psi\in\H_+\;.
\end{align*}
Since $\widehat{V}(g)$ is a bounded operator in $\H_-$ and $\H$ is dense in $\H_-$, $\widehat{V}(g)$ is the extension of $V(g)$ to $\H_-$\,.
\end{proof}

The preceding theorem shows that the $G$-invariance of the quadratic form is equivalent to the existence of unitary representations
on the scale of Hilbert spaces
\[
 \mathcal{H}_+\subset \mathcal{H} \subset\mathcal{H}_-\;.
\]


\section{A class of invariant self-adjoint extensions of the Laplace-Beltrami operator}
\label{sec:InvariantLB}

In this section we are going to analyse the class of quadratic forms introduced in Chapter \ref{cha:QF} when there is a group of symmetries acting on the manifold. Throughout the rest of the section we are going to consider that $G$ is a compact group acting in the Riemannian manifold $(\Omega,\pO,\eta)$\,. We are going to consider that the action of the group is compatible with the restriction to the boundary. This is, the action of $G$ on $\Omega$ maps the boundary $\pO$ into itself, or in other words, the action of $G$ on $\Omega$ induces and action of $G$ on $\pO$ considered as a subset of $\Omega$. Moreover, we are going to assume that the group acts by isometric diffeomorphisms. More concretely, let $g:\Omega\to\Omega$ be the diffeomorphism associated to the element $g\in G$\,, that we will denote with the same symbol for simplicity of notation. Then we have that $g^*\eta=\eta$\,, where $g^*$ stands for the pull-back by the diffeomorphism. Automatically $G$ acts by isometries on the boundary $\partial \Omega$, i.e., $g^* \eta_{\partial\Omega} = \eta_{\partial\Omega}$ for all $g\in G$.  These isometric actions of the group $G$ induce unitary representations of the group on $\Omega$ and $\pO$\,. Let the unitary representations be denoted as follows:
\begin{equation*}
	V:G\to\mathcal{U}(\H^0(\Omega))\;,
\end{equation*}
\begin{equation*}
	V(g)\Phi=(g^{-1})^*\Phi\quad \Phi\in\H^0(\Omega)\;.
\end{equation*}
\begin{equation*}
	\mathrm{v}:G\to\mathcal{U}(\H^0(\pO))\;,\newnot{symb:v(g)}
\end{equation*}
\begin{equation*}
	\mathrm{v}(g)\varphi=(g^{-1})^*\varphi\quad \varphi\in\H^0(\pO)\;.
\end{equation*}
Then we have that
\begin{align*}
	\scalar{V(g^{-1})\Phi}{V(g^{-1})\Psi}&=\int_\Omega (\overline{\Phi\circ g})(\Psi\circ g)\,\d\mu_{\eta}\\
		&=\int_\Omega (\overline{\Phi\circ g})(\Psi\circ g)\, g^*\negthinspace\d\mu_{\eta}\\
		&=\int_{g\Omega}\overline{\Phi} \Psi \,\d\mu_{\eta}\\
		&=\scalar{\Phi}{\Psi}\;,
\end{align*}
where we have used the change of variables formula and the fact that $g^*\negthinspace\d\mu_\eta=\d\mu_{\eta}$\,. The result for the boundary is proved similarly. The induced actions are related with the trace map as follows 
\begin{equation}\label{VPhi=vphi}
	\gamma(V(g)\Phi)=\mathrm{v}(g)\gamma(\Phi)\quad\forall g\in G\;.
\end{equation}

Let us consider the family of quadratic forms associated to the Laplace-Beltrami operator introduced in Chapter \ref{cha:QF}. Consider the quadratic form 
\begin{equation}\label{QFcha5}
	Q_U(\Phi,\Psi)=\scalar{\d\Phi}{\d\Psi}-\scalarb{\gamma(\Phi)}{A_U\gamma(\Phi)}
\end{equation}
with domain
\begin{equation}\label{domaincha5}
	\D_U=\bigl\{ \Phi\in\H^1(\Omega)\bigr|P^{\bot}\gamma(\Phi)=0 \bigr\}\;,
\end{equation}
where $A_U$ is the partial Cayley transform of Definition \ref{partialCayley}. We are considering that the unitary $U$ has gap at $-1$ and the projector $P$ is the orthogonal projector onto the invertibility boundary space $W=\operatorname{Ran}E^{\bot}_{\pi}$\,, where $E_\lambda$ is the spectral resolution of the identity associated to the unitary $U$.
We are going to show the necessary and sufficient conditions for this quadratic form to be $G$-invariant. In order to do so we are going to need the next result.

\begin{proposition}\label{prop:dphiinvariant}
	Let $G$ be a compact group that acts by isometric diffeomorphisms in the Riemannian manifold $\Omega$\,. The quadratic form defined by $$\scalar{\d\Phi}{\d\Psi}$$ with domain $\H^1(\Omega)$ is $G$-invariant.
\end{proposition}

\begin{proof}
	First notice that the pull-back of a diffeomorphism commutes with the action of the exterior differential, cf., \cite[Theorem 7.4.4]{marsden01}. Then we have that $$\d(V(g^{-1})\Phi)=\d(g^*\Phi)=g^*\negthinspace\d\Phi\;.$$ Hence
	\begin{subequations}\label{dphiinvariant}
	\begin{align}
		\scalar{\d(V(g^{-1})\Phi)}{\d(V(g^{-1})\Psi)}&=\int_\Omega \eta^{-1}(g^*\negthinspace\d\Phi,g^*\negthinspace\d\Psi)\d\mu_\eta\\
		&=\int_\Omega g^*\negthickspace\left( \eta^{-1}(\d\Phi,\d\Psi) \right)g^*\negthinspace\d\mu_\eta\\
		&=\int_{g\Omega}\eta^{-1}(\d\Phi,\d\Psi)\d\mu_\eta\\
		&=\scalar{\d\Phi}{\d\Psi}\;,
	\end{align}
	\end{subequations}
where in the second inequality we have used that $g:\Omega\to\Omega$ is an isometry and therefore $$\eta^{-1}(g^*\negthinspace\d\Phi,g^*\negthinspace\d\Psi)=g^*\negthinspace \eta^{-1}(g^*\negthinspace\d\Phi,g^*\negthinspace\d\Psi)=g^*\negthickspace\left( \eta^{-1}(\d\Phi,\d\Psi) \right)\;.$$ The equations \eqref{dphiinvariant} guaranty also that $V(g)\H^1(\Omega)=\H^1(\Omega)$ since $V(g)$ is a unitary operator in $\H^0$ and by Remark \ref{equivalentsobolev} the norm $\sqrt{\norm{\operatorname{d}\cdot\;}^2+\norm{\cdot}^2}$ is equivalent to the Sobolev norm of order 1.
\end{proof}

Now we are able to prove the following theorem:

\begin{theorem}\label{repcommutation}
	Let $\mathrm{v}:G\to\mathcal{U}(\H^0(\pO))$ be the induced unitary representation of the group $G$ at the boundary. Then the quadratic form $Q_U$ of Definition \ref{DefQU} is $G$-invariant iff
	$$\comm{\mathrm{v}(g)}{U}\gamma(\Phi)=0\quad \forall\Phi\in\D_U\;.$$
\end{theorem}

\begin{proof}
	By Proposition \ref{prop:dphiinvariant} it is enough to show that $\scalarb{\gamma(\Phi)}{A_U\gamma(\Phi)}$ defined on $\D_U$ is $G$-invariant iff $$\comm{\mathrm{v}(g)}{U}\gamma(\Phi)=0\quad \forall\Phi\in\D_U\;.$$
	\begin{itemize}
		\item[$(\Leftarrow)$] Let $\Phi\in\D_U$\,. Notice that $\comm{\mathrm{v}(g)}{U}\gamma(\Phi)=0$ is equivalent to  
			\begin{equation*}
				\comm{\mathrm{v}(g)}{P^\bot}\gamma(\Phi)=0\quad\text{and}\quad
				\comm{\mathrm{v}(g)}{A_U}\gamma(\Phi)=0\;,
			\end{equation*}
		where $P^\bot$ is the orthogonal projector onto $\operatorname{Ran}E_{\{\pi\}}$\,, the closed subspace associated to the eigenvalue $-1$ of the unitary operator $U$\,, cf., Definition \ref{P,boundary} and Definition \ref{partialCayley}. From the former equality we have that $$P^\bot\gamma(V(g)\Phi)=P^\bot\mathrm{v}(g)\gamma(\Phi)=\mathrm{v}(g)P^\bot\gamma(\Phi)=0\;,$$ where we have used Eq.~\eqref{VPhi=vphi}. Hence we have that $V(g)\D_U\subset\D_U$\,. From the commutation relation for the partial Cayley transform $A_U$ we have that
		\begin{align*}
			\scalarb{\mathrm{v}(g)\gamma(\Phi)}{A_U\mathrm{v}(g)\gamma(\Phi)}&=\scalarb{\mathrm{v}(g)\gamma(\Phi)}{\mathrm{v}(g)A_U\gamma(\Phi)}\\
				&=\scalarb{\gamma(\Phi)}{A_U\gamma(\Phi)}\;.
			\end{align*}
		\item[$(\Rightarrow)$] Let $\Phi\in\D_U$\,, then $P^\bot\gamma(\Phi)=0$ and clearly $$\mathrm{v}(g)P^\bot\gamma(\Phi)=0\;.$$ Since the quadratic form is $G$-invariant we have that $V(g)\D_U\subset\D_U$ and therefore $P^\bot\mathrm{v}(g)\gamma(\Phi)=P^\bot\gamma(V(g)\Phi)=0\;.$ Hence $$\comm{\mathrm{v}(g)}{P^\bot}\gamma(\Phi)=0\;.$$ On the other hand we have that 
		\begin{equation*}
			\scalarb{\psi}{A_U\varphi}=\scalarb{\psi}{\mathrm{v}(g)^\dagger A_U\mathrm{v}(g)\varphi}\,,\quad \psi,\varphi\in P\H^{1/2}(\pO)\;.
		\end{equation*}
		Since the subspace $P\H^{1/2}(\pO)$ is a dense subspace of the invertibility boundary space $W=\operatorname{Ran}E^\bot_{\{\pi\}}$\,, see Definition \ref{P,boundary}, the equality above shows that $$\norm{\left(A_U-\mathrm{v}(g)^\dagger A_U\mathrm{v}(g)\right)\varphi}_W=0\,,\quad\forall\varphi\in P\H^{1/2}(\pO)\;,$$
		and therefore $$\comm{\mathrm{v}(g)}{A_U}\gamma(\Phi)=0\;.$$
	\end{itemize}
\end{proof}

Summarising we can say that given a compact group acting on a Riemannian manifold, if we are able to find an admissible unitary operator on the boundary verifying the commutation relation above, we can describe self-adjoint extensions of the Laplace-Beltrami operator that are $G$-invariant.

\section{Examples}
\label{sec:SymmetryExamples}

In this section we introduce two particular examples of $G$-invariant quadratic forms. In the first example we are considering a situation where the symmetry group is a discrete group. In the second one we consider $G$ to be a  genuine compact Lie group.

\begin{example}
	Let $\Omega$ be the cylinder $[-1,1]\times[-1,1]/\negthickspace\sim$\,, where $\sim$ is the equivalence relation $(-1,y)\sim(1,y)$\,. The boundary $\pO$ is the disjoint union of the two circles $\Gamma_1=\bigl\{[-1,1]\times\{-1\}/\negthickspace\sim\bigr\}$ and $\Gamma_2=\bigl\{[-1,1]\times\{1\}/\negthickspace\sim\bigr\}$\,. Let $\eta$ be the euclidean metric. Now consider that $G$ is the discrete, abelian group of two elements, $\{e,f\}$ and consider the following action on $\Omega$:
	\begin{align*}
		e:(x,y)&\to(x,y)\;,\\
		f:(x,y)&\to(-x,y)\;.
	\end{align*}
The induced action on the boundary is 
	\begin{align*}
		e:(\pm1,y)\to(\pm1,y)\;,\\
		f:(\pm1,y)\to(\mp1,y)\;.
	\end{align*}
Clearly $G$ transforms $\Omega$ onto itself and preserves the boundary. Moreover it is easy to check that $f^*\eta=\eta$\,.

Since the boundary $\pO$ consists of two disjoints manifolds $\Gamma_1$ and $\Gamma_2$\,, the Hilbert space of the boundary is $\H^0(\pO)=\H^0(\Gamma_1)\oplus\H^0(\Gamma_2)$\,. Let $\Phi\in\H^0(\pO)$\,. Then we can represent it using the associated block structure, namely: $$\Phi=\begin{pmatrix}\Phi_1(y)\\ \Phi_2(y)\end{pmatrix}$$ with $\Phi_i\in\H^0(\Gamma_i)$\,. The identity element will not introduce any restriction so let us concentrate on the element $f$\,. The induced unitary action on $\H^0(\pO)$ is therefore:
\begin{align*}
	\mathrm{v}(f)\begin{pmatrix}\Phi_1(y)\\ \Phi_2(y)\end{pmatrix}&=\mathrm{v}(f)\begin{pmatrix}\Phi(-1,y)\\ \Phi(1,y)\end{pmatrix}\\
	&=\begin{pmatrix}\Phi(1,y)\\ \Phi(-1,y)\end{pmatrix}=\begin{pmatrix}\Phi_2(y)\\ \Phi_1(y)\end{pmatrix}=\begin{pmatrix}0 & \mathbb{I}\\ \mathbb{I} & 0\end{pmatrix}\begin{pmatrix}\Phi_1(y)\\ \Phi_2(y)\end{pmatrix}\;.
\end{align*}
The set of unitary operators that describe the closable quadratic forms as defined in Chapter \ref{cha:QF} is given by unitary operators $$U=\begin{pmatrix} U_{11} & U_{12} \\ U_{21} & U_{22} \end{pmatrix}\;,$$ with $U_{ij}=\H^0(\Gamma_j)\to\H^0(\Gamma_i)$\,. According to Theorem \ref{repcommutation} only those unitary operators commuting with $\mathrm{v}(f)$ will lead to $G$-invariant quadratic forms. Imposing 
$$\left[\begin{pmatrix} 0 & \mathbb{I} \\ \mathbb{I} & 0\end{pmatrix},\begin{pmatrix} U_{11} & U_{12} \\ U_{21} & U_{22} \end{pmatrix}\right]=0\;,$$
we get the conditions
	\begin{align*}
		&U_{21}-U_{12}=0\;,\\
		&U_{22}-U_{11}=0\;.
	\end{align*}
	
We can consider now the particular case of generalised Robin boundary conditions shown in Example \ref{generalized Robin}. These are given by unitaries of the form 
	\begin{equation}
		U=\begin{bmatrix} e^{\mathrm{i}\beta_1}\mathbb{I}_1 & 0\\ 0 & e^{\mathrm{i}\beta_2}\mathbb{I}_2 \end{bmatrix}\;.
	\end{equation}
The $G$-invariance condition forces $\beta_1=\beta_2$\,, as one could expect. Notice that if $\beta_1\not=\beta_2$ we can obtain a self-adjoint extension of the Laplace-Beltrami operator, $\Delta_{\mathrm{min}}$\,, associated to the closable quadratic form defined by the unitary $U$\,. This extension will not be $G$-invariant even though $\Delta_{\mathrm{min}}$ is.

We can consider also conditions of the quasi-periodic type as in Example \ref{ex:quasiperiodic}. In this case 
	\begin{equation}
		U=\begin{bmatrix} 0 & e^{\mathrm{i}\alpha} \\ e^{-\mathrm{i}\alpha} & 0 \end{bmatrix}\,,\quad \alpha\in \C^0(S^1)\;.
	\end{equation} 
The $G$-invariance condition imposes that $e^{\mathrm{i}\alpha}=e^{-\mathrm{i}\alpha}$ and therefore among all the quasi-periodic conditions only the periodic ones, $\alpha\equiv0$\,, are allowed.
\end{example}

\begin{example}\label{ex:upperhemisphere}
	Let $\Omega$ be the unit, upper hemisphere. Its boundary $\pO$ is going to be the unit circle on the horizontal plane. Let $\eta$ be the induced Riemannian metric from the euclidean metric in $\mathbb{R}^3$\,. Consider that $G$ is the compact Lie group $O(2)$ of rotations around the $z$-axis. If we use polar coordinates on the horizontal plane, then the boundary is isomorphic to the interval $[0,2\pi]$ with the two endpoints identified, i.e., $\pO\simeq[0,2\pi]/\negthickspace\sim$\,, where $\sim$ is the equivalence relation $0\sim2\pi$\,. If $\theta$ is the coordinate describing the boundary, then the Hilbert space of the boundary is given by $\H^0([0,2\pi])$\,. 

Let $\varphi\in\H^{1/2}(\pO)$\,. The induced action of the group $G$ in this space is therefore given by 
$$\mathrm{v}(g^{-1}_\alpha)\varphi(\theta)=\varphi(\theta+\alpha)\;.$$
To analyse what are the possible unitary operators that lead to $G$-invariant quadratic forms it is convenient to use the Fourier series expansions of the elements in $\H^0(\pO)$\,. Let $\varphi\in\H^0(\pO)$\,, then 
$$\varphi(\theta)=\sum_{n\in\mathbb{Z}}\hat{\varphi}_ne^{\mathrm{i}n\theta}\;,$$
where the coefficients of the expansion are given by $$\hat{\varphi}_n=\frac{1}{2\pi}\int_0^{2\pi}\varphi(\theta)e^{-\mathrm{i}n\theta}\d\theta\;.$$
We can therefore consider the induced action of the group $G$ as a unitary operator on the Hilbert space ${l}^2$\,. In fact we have that :
\begin{align*}
\widehat{(\mathrm{v}(g^{-1}_\alpha)\varphi)}_n&=\frac{1}{2\pi}\int_0^{2\pi}\varphi(\theta+\alpha)e^{-\mathrm{i}n\theta}\d\theta\\
	&=\sum_{m\in\mathbb{Z}}\hat{\varphi}_me^{\mathrm{i}m\alpha}\int_0^{2\pi}\frac{e^{\mathrm{i}(m-n)\theta}}{2\pi}\d\theta=e^{\mathrm{i}n\alpha}\hat{\varphi}_n\;.
\end{align*}
	So that the induced action of the group $G$ is a unitary operator in $\mathcal{U}(l^2)$\,. More concretely we can represent it as $\widehat{\mathrm{v}(g^{-1}_\alpha)}_{nm}=e^{\mathrm{i}n\alpha}\delta_{nm}\,$\,. From all the possible unitary operators acting on the Hilbert space of the boundary, only those whose representation in $l^2$ commutes with the above operator will lead to $G$-invariant quadratic forms. Since $\widehat{\mathrm{v}(g^{-1}_\alpha)}$ acts diagonally on $l^2$ it is clear that only operators of the form $\hat{U}_{nm}=e^{\mathrm{i}\beta_n}\delta_{nm}$\,, $\beta_n\in\mathbb{R}$\,, will lead to $G$-invariant quadratic forms.
	
	As a particular case we can consider that all the parameters are equal, i.e., $\beta_j=\beta$\,. In this case it is clear that $(\widehat{U\varphi})_n=e^{\mathrm{i}\beta}\varphi_n$ which gives that $$U(\varphi)=e^{\mathrm{i}\beta}\varphi\;,$$ showing that the only generalised Robin boundary conditions compatible with the rotation symmetry are those that have a constant value along the boundary.
\end{example}


\section{Representation of generic unbounded quadratic forms}
\label{sec:Representation}

During this dissertation we have been using Kato's Representation Theorem, Theorem \ref{fundteo}, intensively. It has allowed us to obtain a variety of self-adjoint extensions of the Laplace-Beltrami operator. In fact, this is a powerful tool that allows for the characterisation of self-adjoint extensions of symmetric operators in a much general context. However, there is an important assumption without which the representation theorem does not follow. This assumption is the semi-boundedness assumption. As it is stated in the introduction it remains as one of the main open problems in the field to obtain a generalisation of this theorem to quadratic forms that are truly unbounded, i.e., not semi-bounded.

%

The converse is clearly true. Given a generic self-adjoint  operator one can  always construct an associated quadratic form by means of its spectral resolution of the identity, cf., Theorem \ref{spectraltheorem}. Let $T$ be a self-adjoint operator with domain $\D(T)$ and let $E_\lambda:\H\to\H$ be the corresponding spectral resolution of the identity. Then one can define the following domain, cf., \cite[Section VIII.6]{reed72}:  $$\D(Q_T):=\{\Phi\in\H|\int_{\mathbb{R}}|\lambda|\d(\scalar{\Phi}{E_\lambda\Phi})<\infty\}$$ and an associated quadratic form by $$Q_T(\Phi,\Psi)=\int_{\mathbb{R}}\lambda\d(\scalar{\Phi}{E_\lambda\Psi})\quad\forall\Phi,\Psi\in\D(Q_T)\;.$$

Notice that the domain of the quadratic form contains the domain $\D(T)$ of the operator $T$, that can be identified with $$\D(Q_T):=\{\Phi\in\H|\int_{\mathbb{R}}|\lambda|^2\d(\scalar{\Phi}{E_\lambda\Phi})<\infty\}\;.$$ Moreover, it is clear that this quadratic form might be represented in terms of the operator $T$\,, namely 
$$Q_T(\Phi,\Psi)=\scalar{\Phi}{T\Psi}\,,\quad\Phi,\Psi\in\D(T)\;.$$
This quadratic form is representable and is in general not lower nor upper semi-bounded. The question is now clear. What are the sufficient conditions for an Hermitean, not semi-bounded quadratic form to be representable? The aim of this section is to step forward towards an answer.\\

In order to state the results in the most general form we will need the notion of direct integral of Hilbert spaces, cf., \cite[Chapter II.1]{dixmier81}, \cite[Section 1.2.3]{lions72}.

\begin{definition}
Let $\mathcal{A}$ be a measure space with Radon measure $\mu$ and let $\H_\alpha$\,, $\alpha \in \mathcal{A}$\,, a family of complex separable Hilbert spaces.  We will say that $\H_\alpha$ is a \textbf{$\mu$-measurable field of Hilbert spaces} if there exists a family $\mathcal{M}$ of functions $\Phi \colon \alpha \mapsto \H_\alpha$ such that:

\begin{enumerate}
	\item For all $\Phi \in \mathcal{M}$ the function $\norm{\Phi}(\cdot):\alpha \mapsto || \Phi (\alpha) ||_{\H_\alpha}$ is $\mu$-measurable.
	\item  If $\Psi$ is a function $\alpha \mapsto \H_\alpha$ and for all $\Phi\in \mathcal{M}$, the scalar function $\alpha \mapsto \scalar{\Phi (\alpha)}{\Psi (\alpha)} _{\H_\alpha}$ is $\mu$-measurable, then $\Psi \in \mathcal{M}$.
	\item  There exists a sequence $\Phi_1, \Phi_2, \ldots$ of elements of $\mathcal{M}$ such that, for all $\alpha \in \mathcal{A}$, the sequence $\Phi_1(\alpha), \Phi_2(\alpha), \ldots$ generates $\H_\alpha$\,.
\end{enumerate}
The functions in the family $\mathcal{M}$ will be called \textbf{$\mu$-measurable functions} of the field.

\end{definition}

\begin{definition}\label{def:directintegral}
Let $\mathcal{A}$ be a measure space with Radon measure $\mu$ and let $\H_\alpha$ be a $\mu$-measurable field of Hilbert spaces with family of $\mu$-measurable functions $\mathcal{M}$.   The \textbf{direct integral of the field of Hilbert spaces $\H_\alpha$}, denoted by
$$ \H = \int^\oplus_{\mathcal{A}} \H_\alpha \d\mu (\alpha) \;,$$
is defined as the family of $\mu$-measurable functions of the field $\Phi \in \mathcal{M}$ such that
$$ || \Phi ||^2 = \int_{\mathcal{A}} || \Phi (\alpha) ||_{\H_\alpha}^2 \d\mu (\alpha) < \infty \;.$$
The functions $\Phi\in \H$ can be written as
$$ \Phi = \int_{\mathcal{A}} \Phi (\alpha) d\mu (\alpha) \;,$$
where the integral has to be understood in the sense of Bochner. The scalar product on $\H$ is given by
$$ \scalar{\Phi}{\Psi} = \int_{\mathcal{A}} \scalar{ \Phi (\alpha)} {\Psi (\alpha)}_{\H_\alpha} \d\mu (\alpha) \;.$$ 
\end{definition}

\begin{remark}
	Notice that if $\{W_i\}_{i\in\mathbb{Z}}$ is a discrete family of orthogonal Hilbert subspaces of $\H$ such that $$\H=\bigoplus_{i\in\mathbb{Z}}W_i\;,$$ then automatically $\{W_i\}$ is a $\mu$-measurable field of Hilbert spaces over the discrete space $\mathbb{Z}$ with $\mu$ the counting measure, i.e., $\mu(\{k\})=1$, $k\in\mathbb{Z}$. Then $$\int^{\oplus}_{\mathbb{Z}}W_i\d\mu(i)=\bigoplus_{i\in\mathbb{Z}}W_i=\H\;.$$
	Moreover, the set of $\mu$-measurable functions $\mathcal{M}$ is just the family of vectors $\{\Phi_i\in W_i \mid i\in\mathbb{Z}\}$ and the elements of the direct integral of the field of Hilbert spaces $\{W_i\}$ is given by those functions that verify $$\sum_{i\in\mathbb{Z}}\norm{\Phi_i}^2<\infty\;.$$
\end{remark}


\begin{definition}\label{def: poa}
	Let $Q$ be a Hermitean quadratic form and let $W_\alpha$ be a $\mu$-measurable field of orthogonal subspaces of a Hilbert space $\H$ such that $\H=\int^{\oplus}W_\alpha\d\mu(\alpha)$ and let $P_\alpha:\H\to W_\alpha$ be the corresponding orthogonal projectors. $Q$ is \textbf{partially orthogonally additive} with respect to the family $\{W_\alpha\}$ if $$Q(\Phi+\Psi)=Q(\Phi)+Q(\Psi)\;,$$ whenever $\Phi$ and $\Psi$ are such that $\int_{I}\scalar{P_\alpha\Phi}{P_\alpha\Psi}\d\mu(\alpha)=0$ for all $I\subset \mathcal{A}\;.$
	The \textbf{sectors} $Q_\alpha$  are defined to be the restrictions to the subspaces $W_\alpha$\,, 
	$$Q_\alpha(\cdot):=Q(P_\alpha\cdot)\;.$$
\end{definition}

Notice that the condition 
$\scalar{P_\alpha\Phi}{P_\alpha\Psi}=0$
implies that $$\scalar{\Phi}{\Psi}=\int_{\mathcal{A}}\scalar{P_\alpha\Phi}{P_\alpha\Psi}\d\mu(\alpha)=0\;.$$


\begin{proposition}\label{prop: qdecomposition}
	Let $Q$ be a partially orthogonally additive quadratic form with respect to the family $\{W_\alpha\}$\,, $\alpha\in\mathcal{A}$\,. Then $Q$ can be written as
	$$Q(\Phi)=\int_{\mathcal{A}}Q_\alpha(\Phi)\d\mu(\alpha)\;.$$
\end{proposition}

\begin{proof}
	From Definition \ref{def:directintegral} and Definition \ref{def: poa} we have that 
	$$Q(\Phi)=Q\left(\int_{\mathcal{A}}P_\alpha\Phi\d\mu(\alpha)\right)=\int_{\mathcal{A}}Q(P_\alpha\Phi)\d\mu(\alpha)=\int_{\mathcal{A}}Q_\alpha(\Phi)\d\mu(\alpha)\;.$$
\end{proof}

Now we are ready to state the main theorem of this section.

\begin{theorem}\label{representationpoa}
	Let $Q$ be a partially orthogonally additive quadratic form with respect to the family $\{W_\alpha\}$\,, $\alpha\in\mathcal{A}$\,, defined on the domain $\D$\,. Let each sector $Q_\alpha$ be semi-bounded (upper or lower) and closable. Then
	\begin{enumerate}
		\item The quadratic form can be extended continuously to the closure of the domain with respect to the norm
			$$\normm{\Phi}^2_Q=\int_{\mathcal{A}}\normm{P_\alpha\Phi}^2_{Q_\alpha}\d\mu(\alpha)\;,$$
			i.e., to the domain $$\D(\overline{Q}):=\overline{\D}^{\normm{\cdot}_Q}\;.$$
			The extension will be denoted by $\overline{Q}$.
		\item The quadratic form $\overline{Q}$ defined on $\D(\overline{Q})$ is representable, i.e., it exists a self-adjoint operator $T$ with domain $\D(T)$ such that 
		$$\overline{Q}(\Phi)=\scalar{\Phi}{T\Phi}\quad\forall\Phi\in\D(T)\;.$$
	\end{enumerate}
\end{theorem}

\begin{proof}
	\begin{enumerate}
		\item 	\begin{align*}
					|Q(\Phi,\Psi)|&\leq \int_{\mathcal{A}}|Q(P_\alpha\Phi,P_\alpha\Psi)|\d\mu(\alpha)\\
						&\leq \int_{\mathcal{A}}\normm{P_\alpha\Phi}_{Q_\alpha}\normm{P_\alpha\Psi}_{Q_\alpha}\d\mu(\alpha)\\
						&\leq \int_{\mathcal{A}}\normm{P_\alpha\Phi}_{Q_\alpha}^2\d\mu(\alpha)\int_{\mathcal{A}}\normm{P_\alpha\Psi}_{Q_\alpha}^2\d\mu(\alpha)\;,
				\end{align*}
			where each $\normm{\cdot}_{Q_\alpha}$ is defined using the corresponding lower or upper bound, i.e., if $Q_\alpha$ is lower semi-bounded with lower bound $a$ then 
			$$\normm{\cdot}^2_{Q_\alpha}=(1+a)\norm{P_\alpha\cdot}^2+Q(\cdot)\;,$$
		 	and if $Q_\alpha$ is upper semi-bounded with upper bound $b$ then 
			$$\normm{\cdot}^2_{Q_\alpha}=(1+b)\norm{P_\alpha\cdot}^2-Q(\cdot)\;.$$
		\item By construction the sector $\overline{Q}_\alpha$\,, $\alpha\in\mathcal{A}$\,, defines a closed semi-bounded quadratic form since $Q_\alpha$ is assumed to be closable. Then for each $\alpha\in\mathcal{A}$ it exists a self-adjoint operator $T_\alpha:\D(T_\alpha)\to W_\alpha$ such that $$\overline{Q}_\alpha(\Phi,\Psi)=\scalar{P_\alpha\Phi}{T_\alpha\Psi}\quad\Phi,\Psi\in\D(T_\alpha)\;.$$
			Notice that each $T_\alpha$ commutes with the corresponding $P_\alpha$. Let $T:\D(T)\to\H$ be the operator 
			$$T=\int_{\mathcal{A}}T_\alpha\d\mu(\alpha)$$ 
			with domain $$\D(T)=\{\Phi\in\D(\overline{Q})\mid \int_{\mathcal{A}}\normm{P_\alpha\Phi}^2_{T_\alpha}\d\mu(\alpha)<\infty\}\;,$$
			where $$\normm{\cdot}^2_{T_\alpha}=\norm{P_\alpha\cdot}^2+\norm{T_\alpha\cdot}^2\;.$$
			Notice that $\Phi\in\D(T)$ implies $P_\alpha\Phi\in\D(T_\alpha)$\,. Denote the norm in the above domain by 
			$$\normm{\cdot}^2_{T}=\int_{\mathcal{A}}\normm{P_\alpha\cdot}^2_{T_\alpha}\d\mu(\alpha)\;.$$
			From Definition \ref{def:directintegral} it is clear that $(\D(T),\normm{\cdot}_T)$ is a Hilbert space and therefore $T$ is a closed operator. It is easy to show that it is symmetric. Let us show that this operator is indeed self-adjoint.
			
			Let $\Phi_\alpha\in\D(T_\alpha)$ and let $\Psi\in\mathcal{N_{+}}$. Then
			\begin{align*}
				0=\scalar{\Psi}{(T+\mathrm{i})}&=\int_{\mathcal{A}}\scalar{P_\alpha\Psi}{(T_\alpha+\mathrm{i})\Phi_\alpha}\d\mu(\alpha)\\
					&=\scalar{P_\alpha\Psi}{(T_\alpha+\mathrm{i})\Phi_\alpha}
			\end{align*}
			 This implies that that $P_\alpha\Psi=0$. Since the choice of $P_\alpha$ is arbitrary we have that $\Psi=0$ and therefore $n_+=0$. Equivalently $n_-=0$. Hence $T$ is self-adjoint, cf., Theorem \ref{thmvonNeumann}.
	\end{enumerate}
\end{proof}

%

For a general unbounded quadratic form it is enough that there exist at least two subspaces, $W_+$ and $W_-$\,, that verify the above conditions. The former needs to be lower semi-bounded and the latter upper semi-bounded. The quadratic form will be the sum of the corresponding sectors. To show that the above characterisation is meaningful, let us consider three examples in which we obtain representations for 3 well known operators.

\begin{example}[The multiplication operator]\label{ex: position}
	Consider the Hilbert space $\H=\H^0(\mathbb{R})$\,. Let $Q$ be the Hermitean quadratic with domain $\C^{\infty}_c(\mathbb{R})$ defined by 
	$$Q(\Phi)=\int_{\mathbb{R}}\bar{\Phi}(x)x\Phi(x)\d x\;.$$
Consider the subspaces 
	$$W_{+}=\{\Phi\in\H|\operatorname{supp}\Phi\subset\mathbb{R}^+_{\mathrlap{0}}\}\;,$$
	$$W_{-}=\{\Phi\in\H|\operatorname{supp}\Phi\subset\mathbb{R}^-\}\;.$$ 
The projections onto these subspaces are given by multiplication with the corresponding characteristic functions. The restrictions of the quadratic form above are therefore 
	$$Q_+(\Phi)=\int_{\mathbb{R}^+_0}\bar{\Phi}(x)x\Phi(x)\d x\;,$$
	$$Q_-(\Phi)=\int_{\mathbb{R}^-}\bar{\Phi}(x)x\Phi(x)\d x\;.$$
The sector $Q_+$ is clearly a lower semi-bounded quadratic form. The proof that it is in fact closable is a standard result in the theory of quadratic forms, since it is the quadratic form associated to a lower semi-bounded, symmetric operator, see comments after Remark \ref{rem:represntability}. The domain of the corresponding closed quadratic form and associated self-adjoint operators are
$$\D(\overline{Q}_+)=\{\Phi\in\H|\int_{\mathbb{R}^+_0}x|\Phi|^2\d x<\infty\}\;,$$
$$\D(T_+)=\{\Phi\in\H|\int_{\mathbb{R}^+_0}x^2|\Phi|^2\d x<\infty\}\;.$$

The sector $Q_-$ is an upper semi-bounded quadratic form. Hence consider $-Q_-$\,. It is closable by the same reasons than $Q_+$\,. In this case the corresponding domains are
$$\D(\overline{Q}_-)=\{\Phi\in\H|\int_{\mathbb{R}^-}-x|\Phi|^2\d x<\infty\}\;,$$
$$\D(T_-)=\{\Phi\in\H|\int_{\mathbb{R}^-}x^2|\Phi|^2\d x<\infty\}\;.$$

The quadratic form $Q$ can therefore be extended to the domain $$\D(\overline{Q})=\{\Phi\in\H|\int_\mathbb{R}|x|\,|\Phi|^2\d x<\infty\}\;,$$ where $\overline{Q}(\Phi)=\overline{Q}_+(\Phi)+\overline{Q}_-(\Phi)$ and can be represented by the operator $$T=T_+P_++T_-P_-$$ with domain
$$\D(T)=\{\Phi\in\H|\int_\mathbb{R}x^2\,|\Phi|^2\d x<\infty\}\;.$$
\end{example}

\begin{example}[The momentum operator]\label{ex: momentum}
	Consider the Hilbert space $\H=\H^0(\mathbb{R})$\,. Let $Q$ be the Hermitean quadratic form with domain $\C_c^\infty(\mathbb{R})$ defined by 
\begin{equation}\label{qfmomentum}
Q(\Phi)=\int_{\mathbb{R}}\bar{\Phi}(x)\Bigl(-i\frac{\d}{\d x}\Phi(x)\Bigr)\d x\;.
\end{equation}
We are going to obtain also in this case two subspaces $W_{\pm}$ with respect to which the quadratic form is partially orthogonally additive. To do so we need to make use of the Fourier transform. Let $\Phi\in\C_c^\infty(\mathbb{R})$\,. The Fourier transform is defined by $$\mathcal{F}\Phi(k)=\frac{1}{\sqrt{2\pi}}\int_\mathbb{R}\Phi(x)e^{-ikx}\d x\;.$$ It is well known that this operation can be extended continuously to a unitary operation $\mathcal{F}:\H^0(\mathbb{R})\to\H^0(\mathbb{R})$\,. Moreover, the momentum operator defined above acts diagonally on the transformed space, i.e., $$\mathcal{F}\Bigl(-i\frac{\d}{\d x}\Phi(x)\Bigr)(k)=k\mathcal{F}\Phi(k)\;.$$
Under this transformation this example becomes the example of the multiplication operator above and hence we can define the subspaces $W_{\pm}$ accordingly.
Consider the subspaces 
	$$W_{+}=\{\Phi\in\H|\operatorname{supp}\mathcal{F}\Phi\subset\mathbb{R}^+_{\mathrlap{0}}\}\;,$$
	$$W_{-}=\{\Phi\in\H|\operatorname{supp}\mathcal{F}\Phi\subset\mathbb{R}^-\}\;.$$	
Since the Fourier transform is a unitary operation, the quadratic form of Eq.~\eqref{qfmomentum} is clearly partially orthogonally additive with respect to this decomposition. For the same reason the sectors are semi-bounded. Again the restriction $Q_-$ is upper semi-bounded and we need to consider $-Q_-$\,. In this case the domains of the associated quadratic forms and operators can also be described precisely and have the form of the corresponding domains on Example \ref{ex: position} with $\Phi$ replaced by $\mathcal{F}\Phi$\,. Notice that in this case the quadratic form can be extended to the domain 
$$\D(\overline{Q})=\{\Phi\in\H| \int_{\mathbb{R}}|k||\mathcal{F}\Phi|^2\d k<\infty \}$$ and the corresponding self-adjoint operator $T$ is defined on $$\D(T)=\{\Phi\in\H| \int_{\mathbb{R}}k^2|\mathcal{F}\Phi|^2\d k<\infty \}\;.$$ These last two spaces can be proved to be equivalent to $\H^{1/2}(\mathbb{R})$ and $\H^{1}(\mathbb{R})$ respectively.
\end{example}

In the last example we consider a quadratic form $Q_U$ given as in Definition \ref{DefQU}. However we will not ask the unitary operator to have gap at -1. Nevertheless the Cayley transform, cf., Definition \ref{partialCayley} is still going to be a well defined self-adjoint operator. It will not be bounded though.

\begin{example}\label{ex:nogap}
	Let $\Omega$ be the unit disk $$\Omega=\{(x,y)\in\mathbb{R}^2\mid x^2+y^2\leq1\}$$ and let $\eta$ be the euclidean metric. The boundary is therefore $\pO=S^1$. The rotation group $O(2)$ acts naturally on this manifold by rotations around the origin. In polar coordinates we have that 
	$$g_\alpha:(r,\theta)\to(r,\theta+\alpha)\;.$$
	Let us consider an $O(2)$-invariant quadratic form $Q_U$. Any unitary $U\in\mathcal{U}(\H^0(\pO))$ that verifies 
	$$\comm{\mathrm{v}(g)}{U}\varphi=0\quad\forall\varphi\in\H(\pO)$$ will lead to a $O(2)$-invariant quadratic form. From the analysis in Example \ref{ex:upperhemisphere} we know that only those unitaries whose representation in the Fourier transformed space is diagonal will satisfy the commutation relation. Then consider 
	\begin{equation}
	\hat{U}_{nm}=e^{\mathrm{i}\beta_n}\delta_{nm}\,, \quad\text{with}\quad
		\beta_n=\begin{cases}
			\pi(1-\frac{1}{n})\,, &n>0\\
			0\,, & n=0\\
			-\pi(1+\frac{1}{n})\,, & n<0
		\end{cases}\;.
	\end{equation}

Such an unitary operator does not verify the gap condition, cf., Definition \ref{DefGap}. Hence it will not lead to a a semi-bounded quadratic form. In the transformed space the Cayley transform reads 
	$$(\hat{A_U})_{nm}=-\tan\frac{\beta_n}{2}\delta_{nm}\;.$$
We can express the quadratic form $Q_U$ in terms of the subspaces associated with the irreducible representations of the group $O(2)$. This subspaces coincide with the expansion in the Fourier basis, $\{e^{\mathrm{i}n\theta}\}$. Namely, $$\Phi(r,\theta)=\sum_{n\in\mathbb{Z}}\hat{\Phi}_n(r)e^{\mathrm{i}n\theta}\;.$$ So we have that 
\begin{align*}
	Q_U(\Phi)&=\int_\Omega\eta^{-1}(\d\Phi,\d\Phi)\sqrt{|\eta|}\d r\d\theta-\scalarb{\gamma(\Phi)}{A_U\gamma(\Phi)}\\
		&=\int_0^1\int_0^{2\pi}\left( \Bigl| \frac{\partial \Phi}{\partial r}\Bigr|^2 +\frac{1}{r^2}\Bigl| \frac{\partial \Phi}{\partial \theta}\Bigr|^2 \right)r\d r\d\theta-\\
		& \phantom{aaaaaaaaaaaaaaa}-\scalarb{\gamma(\hat{\Phi}_n(r))e^{\mathrm{i}n\theta}}{A_U\left((\hat{\Phi}_n(r))e^{\mathrm{i}n\theta}\right)}\\
		&=\sum_n\int_0^1\left(\Bigl| \frac{\partial \hat{\Phi}_n}{\partial r}\Bigr|^2 +\frac{n^2}{r^2}| \hat{\Phi}_n |^2\right)r\d r + \sum_n \tan\frac{\beta_n}{2}\norm{\gamma(\hat{\Phi}_n)}^2_{\H^0(\pO)}\\
		&=\sum_n Q_n(\hat{\Phi_n})\;,
\end{align*}
where we have defined $$Q_n(\hat{\Phi}_n):=\int_0^1\left(\Bigl| \frac{\partial \hat{\Phi}_n}{\partial r}\Bigr|^2 +\frac{n^2}{r^2}| \hat{\Phi}_n |^2\right)r\d r + \tan\frac{\beta_n}{2}\norm{\gamma(\hat{\Phi}_n)}^2_{\H^0(\pO)}\;.$$
This shows that $Q_U$ is partially orthogonally additive with respect to the family of subspaces $\{e^{\mathrm{i}n\theta}\}$\,, $n\in\mathbb{Z}$\,. Notice that the boundary term of $Q_U$ not lower nor upper semi-bounded. However, for each sector $Q_n$ we can use similar arguments to those in the proof of Theorem \ref{maintheorem1}, more concretely in Eqs.~\eqref{equationssemibounded}, to show that it is semi-bounded below. That the sectors $Q_n$ are closable can be proved using similar arguments to those in Theorem \ref{maintheorem2}. Hence we are in the conditions of Theorem \ref{representationpoa} and it exist a self-adjoint operator that represents the quadratic form $Q_U$. Theorem \ref{DeltaUextDeltamin} ensures that it is a self-adjoint extension of the minimal closed extension $-\Delta_{\mathrm{min}}$\,.
\end{example}

\begin{corollary}
	Let $\{W_i\}_{i\in\mathbb{Z}}$ be a discrete family of finite dimensional orthogonal subspaces of the Hilbert space $\H$ and let $Q$ be a partially orthogonally additive Hermitean quadratic form with respect to the family $\{W_i\}_{i\in\mathbb{Z}}$. Then, $Q$ is representable.
\end{corollary}

\begin{proof}
	This is a direct application of Theorem \ref{representationpoa} since each sector $Q_i$ is bounded and therefore semi-bounded.
\end{proof}

\begin{corollary}
	Let $G$ be a compact Lie group with unitary representation $V:G\to\mathcal{U}(\H)$, such that the multiplicity of each irreducible representation of $G$ in $V$ is finite. Let $Q$ be a $G$-invariant Hermitean quadratic form on $\H$ with domain $\D(Q)$. Then $Q$ is representable.
\end{corollary}

\begin{proof}
	The Hilbert space $\H$ can be decomposed in terms of the irreducible representations of the group $G$ as $$\H=\bigoplus_{i\in\hat{G}}\H_{i}\;.$$ By assumption, each $\H_i$ is the direct sum of a finite number of copies of the finite dimensional spaces carrying the irreducible representations of the group $G$. By the previous corollary it is enough to show that the quadratic form will be partially orthogonally additive with respect to this family. The following identity holds for $\Phi,\Psi\in\D(Q)$.
	$$Q(\Phi+\Psi)=Q(\Phi)+Q(\Psi)+2\mathfrak{Re} Q(\Phi,\Psi)\;.$$ Let us show that $Q(\Phi,\Psi)=0$ if $\Phi\in\H_i$\,,$\Psi\in\H_j$\,,$i\neq j$\,. Since each $\H_i$ is finite dimensional we can consider the following mapping 
	$$
		\begin{array}{r l}
		\tilde{Q}^*_{ij}:\H_i &\to\H^*_j \\
		\Phi_i  &\to  \tilde{Q}^*_{ij}(\Phi_i)(\cdot):=Q(\Phi_i,\cdot)\;,
		\end{array}$$
where $\H_j^*$ stands for the dual space. We can identify $\H_j^*$ with itself and we shall consider then the corresponding map $$\tilde{Q}_{ij}:\H_i \to\H_j\;.$$ Now we have that
\begin{align*}
	\tilde{Q}^*_{ij}(V(g)\Phi_i)(\Psi_j)&=Q(V(g)\Phi_i,\Psi_j)\\
		&=Q(\Phi_i,V(g)^\dagger\Psi_j)\\
		&=\tilde{Q}^*_{ij}(\Phi_i)(V(g)^\dagger\Psi_j)\;,
\end{align*}
and therefore $$\tilde{Q}_{ij}(V(g)\Phi_i)=V(g)\tilde{Q}_{ij}(\Phi_i)$$ for all $g\in G$. By Schur's Lemma we have that $\tilde{Q}_{ij}=0$.
\end{proof}

%
%

\clearemptydoublepage
\chapter*{Conclusions and Further Work}
\addcontentsline{toc}{chapter}{Conclusions and Further Work}
\markboth{}{Conclusions and Further Work}
\stepcounter{chapter}

In this dissertation we have provided a characterisation of a wide class of self-adjoint extensions of the Laplace-Beltrami operator. The key object to characterise the boundary conditions is the boundary equation 
\begin{equation}\label{asoreycha6}
	\varphi-\mathrm{i}\dot{\varphi}=U(\varphi+\mathrm{i}\dot{\varphi})\;,
\end{equation}
see Proposition \ref{prop: asorey}. The boundary conditions described by such equation are in one-to-one correspondence with self-adjoint extensions of the Laplace-Beltrami operator only if the manifold $\Omega$ is one-dimensional. However, the fact that it is strongly related with the structure of the boundary manifold makes it very suitable to study the different self-adjoint extensions of the aforementioned operator. The results shown in Chapter 4 and Chapter 5 account for this. 

As stated in the introduction, our approach complements the existing characterisations of the extensions of elliptic differential operators. On one hand there is the theory developed by G.~Grubb in the 1960's , cf., \cite{grubb68}, where the structure of the boundary is kept but the the role of the unitary operator $U\in\mathcal{U}(\H^0(\pO))$ is substituted by a family of pseudo-differential operators acting on the Sobolev spaces of the boundary. On the other hand there is the theory of boundary triples, cf., \cite{bruning08}. In this case the boundary form is substituted by a form in an abstract space. This form is such that their maximally isotropic subspaces, i.e., the subspaces where the form vanish identically, are in one-to-one correspondence with the set of unitary operators acting on this space. Thus, in this approach the set of self-adjoint extensions is in one-to-one correspondence with the unitary group, but it fails to capture the structure of the boundary.

In Chapter \ref{cha:QF} we have shown under what assumptions do the boundary conditions described by Eq.~\eqref{asoreycha6} lead to self-adjoint extensions of the Laplace-Beltrami operator. More concretely, a condition on the spectrum of the unitary operator characterising the self-adjoint extension, the gap condition (see Definition \ref{DefGap}), is enough to ensure semi-boundedness. Moreover, a variety of meaningful examples is provided. In particular, the set of generalised periodic boundary conditions described in Example \ref{periodic} is well suited for the analysis of topology change in quantum theory, see \cite{asorey05,balachandran95,shapere12}\,. As showed in Example \ref{generalized Robin} and Example \ref{generalized Robin2} the aforementioned self-adjoint extensions include boundary conditions of Robin type. So, as a particular case, we have shown that Robin boundary conditions of the form $$\dot{\varphi}=g\cdot\varphi\quad g\in\C^0(\pO)\;,$$ lead to semi-bounded self-adjoint extensions of the Laplace-Beltrami operator. 

This kind of Robin boundary conditions appear in the study of certain quantum systems with boundary such as the ones related to topological insulators and to the quantum Hall effect, cf., \cite{hasan10,morandi88} . This boundary conditions appear naturally at the interphase between two materials when one of them is in the superconducting phase. The fact that the corresponding operators describing the dynamics can be proved to be lower semi-bounded is of main importance for the consistency of the physical theory describing such systems. Moreover, the states with the lowest energies are known to be strongly localised at the boundary. The lowest eigenvalue showed at the right of Figure \ref{laplaciano4} is an example of this localisation at the boundary. In recent research together with M.~Asorey and A.P.~Balachandran, \cite{asorey13}, it is shown that the appearance of such low lying edge states is a quite general feature that happens not only for the Laplace operator but also for the Dirac operator or for the generalisation of the Laplace operator to tensor fields, the Laplace-de Rham operator, that includes the electromagnetic field. In fact, the results of Chapter \ref{cha:QF} are easily extendable to the more general setting of Hermitean fibre bundles, where the Laplace-Beltrami operator is generalised by the so called Bochner Laplacian (see, e.g., \cite{booss93,lawson89,poor81}). The biggest obstruction to obtain such generalisation is to find a proper substitute for the radial operator $R$, cf., Definition \ref{Defunidimensional} and Proposition \ref{prop: intervalrobin}. Nevertheless, this can be circumvented too. We leave this generalisation for further work.
In Chapter \ref{cha:FEM} and Chapter \ref{cha:Symmetry} the particular structure of the boundary conditions in Eq.~\eqref{asoreycha6} is used to proof a number of other results.

In Chapter \ref{cha:FEM} a numerical scheme to approximate the spectral problem is introduced. In particular it is shown that the convergence of the scheme is granted provided that the gap condition for the unitary operator holds. As a proof of the factibility of this numerical scheme a one-dimensional version is constructed explicitly. The numerical experiments show, among other things, that the convergence rates provided are satisfied. It is worth to mention that the convergence of the numerical scheme is proved for any dimension. Thus, this family of numerical algorithms can be used to approximate the solutions of problems in dimension two and higher, which in particular can be applied to the analysis of topology change and to the study of the edge states mentioned before. The task of programming a two-dimensional version of the numerical scheme is currently being carried out together with A.~L\'{o}pez Yela.

The Chapter \ref{cha:Symmetry} is devoted to the analysis of the role of symmetries in the construction of the different self-adjoint extensions. During this dissertation we have used the notion of $G$-invariance to avoid confusion with the notion of symmetric operator, cf., Definition \ref{def:symmetric}. We showed using the most abstract setting, i.e., using von Neumann's Theorem, that the set of $G$-invariant self-adjoint extensions of a symmetric operator is in one-to-one correspondence with the isometries $K:\mathcal{N}_+\to\mathcal{N}_-$ that lie in the commutant of the unitary representation of the group. This shall point out that even if a symmetric operator is $G$-invariant it may happen that their self-adjoint extensions are not $G$-invariant. Consider an operator that is constructed as the tensor product of two operators. It is a common error to assume that all the self-adjoint extensions of the former can be obtained by the tensor product of the isometries defining the self-adjoint extensions of the factors in the tensor product. The reason behind this is more easily understood using the characterisation by boundary conditions developed in this dissertation. Consider a symmetric operator defined on a manifold 
$$\Omega=\Omega_1\times\Omega_2\;.$$ Then, the self-adjoint extensions will be parameterised using the Hilbert space of square integrable functions defined on the boundary of the manifold 
$$\pO=\pO_1\times\Omega_2 \sqcup \Omega_1\times\pO_2\;.$$
The common error consists in just considering boundary conditions defined over $$\pO_1\times\pO_2\;.$$ The $G$-invariance results proved in Chapter 5, mainly Theorem \ref{Ginvariantoperator}, Theorem \ref{QAinvariant}, and Theorem \ref{InvariantRepresentation} provide tools that help to deal with this considerations. In particular we have focused again in the particular family of quadratic forms associated to the Laplace-Beltrami operator introduced in Chapter \ref{cha:QF}. We show that the set of self-adjoint extensions compatible with the symmetries of the underlying Riemannian manifold is also related with the commutant of the unitary representation of the group.

The analysis of self-adjoint extensions compatible with symmetries can be carried a step further. We have already discussed that in quantum mechanics situations where one operator is invariant under the action of a symmetry are common and very important. However, sometimes it is the case that the dynamics of a given physical situation is given over the quotient of a manifold by the action of a group. A typical example of this is the treatment of indistinguishable particles, where the configuration space is the quotient of the space where the particles move by the permutation group (see \cite{leinaas77} and references therein). The description of such situations can become very involved. In particular, if the action of the group has fixed points the quotient is no longer a differentiable manifold. Moreover, the fixed points become elements of the boundary of the quotient space. The description of possible self-adjoint extensions over those quotients is therefore at jeopardy. The study of the relations of the self-adjoint extensions compatible with the symmetry may help to analyse the structure of the dynamics at the quotient space.

The final section of Chapter \ref{cha:Symmetry} is devoted to the generalisation of Kato's Representation Theorem to quadratic forms that are not semi-bounded. We have introduced the notion of sector of the quadratic form, cf., Definition \ref{def: poa}. This notion shall play the analogue role to the concept of invariant subspaces that appears when analysing the structure of self-adjoint operators. We have shown in Theorem \ref{representationpoa} that quadratic forms whose sectors are both, semi-bounded and closable, are representable in terms of a self-adjoint operator. Thus we have paved the way towards a generalisation of Kato's Representation Theorem for truly unbounded, i.e., not semi-bounded, quadratic forms. This generalisation remains to be one of the main open problems in the field. The last step in order to obtain the proper generalisation would be to identify under what circumstances do Hermitean quadratic forms possess semi-bounded and closable sectors or to show that they do always possess them. In particular we have shown by means of Example \ref{ex:nogap} that the gap condition, cf., Definition \ref{DefGap}, is not mandatory and that the class of self-adjoint extensions of the Laplace-Beltrami operator that can be obtained by means of the quadratic forms introduced in Chapter \ref{cha:QF} is even wider.

\clearemptydoublepage
\chapter*{List of Publications}
\addcontentsline{toc}{chapter}{List of Publications}

We finish this dissertation with a list of the publications that account for the results appearing in it, together with others that are not directly related.\\

\begin{itemize}
	\item {\sc J.M.~P\'erez-Pardo.} Quadratic Forms, Unbounded Self-Adjoint Operators and Quantum Observables. {\em  Il Nuovo Cimento C}, {\bf 36}(3) (2013), 205--214.
	
	Chapter \ref{cha:Preliminaries} and Chapter \ref{cha:QF}

	\item {\sc A.~Ibort, F.~Lled\'o and J.M.~P\'erez-Pardo.} Self-Adjoint Extensions of the Laplace-Beltrami Operator and unitaries at the boundary. {\em ArXiv:1308.0527} (2013).
	
	Chapter \ref{cha:QF}

	\item {\sc A.~Ibort and J.M.~P\'erez-Pardo.}  Numerical Solutions of the Spectral Problem for Arbitrary Self-Adjoint Extensions of the One-Dimensional Schr\"odinger Equation. {\em SIAM J. Numer. Anal.}, \textbf{51}(2) (2013), 1254--1279.
	
	Chapter \ref{cha:FEM}

	\item {\sc A.\ Ibort and J.M.\ P\'erez-Pardo.} Quantum Control and Representation Theory. {\em J. Phys. A: Math. Theor}. \textbf{42} (2009), 205301.

	\item {\sc M.~Asorey, A.P.~Balachandran and J.M.~P\'erez-Pardo.} Edge States: Topological Insulators, Superconductors and QCD Chiral Bags. {\em ArXiv:1308.5635} (2013).

\end{itemize}

\clearemptydoublepage

\chapter*{List of Symbols} 
\markboth{List of Symbols}{}
\addcontentsline{toc}{chapter}{List of Symbols}
\begin{tabbing}
$A_U$~~~~~~~~~~~\=\parbox{4.3in}{Partial Cayley transform of the unitary operator $U$\dotfill \pageref{symb:AU}}\\
\addsymbol \C^\infty (\Omega): {Space of smooth functions defined over the manifold $\Omega$}{symb:smooth}
\addsymbol \C_c^\infty (\Omega): {Space of smooth functions with compact support in the interior of $\Omega$}{symb:smooth compact}
\addsymbol \mathcal{C}: {Cayley transformation}{symb:Cayley}
\addsymbol \D(T): {Dense domain of the operator $T$}{symb:domainT}
\addsymbol \D_U: {Domain of the quadratic form associated to the Laplace-Beltrami operator described by the unitary at the boundary $U$}{symb:DU}
\addsymbol \d: {Exterior differential}{symb:d}
\addsymbol \d\mu_\eta: {Riemannian volume form}{symb:volume}
\addsymbol \Delta_\eta: {Laplace-Beltrami operator of the Riemannian manifold $(\Omega,\eta)$}{symb:laplacian}
\addsymbol \Delta_0: {Laplace-Beltrami operator defined over the compactly supported smooth functions}{symb:Delta0}
\addsymbol \Delta_{\mathrm{min}}: {Minimal extension of $\Delta_0$}{symb:Deltamin}
\addsymbol \Delta_{\mathrm{max}}: {Maximal extension of $\Delta_0$\,. It coincides with the adjoint operator of $\Delta_0$\quad}{symb:Deltamax}
\addsymbol \Delta_U: {Self-adjoint extension of $\Delta_{\mathrm{min}}$ associated to the unitary operator $U$}{symb:DeltaU}
\addsymbol E_\lambda: {Spectral resolution of the identity}{symb:spectralres}
\addsymbol \gamma(\cdot): {Trace map}{symb:trace}
\addsymbol \H^k(\Omega): {Sobolev space of class $k$ over the manifold $\Omega$\,, $k\in\mathbb{R}$}{symb:sobolevomega}
\addsymbol \H^k_0(\Omega): {Completion of the smooth functions with compact support in the Sobolev norm of order $k$ over the manifold \smash{$\Omega$}}{symb:sobolevomega0}
\addsymbol \H^k(\pO): {Sobolev space of class $k$ over the boundary manifold of $\Omega$}{symb:sobolevpo}
\addsymbol I: {Canonical isometric bijection of the scale of Hilbert spaces $\H_+\subset\H\subset\H_-$}{symb:I}
\addsymbol \Lambda^1(\Omega): {Space of smooth one-forms of the manifold $\Omega$}{symb:oneforms}
\addsymbol \eta: {Riemannian metric}{symb:metric}
\addsymbol \partial\eta: {Riemannian metric induced at the boundary}{symb:partialeta}
\addsymbol \norm{\cdot}: {Norm of the Hilbert space $\H$}{symb:norm}
\addsymbol \norm{\cdot}_k: {Sobolev norm of order $k$}{symb:normk}
\addsymbol \norm{\cdot}_{\H^k(\Omega)}: {Sobolev norm of order $k$}{symb:normkomega}
\addsymbol \normm{\cdot}_T: {Graph norm of the operator $T$}{symb:graphnormT}
\addsymbol \normm{\cdot}_Q: {Graph norm of the lower semi-bounded quadratic form $Q$}{symb:graphnormQ}
\addsymbol \mathcal{N}_{\pm}: {Deficiency spaces}{symb:deficiencyspace}
\addsymbol n_{\pm}: {Deficiency indices}{symb:deficiencyindex}
\addsymbol \Omega: {Smooth, orientable, compact Riemannian manifold}{symb:omega}
\addsymbol \partial \Omega: {Boundary manifold of the manifold $\Omega$}{symb:po}
\addsymbol \tilde{\Omega}: {Smooth Riemannian manifold without boundary}{symb:omeganotboundary}
\addsymbol \smash{\overset{\scriptscriptstyle\circ}{\Omega}}: {Interior of the manifold $\Omega$}{symb:intOmega}
\addsymbol \varphi,\,\psi: {Restriction to the boundary of the functions denoted by the corresponding capital greek letters $\Phi,\,\Psi$}{symb:varphi}
\addsymbol \dot{\varphi},\,\dot{\psi}: {Restriction to the boundary of the normal derivative of the functions denoted by the corresponding capital greek letters  $\Phi,\,\Psi$}{symb:dotvarphi}
\addsymbol P: {Orthogonal projection onto $W$}{symb:P}
\addsymbol Q_U: {Quadratic form associated to the Laplace-Beltrami operator described by the unitary operator $U$}{symb:QU}
\addsymbol \scalar{\cdot}{\cdot}: {Scalar product of the Hilbert space $\H$}{symb:scalar}
\addsymbol \pair{\cdot}{\cdot}: {Pairing associated to the scale of Hilbert spaces $\H_+\subset\H\subset\H_-$}{symb:pair}
\addsymbol \scalarb{\cdot}{\cdot}: {Induced scalar product at the boundary}{symb:scalarb}
\addsymbol \sigma(T): {Spectrum of the operator $T$}{symb:spectrum}
\addsymbol V(g): {Unitary representation of the element $g\in G$ of the group $G$}{symb:Vg}
\addsymbol \mathrm{v}(g): {Induced unitary representation at the boundary of the element $g\in G$ of the group G}{symb:v(g)}
\addsymbol W: {Invertibility boundary space}{symb:W}
\addsymbol \mathfrak{X}(\Omega): {Space of smooth vector fields of the manifold $\Omega$}{symb:vectorfield}
\addsymbol \Xi: {Collar neighbourhood of the boundary manifold $\pO$}{symb:collar}
\end{tabbing}

 \clearemptydoublepage

\bibliographystyle{JMPStyle2}
\addcontentsline{toc}{chapter}{References}
{\small \bibliography{Bibliografia-1-1-2.bib}}

\end{document}